\documentclass[reqno]{amsart}
\usepackage{amssymb}

\newtheorem{theorem}{Theorem}[section]
\newtheorem{proposition}[theorem]{Proposition}
\newtheorem{lemma}[theorem]{Lemma}
\newtheorem{corollary}[theorem]{Corollary}

\theoremstyle{remark}
\newtheorem{remark}[theorem]{Remark}

\numberwithin{equation}{section}

\begin{document}

\title[Orthogonality of Bethe Ansatz Eigenfunctions]
{Orthogonality of Bethe Ansatz eigenfunctions for the Laplacian  on a \\ hyperoctahedral Weyl alcove}

\author{J.F.  van Diejen}

\address{
Instituto de Matem\'atica y F\'{\i}sica, Universidad de Talca,
Casilla 747, Talca, Chile}

\email{diejen@inst-mat.utalca.cl}

\author{E. Emsiz}

\address{
Facultad de Matem\'aticas, Pontificia Universidad Cat\'olica de Chile,
Casilla 306, Correo 22, Santiago, Chile}
\email{eemsiz@mat.uc.cl}


\thanks{This work was supported in part by the {\em Fondo Nacional de Desarrollo
Cient\'{\i}fico y Tecnol\'ogico (FONDECYT)} Grants \# 1130226 and  \# 1141114.}

\date{February 2016}

\begin{abstract}
We prove the
orthogonality of the Bethe Ansatz eigenfunctions for the Laplacian on a hyperoctahedral Weyl alcove with repulsive homogeneous Robin boundary conditions at the walls.
To this end these eigenfunctions  are retrieved as the continuum limit of
an orthogonal basis of
algebraic Bethe Ansatz eigenfunctions for a finite $q$-boson system endowed with diagonal open-end boundary interactions. 
\end{abstract}

\maketitle



\section{Introduction}\label{sec1}
It is well-known that the Lieb-Liniger
spectral problem for $n$ bosonic particles on the circle with a pairwise delta potential interaction
\cite{lie-lin:exact,gau:bethe,kor-bog-ize:quantum,mat:many-body} is---in the center-of-mass frame---mathematically equivalent to that of the Laplacian
on a Weyl alcove. (This alcove arises as a fundamental domain for the configuration space 
with respect to the action of the symmetric group $S_n$ permuting the particles.) In this picture, the delta potential interaction gives rise to homogeneous Robin boundary conditions at the walls of the alcove. For $n=3$, the alcove in question
amounts to an
equilateral triangle and the corresponding spectral problem for the Laplacian has in such case a rich history going back to Lam\'e (who considered the special situations of Dirichlet and Neumann boundary conditions at the triangle's edges) \cite{das-fok:basic,mcc:laplacian}.

From the point of view of spectral- and harmonic analysis, a crucial issue regarding the eigen\-functions of the Laplacian---which were constructed by Lieb and Liniger with the aid of the coordinate Bethe Ansatz method via a suitable linear combination of plane waves so as to match the boundary conditions at the
walls---is the proof of their orthogonality and completeness in an appropriate Hilbert space setting. 
In the repulsive coupling parameter regime both questions were settled by Dorlas 
\cite{dor:orthogonality}, who combined a variational  method of Yang and Yang  describing the solutions of the Bethe equations characterizing the spectrum \cite{yan-yan:thermodynamics} with the observation that
the eigenfunctions under consideration can be retrieved via a continuum limit from the algebraic Bethe Ansatz wave functions for
 a lattice quantum nonlinear Schr\"odinger equation studied by Bogoliubov, Izergin and Korepin \cite[Ch. V.III]{kor-bog-ize:quantum}.
The crux is that the diagonalization of the transfer operator for the lattice model allows to remove possible degeneracies in the spectrum that prevent a direct orthogonality proof based  on the self-adjointness of the Laplacian alone. 
The completeness question is moreover dealt with by exploiting the continuity and monotonicity in the coupling parameter, which allows for a reduction to the elementary case of Neumann boundary conditions (where the completeness of the Bethe Ansatz eigenfunctions follows from the completeness of the standard Fourier basis).

Our main goal is to extend the above picture to the case of the harmonic analysis on a Weyl alcove associated  with 
the {\em hyperoctahedral group} $W:=S_n\ltimes  \{1 ,-1\}^n $ of {\em signed permutations} $(\sigma ,\epsilon)$:
$$
{x}=(x_1,\ldots ,x_n)\stackrel{(\sigma ,\epsilon)}{\longrightarrow} (\epsilon_1 x_{\sigma_1},\ldots ,\epsilon_n x_{\sigma_n}) ,
$$
where $\sigma = \left( \begin{matrix} 1& 2& \cdots & n \\
 \sigma_1&\sigma_2&\cdots & \sigma_n
 \end{matrix}\right) \in S_n$ and $\epsilon = (\epsilon_1,\ldots,\epsilon_n)$ with $\epsilon_j\in \{ 1,-1\}$ ($j=1,\ldots, n$).
The underlying particle model is governed by
a formal Schr\"odinger operator of the form
\cite{gau:boundary,gau:bethe,gut-sut:completely,gut:integrable,hec-opd:yang,die:plancherel,ems-opd-sto:periodic,ems-opd-sto:trigonometric,ems:completeness}:
\begin{equation}\label{SEQ}
-\Delta +\sum_{1\leq j< k\leq n}   c \Bigl(    \delta (x_j-x_k)   + \delta (x_j+x_k)  \Bigr)  
 +\sum_{1\leq j\leq n}  \Bigl( c_1 \delta (x_j)     + c_2\delta (2x_j ) \Bigr) ,
\end{equation}
where $\Delta$ represents the Laplacian $\partial_{x_1}^2+\cdots +\partial_{x_n}^2$ and $\delta (\cdot)$ refers to the \emph{periodic} delta distribution (or Dirac comb) with $\delta (x_j+m) =\delta(x_j)$ ($m\in\mathbb{Z}$), so
$\delta(2x_j)=\frac{1}{2}\bigl( \delta (x_j)+\delta (x_j+\frac{1}{2})\bigr)$.
The position variables $x_1,\ldots ,x_n$  are again assumed to lie on the circle $\mathbb{R}/\mathbb{Z}$ and the constants $c$, $c_1$ and $c_2$ are interpreted as coupling parameters regulating the strengths of the delta potential interactions.
Upon restriction to the space of $W$-invariant functions on the configuration space $(\mathbb{R}/\mathbb{Z})^n$,
the eigenvalue problem for the Schr\"odinger operator \eqref{SEQ}
becomes manifestly equivalent to the  following eigenvalue equation \cite{gut-sut:completely,gut:integrable,ems-opd-sto:periodic,ems-opd-sto:trigonometric}
\begin{subequations}
\begin{equation}\label{LEP}
-\Delta \psi = E \psi 
\end{equation}
for the Laplacian on the fundamental  hyperoctahedral Weyl alcove
\begin{equation}\label{A}
\text{A} =\{   {x}\in\mathbb{R}^n\mid  {\textstyle \frac{1}{2} }> x_1> x_2> \cdots > x_n>  0 \} ,
\end{equation}
endowed
with homogeneous Robin boundary conditions at the walls of the form
\begin{equation}\label{BC1}
\Bigl. (\partial_{x_j}-\partial_{x_{j+1}}-g)\psi \Bigr|_{x_j-x_{j+1}=0} =0\quad \text{for} \quad j=1,\ldots ,n-1
\end{equation}
and
\begin{equation}\label{BC2}
\Bigl. (\partial_{x_1}+g_+)\psi \Bigr|_{x_1=\frac{1}{2}} =0, \qquad   \Bigl. (\partial_{x_n}-g_-)\psi \Bigr|_{x_n=0} =0.
\end{equation}
\end{subequations}
Here $E$ represents the energy eigenvalue and the boundary parameters $g$, $g_+$, $g_-$ are related to the coupling parameters 
$c$, $c_1$, $c_2$ via $c=2g$, $c_1=2(g_--g_+) $ and $c_2= 4g_+$.

Below we will demonstrate the orthogonality of the Bethe Ansatz eigenfunctions for the spectral problem in Eqs. \eqref{LEP}--\eqref{BC2}, under the assumption that the boundary parameters $g$, $g_+$ and $g_-$ are positive (so our walls are repulsive).
This is achieved by showing that the eigenfunctions under consideration arise as the continuum limit of
the algebraic Bethe Ansatz eigenfunctions for a $q$-boson system \cite{kul:quantum,bog-bul:q-deformed,bog-ize-kit:correlation,sas-wad:exact,tsi:quantum,die:diagonalization,kor:cylindric,die-ems:diagonalization,bor-cor-pet-sas:spectral} built of a finite number of
$q$-oscillators \cite{maj:foundations,kli-sch:quantum} endowed with
diagonal open-end boundary interactions \cite{die:plancherel,li-wan:exact,die-ems:discrete,die-ems:semi-infinite,whe-zin:refined}. 
A key advantage of invoking this $q$-boson system (rather than a lattice quantum nonlinear Schr\"odinger equation with boundary interactions) is that the corresponding
 Bethe Ansatz eigenfunctions can be expressed explicitly in terms of Macdonald's hyperoctahedral Hall-Littlewood polynomials \cite{mac:orthogonal,nel-ram:kostka}, which greatly facilitates our analysis of the continuum limit. Indeed, analogous explicit expressions   for the Bethe Ansatz eigenfunctions of the finite $q$-boson system with {\em periodic} boundary conditions in terms of Hall-Littlewood polynomials \cite{bog-ize-kit:correlation,tsi:quantum,die:diagonalization,kor:cylindric}, were instrumental for a
 simpler alternative continuum limit reconfirming the orthogonality of the Lieb-Liniger eigenfunctions for the Weyl alcove associated with $S_n$ \cite{die:diagonalization} (in comparison with the more elaborate continuum limit employed in Dorlas' original approach based on
 the lattice quantum nonlinear Schr\"odinger equation).
Together, the orthogonality of the Bethe Ansatz eigenfunctions on the Weyl alcoves associated the symmetric group and the hyperoctahedral group
imply the orthogonality of a basis of Bethe Ansatz eigenfunctions for the corresponding Laplacian with repulsive Robin type boundary conditions on any Weyl alcove of classical type. Indeed, the cases of the Weyl alcoves of types $B$, $C$ and $D$ are recovered from the spectral problem in Eqs. \eqref{LEP}--\eqref{BC2} via specializations of the boundary parameters $g_+$ and $g_-$.

The material is organized as follows.  
 First our main result,
concerning the orthogonality of the Bethe Ansatz eigenfunctions for the eigenvalue problem
in Eqs. \eqref{LEP}--\eqref{BC2}, is formulated precisely in Section \ref{sec2}. Sections \ref{sec3}--\ref{sec12} are devoted to the proof of this orthogonality, which breaks up in several steps.  Section \ref{sec3} describes the transfer operator of a finite $q$-boson system with
diagonal open-end boundary interactions \cite{li-wan:exact,die-ems:semi-infinite}. Some key points of
Sklyanin's quantum inverse scattering formalism for open systems  \cite{skl:boundary} sustaining our discussion have been briefly summarized in Appendix \ref{appA} (at the end of the paper).
Section \ref{sec4} analyzes the main properties of the underlying $q$-boson boundary monodromy matrix in further detail.
These properties are subsequently exploited in Section \ref{sec5} to identify the $q$-boson Hamiltonian generated by the boundary transfer operator.
In Section \ref{sec6},  the Hamiltonian and the boundary transfer operator are represented  as (essentially) self-adjoint operators in the $q$-boson Fock space. The explicit action of the boundary transfer operator in Fock space is computed in Section \ref{sec7}.
Next, the eigenfunctions are constructed  in Section \ref{sec8} by means of Sklyanin's algebraic Bethe Ansatz formalism for open systems (which involves some technical computations supplemented in Appendix \ref{appB}).
In Section \ref{sec9} the Bethe Ansatz creation operator is employed to derive a branching rule that permits to compute the $n$-particle Bethe Ansatz wave function inductively in the number of particles/variables. 
After  providing the solutions of the Bethe Ansatz equations---via the minima of a family of strictly convex Morse functions following the method of Yang and Yang \cite{yan-yan:thermodynamics}---the self-adjointness of the boundary transfer operator is exploited in Section \ref{sec10} to infer that the corresponding Bethe Ansatz eigenfunctions constitute an orthogonal basis for the $q$-boson Fock space.
Next, in Section \ref{sec11}, we rely on the branching rule in combination with recent results drawn from Ref. \cite{whe-zin:refined} to express the $n$-particle Bethe-Ansatz eigenfunctions in terms of Macdonald's hyperoctahedral Hall-Littlewood polynomials of $BC_n$-type \cite[\S 10]{mac:orthogonal}.
This explicit coordinate representation of the $q$-boson eigenfunctions is then put to use in Section \ref{sec12}, where the orthogonality of the Bethe Ansatz eigenfunctions for  the Laplacian on the hyperoctahedral Weyl alcove is retrieved via a continuum limit ($q\to 1$). 
Finally, in Section \ref{sec13} some details are provided on how to derive the orthogonality of a basis of Bethe Ansatz eigenfunctions for the Laplacian with repulsive Robin boundary conditions on any Weyl alcove of classical type.


\section{Main result}\label{sec2}
For any $\xi=(\xi_1,\ldots,\xi_n)\in\mathbb{R}^n$ such that vanishing denominators are avoided, the following linear combination of plane waves

\begin{equation}\label{BA-WF}
\psi (\xi ,x)=   \sum_{\substack{ \sigma\in S_n \\ \epsilon\in \{ 1,-1\}^n}}   C(\epsilon_1 \xi_{\sigma_1},\ldots , \epsilon_n \xi_{\sigma_n})
\exp (i\epsilon_1 \xi_{\sigma_1}x_1+\cdots +i \epsilon_n \xi_{\sigma_n} x_n) ,
\end{equation}
with coefficients of the form
\begin{equation*}
C(\xi_1,\ldots, \xi_n) =  \prod_{1\leq j\leq n}    \left( \frac{\xi_j-ig_-}{\xi_j}      \right)
 \prod_{1\leq j<k\leq n} \left( \frac{\xi_j+\xi_k-ig}{\xi_j+\xi_k} \right) \left( \frac{\xi_j-\xi_k-ig}{\xi_j-\xi_k} \right) ,
\end{equation*}
provides a solution to the eigenvalue equation \eqref{LEP}, \eqref{A} that complies with the boundary conditions
\eqref{BC1}, \eqref{BC2} at those walls that pass through the origin \cite{gau:boundary,gut-sut:completely,gut:integrable,hec-opd:yang,die:plancherel}. The corresponding eigenvalue is given by
\begin{equation}
E= E(\xi):= \xi_1^2+\cdots +\xi_n^2.
\end{equation}
In order to satisfy also the boundary condition at the affine wall $x_1=\frac{1}{2}$, the spectral parameter $\xi$ must additionally obey the
Bethe Ansatz equations \cite{ems-opd-sto:periodic,bus-die-maz:norm,ems-opd-sto:trigonometric,ems:completeness}
\begin{equation}\label{BA-EQ}
e^{i\xi_j} =
\left(\frac{  ig_+ +\xi_j}{  ig_+ -\xi_j}\right) \left( \frac{  ig_- +\xi_j}{  ig_- -\xi_j} \right)
\prod_{\substack{ 1\leq k\leq n\\ k\neq j}}  \left(\frac{  ig +\xi_j+\xi_k}{  ig -\xi_j-\xi_k } \right) \left( \frac{  ig +\xi_j-\xi_k}{  ig -\xi_j+\xi_k }\right)
\end{equation}
for $j=1,\ldots ,n$.
Furthermore, following the paradigm instaurated by Yang and Yang \cite{yan-yan:thermodynamics}, 
a solution of the Bethe Ansatz equations \eqref{BA-EQ}  is
obtained for positive boundary parameters $g, g_+,g_-$ as the unique global minimum $\xi_\lambda$ of the following strictly convex Morse function  \cite{ems-opd-sto:periodic,bus-die-maz:norm,ems-opd-sto:trigonometric,ems:completeness}
\begin{align}\label{morse}
V_\lambda(\xi) =&\sum_{1\leq j\leq n}  \left( {\textstyle \frac{1}{2}}\xi_j^2   -2\pi  (\rho_j+\lambda_j)\xi_j
+2 \int_0^{\xi_j}  \arctan \Bigl(\frac{u}{g_+} \Bigr) +\arctan \Bigl(\frac{u}{g_-} \Bigr) \text{d}u \right) \nonumber \\
+ &\sum_{1\leq j<k\leq n} 2 \left( 
\int_0^{\xi_j+\xi_k} \arctan \Bigl(\frac{u}{g} \Bigr) \text{d}u +\int_0^{\xi_j-\xi_k} \arctan \Bigl(\frac{u}{g} \Bigr) \text{d}u
\right)
\end{align}
associated with a partition
\begin{equation}\label{partitions}
\lambda\in \Lambda =\{ (\lambda_1,\ldots ,\lambda_n)\in \mathbb{Z}^n  \mid \lambda_1\geq\lambda_2\geq \cdots \geq \lambda_n\geq 0\} ,
\end{equation}
where $\rho_j=n+1-j$ for $j=1,\ldots,n$.
Indeed, the critical equation $\nabla_\xi V_\lambda (\xi)=0$, i.e.
\begin{align}\label{critical}
\xi_j          &  + 2  \arctan \Bigl(\frac{\xi_j}{g_+} \Bigr) + 2  \arctan \Bigl(\frac{\xi_j}{g_-} \Bigr) \\
&+\sum_{\substack{ 1\leq k\leq n\\ k\neq j}} 
2\left( \arctan \Bigl(\frac{\xi_k+\xi_j}{g} \Bigr)  -\arctan \Bigl(\frac{\xi_k-\xi_j}{g} \Bigr) \right) 
=2\pi (\rho_j+\lambda_j ) \nonumber
\end{align}
($j=1,\ldots ,n$), goes over into the Bethe Ansatz equation \eqref{BA-EQ} upon exponentiation after multiplication by the imaginary unit, and the Hessian
$H_{j,k}=\partial_{\xi_j}\partial_{\xi_k} V_\lambda (\xi)$, i.e.
\begin{equation}
H_{j,k}=\begin{cases}
1+\frac{2g_+}{g_+^2+\xi_j^2}+\frac{2g_-}{g_-^2+\xi_j^2} +  \sum_{l\neq j} \left( \frac{2g}{g^2+(\xi_j+\xi_l)^2} + \frac{2g}{g^2+(\xi_j-\xi_l)^2}  \right) &\text{if}\ k = j\\
\frac{2g}{g^2+(\xi_j+\xi_k)^2} - \frac{2g}{g^2+(\xi_j-\xi_k)^2}&\text{if}\ k\neq j
\end{cases}
\end{equation}
($j,k=1,\ldots,n$), is positive definite. It moreover follows from Eq. \eqref{critical} that the minimum $\xi_\lambda$ of $V_\lambda (\xi)$ \eqref{morse} is assumed inside the Weyl chamber
\begin{equation}\label{chamber}
\text{C}=\{\xi\in \mathbb{R}^n \mid \xi_1 >\xi_2>\cdots >\xi_n>0\}  
\end{equation}
(cf. Remark \ref{minimum} below).

We are now in the position to formulate the main contribution of the present work.
\begin{theorem}[Orthogonality]\label{orthogonality:thm} For $g, g_+,g_->0$, the Bethe Ansatz eigenfunctions $\psi(\xi_\lambda ,x )$, $\lambda\in \Lambda$ constitute an orthogonal basis for the Hilbert space
$L^2(\text{A},\text{d}{x})$ of quadratically integrable functions with respect to the standard Lebesgue measure $\text{d}{x}$ on $A\subset \mathbb{R}^n$.
\end{theorem}
The principal issue to be settled here is the orthogonality of the Bethe Ansatz eigenfunctions, as the completeness of the system $\psi(\xi_\lambda ,x)$, $\lambda\in \Lambda$ in $L^2(\text{A},\text{d}{x})$ is readily deduced with the aid of the ideas in \cite[Sec. 3]{dor:orthogonality}   and \cite{ems:completeness}. Indeed, when $g_+=g_-$ the completeness follows directly from \cite[Thm. 1]{ems:completeness} (upon specialization to the root system $R=C_n$). The extension to the case that $g_+\neq g_-$ hinges in turn on
\cite[Thm. 3.2]{dor:orthogonality} and exploits that our Laplacian is  continuous and
monotonically nondecreasing (as a quadratic form) both in $g_+$ and in $g_-$, cf. \cite[Sec. 3]{dor:orthogonality}  and \cite{ems:completeness} for the precise details of this argument.

From the diagonalization of the Laplacian by the orthogonal basis in Theorem \ref{orthogonality:thm}, one deduces its spectral properties.

\begin{corollary}[Spectrum]\label{spectrum:cor}
 For $g, g_+,g_->0$, the Laplacian $-\Delta$ in Eqs. \eqref{LEP}--\eqref{BC2} constitutes an unbounded essentially self-adjoint operator in
 $L^2(\text{A},\text{d}{x})$  that has purely discrete spectrum given by  the positive eigenvalues
 $E(\xi_\lambda)$, $\lambda\in\Lambda$.
\end{corollary}

\begin{remark}\label{minimum}
Since the function $\arctan (\cdot)$ is odd and strictly monotonously increasing,
it is clear from Eq. \eqref{critical} that $\xi_j>0$ at the critical point of $V_\lambda (\xi)$.
Moreover, subtracting
the $k$th equation from the $j$th equation in Eq. \eqref{critical} reveals
similarly that at this critical point $\xi_j>\xi_k$ if $j<k$.  In other words, one has that
\begin{equation}
\xi_1>\xi_2>\cdots > \xi_n>0 \quad\text{at}\quad \xi=\xi_\lambda.
\end{equation}
Upon refining this analysis slightly with the aid of
the elementary inequality $0<\arctan (u)-\arctan(v)< u-v$ for $u> v$, one readily finds more precise estimates for $\xi_j$ and the moment gaps $\xi_j-\xi_k$ at the spectral value
$\xi=\xi_\lambda$:
\begin{subequations}
\begin{equation}
\frac{2\pi (\rho_j+\lambda_j)}{1+\kappa }<   \xi_j   <  2\pi (\rho_j+\lambda_j) \quad\text{at}\quad \xi=\xi_\lambda ,
\end{equation}
for $1\leq j\leq n$, and
\begin{equation}
\frac{2\pi (\rho_j-\rho_k+\lambda_j-\lambda_k)}{1+\kappa }<   \xi_j -  \xi_k <  2\pi (\rho_j-\rho_k+\lambda_j-\lambda_k) \quad\text{at}\quad \xi=\xi_\lambda ,
\end{equation}
\end{subequations}
for $1\leq j<k\leq n$,   where $\kappa :=2\bigl(g_+^{-1}+g_-^{-1}+2(n-1)g^{-1}\bigr)$.
The former inequalities confirm in particular that
the eigenvalues $E(\xi_\lambda)$ do not remain bounded from above when $\lambda$ varies over $\Lambda$ \eqref{partitions}, and also reveal the behavior of the spectrum in the  limit of strong coupling:
\begin{equation}
\lim_{g,g_+,g_-\to +\infty}   \xi_\lambda = 2\pi (\rho +\lambda ).
\end{equation}
\end{remark}


\section{Open-end $q$-Boson system}\label{sec3}
In this section we  describe the transfer operator of a finite $q$-boson system with diagonal open-end boundary interactions.

\subsection{$q$-Boson algebra} 

For a fixed  nonnegative integer $m$ and a deformation parameter $t:=q^2$  with $0<q<1$,
let us define the $q$-boson algebra $\mathbb{A}_m$ on the finite integral lattice
\begin{equation}\label{lattice}
\mathbb{N}_m:=\{ 0,1,2,\ldots, m \}
\end{equation}
as
the unital associative algebra over $\mathbb{C}$ with ultralocal generators $\beta_l$, $\beta^*_l$,  $t^{\pm N_l}$ ($l\in \mathbb{N}_m$) subject to the relations  \cite{kul:quantum,maj:foundations,kli-sch:quantum}
\begin{subequations}
\begin{equation}\label{e:qboson-a}
t^{N_l}\beta_l^*=t\beta^*_l t^{N_l}, \quad 
\beta_l t^{N_l} =t t^{N_l}\beta_l, \quad t^{N_l}t^{-N_l}=t^{-N_l}t^{N_l}=1
\end{equation}
and 
\begin{equation}\label{e:qboson-b}
\beta_l\beta_l^* - \beta_l^*\beta_l=t^{N_l}, \qquad 
\beta_l\beta_l^* - t\beta_l^*\beta_l =1.
\end{equation}
\end{subequations}
(By ultralocal it is meant that the generators corresponding to distinct sites $l\in\mathbb{N}_m$ commute.) It follows from the relations \eqref{e:qboson-a}, \eqref{e:qboson-b} that
\begin{subequations}
\begin{equation} \
\beta_l\beta_l^*=[N_l+1], \quad
\beta_l^*\beta_l = [N_l],\quad
\beta_l [N_l]=[N_l+1]\beta_l, \quad
\beta_l^* [N_l]=[N_l-1]\beta_l,
  \end{equation}
where 
\begin{equation}
[N_l+k]:=\frac{1-t^k t^{N_l}}{1-t}\qquad (k\in\mathbb{Z}).
\end{equation}
\end{subequations}

\subsection{Periodic transfer operator} The periodic $q$-boson system is a lattice regularization of the quantum nonlinear Schr\"odinger equation on the circle, which was introduced in
Refs. \cite{bog-bul:q-deformed,bog-ize-kit:correlation}
within the framework of the quantum inverse scattering formalism originating from Faddeev's celebrated St. Petersburg school
\cite{tak:integrable,kor-bog-ize:quantum,jim-miw:algebraic,fad:how}.

Let $M_k(\mathbb{A}_m)$ denote the algebra of square matrices of size $k$ with coefficients from $\mathbb{A}_m$ and
let $\mathbb{C}[u^{\pm 1}]$ be the algebra of complex Laurent polynomials in the spectral parameter $u$. For our purposes, a convenient system of generators
for $\mathbb{C}[u^{\pm 1}]$  is given by $s(u) := u-u^{-1}$ and $ c(u):=u+u^{-1}$.

The $q$-boson system is characterized by the following Lax matrix
\cite{bog-ize-kit:correlation}
\begin{equation}\label{L-op}
L_l(u)= 
\begin{pmatrix}
u^{-1} & (1-t)\beta_l^* \\
\beta_l & u
\end{pmatrix}
\in \mathbb{C}[u^{\pm 1}] \otimes M_2(\mathbb{A}_m)
\end{equation}
satisfying  the quantum Yang-Baxter equation (a.k.a. fundamental commutation relation)
\begin{equation}\label{e:L-YBE}
R(u/v) L_l(u)_1 L_l(v)_2 =  L_l(v)_1 L_l(u)_2 R(u/v)
\end{equation}
in $\mathbb{C}[u^{\pm 1},v^{\pm 1}]\otimes M_4(\mathbb{A}_m)$, where $$L_l(u)_1:= L_l(u)\otimes \begin{pmatrix}
1 & 0 \\
0& 1
\end{pmatrix} ,\quad L_l(u)_2:=  \begin{pmatrix}
1 & 0 \\
0& 1
\end{pmatrix} \otimes  L_l(u),$$ and
$R(u)\in \mathbb{C}[u^{\pm 1}]\otimes M_4 (\mathbb{C}) $ denotes the (trigonometric) $R$-matrix 
\begin{equation}\label{R-matrix}
R(u)=
\begin{pmatrix}
 s(q^{-1}u) & 0 & 0 & 0 \\
  0 &   s(q^{-1}) &   q^{-1} s(u) &0 \\
 0 &  q s(u) &  s(q^{-1}) & 0 \\
  0 & 0 & 0 &   s(q^{-1}u)
\end{pmatrix} .
\end{equation}

By taking the trace of the associated periodic monodromy matrix 
\begin{align}\label{U}
U_m(u) & :=  L_{m}(u) \cdots L_1(u)  L_0(u)\\
&=
\begin{pmatrix}
A_m(u) & B_m(u) \\
C_m(u) & D_m(u)
\end{pmatrix} 
\in \mathbb{C}[u^{\pm 1}]\otimes M_2(\mathbb{A}_m) \nonumber
\end{align}
solving the quantum Yang-Baxter equation \eqref{e:L-YBE} (with $U_m(u)$ substituted for $L_l(u)$), 
one arrives at the transfer operator for the periodic $q$-boson system
\begin{equation}\label{T}
T_m(u)=\text{tr} \,U_m(u)=A_m(u)+D_m(u) \in \mathbb{C}[u^{\pm 1}]\otimes \mathbb{A}_m ,
\end{equation}
satisfying the commutativity
\begin{equation}
[T_m(u),T_m(v)]=0 
\end{equation}
in $ \mathbb{C}[u^{\pm 1},v^{\pm 1}]\otimes  \mathbb A_m$
(where the bracket $[ \cdot ,\cdot ]$ refers to the commutator product).

\subsection{Boundary transfer operator} 
The quantum inverse scattering formalism was extended through a fundamental contribution of Sklyanin allowing for open-end boundary conditions (as opposed to periodic boundary conditions) \cite{skl:boundary}. 
Over the years the technical assumptions on the properties of the permitted $R$-matrices were relaxed somewhat, putting Sklyanin's formalism at a similar level of generality as the original quantum inverse scattering formalism for periodic systems \cite{mez-nep:integrable,fan-shi-hou-yan:integrable,vla:boundary}. 
The objectives of this paper only require to consider a finite $q$-boson system with diagonal boundary interactions at the lattice endpoints
\cite{li-wan:exact}.  For completeness, the main highlights of Sklyanin's formalism underlying our presentation have been briefly summarized
in Appendix \ref{appA}  with pointers to more detailed
discussions in the literature.

The diagonal boundary $K$-matrices in $ \mathbb{C}[u^{\pm 1}]\otimes M_2(\mathbb{C})$ of the form \cite{che:factorizing}
\begin{equation}\label{K-matrix}
K_- (u;a_-) = \begin{pmatrix} e(u;a_-) & 0\\ 0 & f(u;a_-) \end{pmatrix},\quad
K_+ (u;a_+) = \begin{pmatrix} f(u;a_+) & 0\\ 0 & e(u;a_+) \end{pmatrix}  
\end{equation}
with $e(u;a):=au -u^{-1}$ and $f(u;a):= e(q^{-1} u^{-1};a)$, satisfy respectively the
left reflection equation 
\begin{subequations}
\begin{align}\label{e:RE1}
R(u/v)  K_-(u;a_-)_1 R(quv)  & K_-(v;a_-)_1  \\
= &
  K_-(v;a_-)_1 R(quv)   K_-(u;a_-)_1 R(u/v) \nonumber
\end{align}
and the right reflection equation
\begin{align}\label{e:RE2}
R(u/v)  K_+(u;a_+)_2  R(q u v)  & K_+(v;a_+)_2 \\
= &K_+(v;a_+)_2 R(quv)   K_+(u;a_+)_2  R(u/v)) \nonumber
\end{align}
\end{subequations}
in $ \mathbb{C}[u^{\pm 1}]\otimes M_4(\mathbb{C})\subset    \mathbb{C}[u^{\pm 1}]\otimes M_4(\mathbb{A}_m)$. Here $a_-, a_+\in (-1,1)$ encode two coupling parameters that
govern the boundary interactions at the sites $l=0$ and  $l=m$, respectively.

From the  boundary monodromy matrix
\begin{align}\label{U-boundary}
\mathcal U_m(u;a_-):=& U_m(u) K_-(u;a_-) U_m^{-1}(q^{-1}u^{-1})\\
=& \begin{pmatrix} \mathcal A_m(u;a_-) &  \mathcal  B_m(u;a_-) \\  \mathcal  C_m(u;a_-) & \mathcal  D_m(u;a_-) \end{pmatrix}
\in  \mathbb{C}[u^{\pm 1}]\otimes M_2(\mathbb{A}_m) \nonumber
\end{align}
solving the left reflection equation \eqref{e:RE1} (with $ U_m(u;a_-)$ substituted for $K_-(u;a_-)$), one arrives at the boundary transfer operator
\begin{align}\label{T-boundary}
\mathcal T_m(u;a_+,a_-):=&\text{tr} \left( K_+(u^{-1};a_+)\, \mathcal U_m(u;a_-) \right) \\
=&f(u^{-1};a_+) \mathcal A_m(u; a_-) +  e(u^{-1};a_+)\mathcal D_m(u;a_-)  \nonumber
\end{align}
in $ \mathbb{C}[u^{\pm 1}]\otimes  \mathbb A_m$, which
satisfies the commutativity
\begin{equation}\label{com-T-boundary}
[\mathcal T_m(u;a_+,a_-),\mathcal T_m(v;a_+,a_-)]=0 
\end{equation}
(in $ \mathbb{C}[u^{\pm 1},v^{\pm 1}]\otimes  \mathbb A_m$).

\begin{remark}\label{L-inverse:rem}
To confirm that the inverse $U_m^{-1}(u)$  of $U_m(u)$ \eqref{U} and thus $\mathcal{U}_m(u;a_-)$ \eqref{U-boundary} are indeed well-defined  matrices in
$ \mathbb{C}[u^{\pm 1}]\otimes  M_2(\mathbb A_m)$, it suffices to observe that following matrix
in $ \mathbb{C}[u^{\pm 1}] \otimes M_2(\mathbb{A}_m)$
\begin{equation*}
 \begin{pmatrix}
u  & (1-t^{-1})\beta_l^* \\
-\beta_l & u^{-1}t^{-1}
\end{pmatrix} t^{-N_l}=
 {\begin{pmatrix}
1 & 0 \\
0&-q
\end{pmatrix} }L_l(q^{-1}u^{-1})\begin{pmatrix}
1 & 0 \\
0&-q
\end{pmatrix}^{-1} q^{-1} t^{-N_l}
\end{equation*}
constitutes an inverse for the Lax matrix $L_l(u)$ \eqref{L-op}.
\end{remark}


\section{Monodromy matrix}\label{sec4}
This section isolates some salient features of the $q$-boson monodromy matrices $U_m(u)$ \eqref{U} and $\mathcal{U}_m(u;a)$ \eqref{U-boundary}. (The reader might wish to skip this material at first reading and
refer back to it when needed.)

\subsection{Periodic monodromy matrix}
Let the adjoint $\ast : \mathbb{A}_m\to\mathbb{A}_m$ be the antilinear anti-involution $a\stackrel{\ast}{\to} a^\ast$ ($a\in\mathbb{A}_m$) induced by the following action on the generators
\begin{subequations}
\begin{equation}
\beta_l\stackrel{\ast}{\to}\beta_l^\ast, \qquad \beta^\ast\stackrel{\ast}{\to}\beta_l ,\qquad t^{\pm N_l}\stackrel{\ast}{\to}t^{\pm N_l},
\end{equation}
and let  $r: \mathbb{A}_m\to\mathbb{A}_m$ denote the linear site-reversal involution $a\stackrel{r}{\to} a^{r}$ induced by the action
\begin{equation}
\beta_l\stackrel{r}{\to}\beta_{m-l}, \qquad \beta_l^\ast\stackrel{r}{\to}\beta_{m-l}^\ast ,\qquad t^{\pm N_l}\stackrel{r}{\to}t^{\pm N_{m-l}}
\end{equation}
\end{subequations}
on the generators (a.k.a. space inversion). Both operations on $\mathbb{A}_m$ are extended naturally to $\mathbb{C}[u^{\pm 1}]\otimes\mathbb{A}_m$ through the prescription $u^*:=u=:u^r$ (so in particular, we treat our spectral parameter $u$ here momentarily as a real-valued indeterminate). 
A further extension to matrices $M(u)\in \mathbb{C}[u^{\pm 1}]\otimes M_k(\mathbb{A}_m)$ is achieved 
element-wise, applied to the transposed matrix $M^t(u)$ in the case of the adjoint and to the matrix $M(u)$ itself in the case of the site-reversal involution.

\begin{lemma}\label{U-adjoint-reverse:lem}
The adjoint $U_m^*(u)=L_0^\ast (u)L_1(u)\cdots L_m^\ast (u)$ and the reverse $U_m^r(u)=L_0(u)L_1(u)\cdots L_m(u)$ of the periodic monodromy matrix $U_m(u)$ \eqref{U} are given by
\begin{subequations}
\begin{equation}\label{U-adjoint}
U_m^*(u)=
\begin{pmatrix}
A^*_m(u) & C^*_m(u) \\
B^*_m(u) & D^*_m(u)
\end{pmatrix}
=
\begin{pmatrix}
D_m(u^{-1}) & (1-t)^{-1} B_m(u^{-1}) \\
(1-t) C_m(u^{-1}) & A_m(u^{-1})
\end{pmatrix}
\end{equation}
and
\begin{equation}\label{U-reverse}
U_m^r(u)=
\begin{pmatrix}
A^r_m(u) & B^r_m(u) \\
C^r_m(u) & D^r_m(u)
\end{pmatrix}
=
\begin{pmatrix}
D_m(u^{-1}) &  B_m(u^{-1}) \\
C_m(u^{-1}) & A_m(u^{-1})
\end{pmatrix} .
\end{equation}
\end{subequations}
\end{lemma}
\begin{proof}
We employ induction in $m$ starting from $m=0$, in which case the lemma is immediate
from the explicit expression for the Lax matrix in Eq. \eqref{L-op}.
For $m>0$, the monodromy matrix $U_m(u)$ \eqref{U} satisfies the recurrence
\begin{align}\label{e:Am+1}
&  \begin{pmatrix}
A_{m}(u) & B_{m}(u) \\
C_{m}(u) & D_{m}(u) 
  \end{pmatrix}
=
\begin{pmatrix}
u^{-1} & (1-t)  \beta_{m}^* \\
 \beta_{m} & u
\end{pmatrix}
  \begin{pmatrix}
A_{m-1}(u) & B_{m-1}(u) \\
C_{m-1}(u) & D_{m-1}(u) 
  \end{pmatrix}\\
 & =
  \begin{pmatrix}
u^{-1} A_{m-1}(u) + (1-t) \beta_{m}^* C_{m-1}(u) & u^{-1} B_{m-1}(u)+ (1-t) \beta_{m}^*  D_{m-1}(u) \\
 \beta_{m} A_{m-1}(u) + u C_{m-1}(u)  & \beta_{m} B_{m-1}(u) + u D_{m-1}(u)
  \end{pmatrix} . \nonumber
  \end{align}
 To verify the first part of the lemma, we
take the adjoint on both sides of Eq. \eqref{e:Am+1}, invoke the induction hypothesis to rewrite $A^\ast_{m-1}(u)$, $B^\ast_{m-1}(u)$, $C^\ast_{m-1}(u)$ and $D^\ast_{m-1}(u)$ in terms of
$A_{m-1}(u^{-1})$, $B_{m-1}(u^{-1})$, $C_{m-1}(u^{-1})$ and $D_{m-1}(u^{-1})$, and commute on the RHS all generators corresponding site $m$ to the left of the matrix elements
(where one uses that generators corresponding to distinct sites commute in $\mathbb{C}[u^\pm]\otimes\mathbb{A}_m$ in view of the ultralocality). 
Eq. \eqref{U-adjoint} now follows upon comparing with the recurrence originating from Eq. \eqref{e:Am+1} with  $u$ replaced by $u^{-1}$. The verification of the second part is similar and
hinges on a comparison of the recurrence
\begin{align*}
&  \begin{pmatrix}
A^r_{m}(u) & B^r_{m}(u) \\
C^r_{m}(u) & D^r_{m}(u) 
  \end{pmatrix}
=
  \begin{pmatrix}
A^{r}_{m-1}(u) & B^r_{m-1}(u) \\
 C^{r}_{m-1}(u) & D^r_{m-1}(u) 
  \end{pmatrix}
  \begin{pmatrix}
u^{-1} & (1-t)  \beta_{m}^* \\
 \beta_{m} & u
\end{pmatrix}
  \\
 & =
  \begin{pmatrix}
u^{-1} A^r_{m-1}(u) +\beta_{m} B^r_{m-1}(u) &  (1-t) \beta_{m}^*  A^r_{m-1}(u) + u  B^r_{m-1}(u)
 \\
u^{-1} C^r_{m-1}(u) +\beta_{m} D^r_{m-1}(u)  & (1-t) \beta_{m}^*  C^r_{m-1}(u) + u  D^r_{m-1}(u)
  \end{pmatrix}  \nonumber
  \end{align*}
with that of  Eq. \eqref{e:Am+1}
with  $u$ replaced by $u^{-1}$.  (Here we mean by $A^{r}_{m-1}(u)$, $B^r_{m-1}(u) $, 
 $C^{r}_{m-1}(u)$ and $D^r_{m-1}(u) $ the site-reversals of these matrix elements in $\mathbb{C}[u^{\pm 1}]\otimes \mathbb{A}_{m-1}$.)
\end{proof}

By combining the above formulas for  $U^\ast_m(u)$ and $U_m^r(u)$  with the inverse of the Lax matrix in Remark \ref{L-inverse:rem}, one deduces a compact expression for the inverse of the periodic monodromy matrix.
\begin{lemma}\label{U-inverse:lem}
The inverse of the periodic monodromy matrix $U_m(\cdot )$ \eqref{U} is given by
\begin{subequations}
\begin{equation}\label{U-inverse}
U^{-1}_m(q^{-1}u^{-1})= 
\begin{pmatrix}
A^*_m(u) & s(q)C^*_m(u) \\
s(q)^{-1}B^*_m(u) & D^*_m(u)
\end{pmatrix}\mathcal{N}_m^{-1} ,
\end{equation}
where $\mathcal{N}_m$ encodes a $q$-deformed number operator defined by
\begin{equation}\label{Nm}
\mathcal{N}_m:=\prod_{0\leq l\leq m} q t^{N_l} .
\end{equation}
\end{subequations}

\end{lemma}
\begin{proof}
From Remark \ref{L-inverse:rem} it is immediate that
$$
U^{-1}_m(q^{-1}u^{-1})=
{\begin{pmatrix}
1 & 0 \\
0&-q
\end{pmatrix} }U_m^r(u) \begin{pmatrix}
1 & 0 \\
0&-q
\end{pmatrix}^{-1} \mathcal{N}_m^{-1},
$$
whence the lemma follows upon
making $U_m^r(u)$ explicit by means of Lemma \ref{U-adjoint-reverse:lem}.
\end{proof}

Finally, the particle-conservation of the periodic $q$-boson system reflects itself in the following commutation relations between the $q$-deformed number operator and the matrix elements of the periodic monodromy matrix $U_m(u)$ \eqref{U}.
\begin{lemma}\label{U-N-com:lem}
One has that 
\begin{equation}
\mathcal{N}_m
\begin{pmatrix}
A_{m}(u) & B_{m}(u) \\
C_{m}(u) & D_{m}(u) 
  \end{pmatrix}
=
\begin{pmatrix}
A_{m}(u) & tB_{m}(u) \\
t^{-1} C_{m}(u) & D_{m}(u) 
  \end{pmatrix}
\mathcal{N}_m
\end{equation}
in $\mathbb{C}[u^{\pm 1}]\otimes M_2(\mathbb{A}_m)$.
\end{lemma}
\begin{proof}
By induction in $m$ using Eq. \eqref{e:Am+1} and the ultralocal $q$-boson commutation relations in Eq. \eqref{e:qboson-a}.
\end{proof}

\subsection{Boundary monodromy matrix}
Upon plugging the inverse of Lemma \ref{U-inverse:lem} into the definition of the boundary monodromy matrix $\mathcal U_m(u;a) $ \eqref{U-boundary}:
\begin{align}
&\mathcal U_m(u;a) = \\
&\begin{pmatrix}
A_{m}(u) & B_{m}(u) \\
C_{m}(u) & D_{m}(u) 
  \end{pmatrix}  \begin{pmatrix} e(u;a) & 0\\ 0 & f(u;a) \end{pmatrix}\begin{pmatrix}
A^*_m(u) & s(q)C^*_m(u) \\
s(q)^{-1}B^*_m(u) & D^*_m(u)
\end{pmatrix}\mathcal{N}_m^{-1} ,\nonumber
\end{align}
the following expressions for its matrix elements become evident (cf. also Lemma \ref{U-adjoint-reverse:lem}).

\begin{lemma}\label{U-boundary-matrix-elements:lem}
The matrix elements of the boundary monodromy matrix $\mathcal U_m(u;a) $ \eqref{U-boundary}
are given by
\begin{equation*}
\begin{split}
\mathcal A_m(u;a)&=\Bigl( e(u;a) A_m(u) A^*_m(u)     + s(q)^{-1} f(u;a) B_m(u) B^*_m(u)\Bigr) \mathcal N_m^{-1}    \\
&= \Bigl( e(u;a) A_m(u) D_m(u^{-1})     - q  f(u;a)  B_m(u) C_m(u^{-1})\Bigr)  \mathcal N_m^{-1} ,\\
\mathcal B_m(u;a)&=  \Bigl( s(q) e(u;a)  A_m(u) C^*_m(u) +  f(u;a) B_m(u) D^*_m(u)\Bigr) \mathcal N_m^{-1} \\
&=  \Bigl( -q^{-1} e(u;a)  A_m(u) B_m(u^{-1}) +    f(u;a) B_m(u) A_m(u^{-1})\Bigr) \mathcal N_m^{-1} ,\\
 \mathcal C_m(u;a)&= \Bigl( e(u;a)  C_m(u) A^*_m(u) +  s(q)^{-1} f(u;a) D_m(u) B^*_m(u)\Bigr) \mathcal N_m^{-1} \\
 &= \Bigl( e(u;a)  C_m(u) D_m(u^{-1}) -q   f(u;a) D_m(u) C_m(u^{-1}) \Bigr)\mathcal N_m^{-1} ,\\
 \mathcal D_m(u;a)&= \Bigl( s(q) e(u;a)  C_m(u) C^*_m(u) +  f(u;a) D_m(u) D^*_m(u) \Bigr) \mathcal N_m^{-1} \\
&= \Bigl( -q^{-1}  e(u;a)  C_m(u) B_m(u^{-1}) +    f(u;a) D_m(u) A_m(u^{-1}) \Bigr)\mathcal N_m^{-1} .
\end{split}
\end{equation*}
\end{lemma}

With the aid of these expressions, we confirm the commutation relations
between the $q$-deformed number operator $\mathcal{N}_m$ \eqref{Nm} and the matrix elements of the boundary monodromy matrix $\mathcal U_m(u)$ \eqref{U-boundary}
(cf. Lemma \ref{U-N-com:lem}).
\begin{lemma}\label{U-boundary-N-com:lem}
One has that 
\begin{equation}
\mathcal{N}_m
\begin{pmatrix}
\mathcal A_{m}(u;a) & \mathcal B_{m}(u;a) \\
\mathcal C_{m}(u;a) & \mathcal D_{m}(u;a) 
  \end{pmatrix}
=
\begin{pmatrix}
\mathcal A_{m}(u;a) & t\mathcal B_{m}(u;a) \\
t^{-1} \mathcal C_{m}(u;a) & \mathcal D_{m}(u;a) 
  \end{pmatrix}
\mathcal{N}_m
\end{equation}
in $\mathbb{C}[u^{\pm 1}]\otimes M_2(\mathbb{A}_m)$.
\end{lemma}
\begin{proof}
Immediate from Lemmas \ref{U-N-com:lem} and \ref{U-boundary-matrix-elements:lem}.
\end{proof}

Finally, we compute the adjoint of the boundary monodromy matrix in terms of previously known matrix elements (cf. Lemma \ref{U-adjoint-reverse:lem}).

\begin{lemma}\label{U-boundary-adjoint:lem}
The adjoint of the boundary monodromy matrix $\mathcal U_m(u;a)$ \eqref{U-boundary} reads
\begin{align}\label{U-boundary-adjoint}
\mathcal U_m^*(u;a) &= \\
&\begin{pmatrix}
\mathcal A^*_m(u;a) & C^*_m(u;a) \\
B^*_m(u;a) & D^*_m(u;a)
\end{pmatrix}
 =
\begin{pmatrix}
\mathcal A_m(u;a) & s(q)^{-1}  t^{-1}\mathcal B_m(u;a) \\
s(q) t \mathcal  C_m(u;a) & \mathcal D_m(u;a)
\end{pmatrix} .\nonumber
\end{align}
\end{lemma}
\begin{proof}
Immediate from Lemmas \ref{U-boundary-matrix-elements:lem} and \ref{U-boundary-N-com:lem}.
\end{proof}


\section{Hamiltonian}\label{sec5}
In this section, we identify the $q$-boson Hamiltonian generated by our boundary transfer operator $\mathcal T_m(u;a_+,a_-)$ \eqref{T-boundary}.
It corresponds to an
open $q$-boson system on $\mathbb{N}_m$ \eqref{lattice} with integrable boundary interactions of a type  considered in Ref. \cite{li-wan:exact} (on the finite lattice) and more generally in Ref. \cite{die-ems:semi-infinite} (on the semi-infinite lattice).

\subsection{Particle conservation, integrability and hermiticity}
The commuting quantum integrals generated by
our boundary transfer operator
$\mathcal{T}_m(u;a_+,a_-)$ \eqref{T-boundary} are retrieved from its power expansion in the spectral parameter $u$. The simplest quantum integral corresponds to the
 $q$-deformed number operator $\mathcal{N}_m$ \eqref{Nm}.
 \begin{proposition}[Particle Conservation]\label{particle-conservation:prp}
The $q$-deformed number operator $\mathcal{N}_m=\prod_{l=0}^m  qt^{N_l}$ commutes with the transfer operators $T_m(u)$ \eqref{T}
and $\mathcal T_m(u;a_+,a_-)$ \eqref{T-boundary}:
\begin{equation}
   [\mathcal{N}_m,T_m(u)]=0\quad \text{and} \quad  [\mathcal{N}_m,\mathcal T_m(u;a_+,a_-)]=0
\end{equation}
in $\mathbb{C}[u^{\pm 1}]\otimes \mathbb{A}_m$.
\end{proposition}
\begin{proof}
Immediate from Lemmas \ref{U-N-com:lem} and \ref{U-boundary-N-com:lem}.
\end{proof}
The next term in the expansion gives rise to the following $q$-boson Hamiltonian
\begin{equation}\label{Hm}
\mathcal{H}_m:=a_-[N_0]+a_+[N_{m}]
+\sum_{ 0\leq l < m } \left( \beta_l\beta_{l+1}^*+\beta^*_l\beta_{l+1} \right)  .
\end{equation}
Notice in this connection that it is already manifest from the ultralocal $q$-boson commutation relations \eqref{e:qboson-a} that $\mathcal{H}_m$ \eqref{Hm} commutes with  $\mathcal{N}_m$ \eqref{Nm} in $\mathbb{A}_m$.

\begin{proposition}[Quantum Integrability]\label{Hamiltonian:prp}
The $q$-boson Hamiltonian $\mathcal{H}_m$ \eqref{Hm} commutes with the boundary transfer operator
$\mathcal{T}_m (u;a_+,a_-)$ \eqref{T-boundary} in $\mathbb{C}[u^{\pm 1}]\otimes\mathbb{A}_m$:
\begin{equation}
 [\mathcal{H}_m,\mathcal{T}_m(u;a_+,a_-)]=0 .
\end{equation}
\end{proposition}

Before entering the proof of this proposition (below), let us observe that it is
clear from the explicit representations in terms of the generators that the  $q$-deformed number operator $\mathcal{N}_m$ \eqref{Nm} and the Hamiltonian $\mathcal{H}_m$ \eqref{Hm} are Hermitian elements in $\mathbb{A}_m$, i.e. $\mathcal{N}_m^\ast=\mathcal{N}_m$ and  $\mathcal{H}_m^\ast=\mathcal{H}_m$. This hermiticity turns out to extend to all commuting quantum integrals generated by the  boundary transfer operator.
\begin{proposition}[Hermiticity] \label{hermitian:prp} The boundary transfer operator $\mathcal{T}_m(u;a_+,a_-)$ \eqref{T-boundary} constitutes a Hermitian Laurent polynomial in $\mathbb{C}[u^{\pm 1}]\otimes\mathbb{A}_m$:
\begin{equation}
\mathcal{T}_m^\ast (u;a_+,a_-)=\mathcal{T}_m(u;a_+,a_-).
\end{equation}
\end{proposition}
\begin{proof}
Immediate from Eq. \eqref{T-boundary}  and Lemma \ref{U-boundary-adjoint:lem}.
\end{proof}

\subsection{Proof of Proposition \ref{Hamiltonian:prp}}
Our proof consists  of showing that the boundary transfer operator is a Laurent polynomial in $u$ of the form
\begin{subequations}
\begin{equation}\label{T-boundary-exp}
\mathcal{T}_m(u;a_+,a_-)=\sum_{k=-m-2}^{m+2}   \tau_{m,k} (a_+,a_-) u^{-2k}\qquad ( \tau_{m,k} (a_+,a_-)\in \mathbb{A}_m),
\end{equation}
with 
\begin{equation}\label{lead-coeff}
 \tau_{m,k} (a_+,a_-)= \begin{cases} q\mathcal{N}_m^{-1}&\text{if}\ k=m+2, \\
 \bigl( (1-t) \mathcal{H}_m -a_+ - a_-\bigr)q \mathcal{N}_m^{-1} &\text{if}\ k=m+1 .
 \end{cases}
\end{equation}
\end{subequations}
Because the expansion coefficients $\tau_{m,k}(a_+,a_-)$ ($k=-m-2,\ldots,m+2$) commute with $\mathcal{T}_m(u;a_+,a_-)$ in view of Eq. \eqref{com-T-boundary}, this confirms that both  $\mathcal{N}_m$ and  $\mathcal{H}_m$ commute with $\mathcal{T}_m(u;a_+,a_-)$
(and among themselves).

To verify Eqs.  \eqref{T-boundary-exp}, \eqref{lead-coeff}, we first compute the leading terms of the matrix elements of the periodic monodromy matrix $U_m(u)$  \eqref{U} (cf. also the proof of \cite[Prp. 3.10]{kor:cylindric} for the full power expansions of $u^{m+1}A_m(u)$ and $u^{m+1}D_m(u)$).
\begin{lemma}\label{T-exp:lem}
The matrix elements of $u^{m+1}U_m(u)$ are polynomials in the spectral parameter $u$ of the form:
\begin{align*}
u^{m+1}A_m(u)&= 1+u^2(1-t)\sum_{0\leq l< m} \beta_l\beta_{l+1}^\ast +r_{m,A}(u) ,\\
u^{m+1}B_m(u)&=u(1-t)\beta_0^* +r_{m,B}(u) ,\\
u^{m+1}C_m(u)&= u\beta_m               +r_{m,C}(u) ,\\
u^{m+1}D_m(u)&=   u^2(1-t)\beta_m\beta_0^*      +r_{m,D}(u) ,
\end{align*}
with $r_{m,A}(u),r_{m,D}(u)\in u^4\mathbb{C}[u^2]\otimes\mathbb{A}_m$ of degree $2m$ and $2m+2$, respectively, and
 $r_{m,B}(u),r_{m,C}(u)\in u^3\mathbb{C}[u^2]\otimes\mathbb{A}_m$ of degree $2m+1$.
\end{lemma}
\begin{proof}
For $m=0$ this is clear from the Lax matrix $L_l(u)$ \eqref{L-op}. The general case then follows inductively in $m$ via Eq. \eqref{e:Am+1}.
\end{proof}
Armed with these formulas, it is not difficult to infer that
\begin{multline*}
f(u^{-1};a_+)e(u;a_-) A_m(u) A_m^*(u)
= \\
q u^{-2m-4}  + 
q  u^{-2m-2} \Bigl(
 (1-t)
 \sum_{0\leq l < m} (\beta_l \beta_{l+1}^* + \beta_l^* \beta_{l+1})
- a_- - t^{-1}a_+\Bigr) 
+O(u^{-2m}),
\end{multline*}
\begin{equation*}
s(q)^{-1} f(u^{-1};a_+)f(u;a_-)  B_m(u) B_m^*(u)
= q u^{-2m-2} a_- (1-t)  [N_0] + O(u^{-2m}) ,
\end{equation*}
\begin{align*}
s(q)  e(u^{-1};a_+)e(u;a_-)& C_m(u) C_m^*(u)
= \\
& q u^{-2m-2} a_+ \left( (1-t) [N_m] - 1 +t^{-1}\right) 
+ O(u^{-2m}) 
\end{align*}
and
\begin{equation*}
e(u^{-1};a_+)f(u;a_-)  D_m(u) D_m^*(u) = O(u^{-2m}),
\end{equation*}
where $O(u^{-2m})$ stands for certain  higher-degree
terms in the expansions corresponding to the powers $u^{2k}$ for $k=-m,-m+1,\ldots ,m+2$ (at most).
When substituting these expansions in the boundary transfer operator
\begin{align*}
\mathcal T_m(u;a_+,a_-) & \\
\stackrel{Eq. \eqref{T-boundary}}{=}&f(u^{-1};a_+) \mathcal A_m(u; a_-) +  e(u^{-1};a_+)\mathcal D_m(u;a_-)  \\
 \stackrel{Lem. \ref{U-boundary-matrix-elements:lem}}{=}&f(u^{-1};a_+)\Bigl( e(u;a_-) A_m(u) A^*_m(u)     + s(q)^{-1} f(u;a_-)  B_m(u) B^*_m(u)\Bigr) \mathcal N_m^{-1} \\ & +   e(u^{-1};a_+) \Bigl( s(q) e(u;a_-)  C_m(u) C^*_m(u) +  f(u;a_-) D_m(u) D^*_m(u) \Bigr) \mathcal N_m^{-1} ,
\end{align*}
we end up with Eqs. \eqref{T-boundary-exp}, \eqref{lead-coeff}.


\section{Fock space representation}\label{sec6}
The  algebra $\mathbb{A}_m$ will now be represented unitarily on the $q$-boson Fock space. In this representation, the actions of the boundary transfer operator $\mathcal{T}_m(u;a_+,a_-)$ \eqref{T-boundary}  and the Hamiltonian $\mathcal{H}_m$ \eqref{Hm}   give rise to (essentially) self-adjoint operators in Fock space. Let us recall/emphasize in this connection that throughout
\begin{equation}\label{parameter-domain}
t=q^2,\quad 0<q<1 \quad\text{and}\quad -1<a_+,a_- <1 .
\end{equation}

\subsection{$q$-Boson Fock space}
For a positive integer $n$, we think of
\begin{equation}\label{dominant}
\Lambda_{n,m}:=\{(\lambda_1,\dots, \lambda_n)\in\mathbb{Z}^n \mid m\geq \lambda_1\geq \lambda_2 \geq \cdots \geq \lambda_n\ge 0\} 
\end{equation}
as a $q$-boson configuration space with the parts of $\lambda=(\lambda_1,\ldots,\lambda_n)\in\Lambda_{n,m}$ encoding the positions of a configuration of $n$ particles ($q$-bosons) placed on the finite lattice $\mathbb{N}_m$ \eqref{lattice}.
Let us define
\begin{subequations}
\begin{equation}\label{n-particle-sector}
\mathcal F_{n,m} :=\ell^2(\Lambda_{n,m},\delta_{n,m})
\end{equation}
 as the ($n$-particle) Hilbert space of complex functions  $f:\Lambda_{n,m}\to\mathbb{C}$, endowed with the inner product
\begin{equation}\label{inner-product}
(f,g)_{n,m} :=
\sum_{\lambda\in \Lambda_{n,m}} f(\lambda)\overline{g(\lambda)}\delta_{n,m}(\lambda) 
\end{equation}
determined by the weight function
\begin{equation}\label{weight-function}
\delta_{n,m}(\lambda)
:=\frac{1}{\prod_{l\in\mathbb{N}_m}[\text{m}_l(\lambda)]!}.
\end{equation}
\end{subequations}
Here $\text{m}_l(\lambda)$ is given by the number of parts $\lambda_j$ ($1\leq j\leq n$) such that $\lambda_j=l$, i.e.
$\text{m}_l(\lambda)$ counts the multiplicity of (the number of particles at site) $l$ in (the configuration) $\lambda$, and we have also employed  $q$-deformed integers and $q$-deformed factorials of the form
\begin{equation*}
[k]:=\frac{1-t^k}{1-t}\quad \text{and} \quad [k]! :=[k][k-1]\dots [2][1]=
\prod_{1\leq j\leq k}\frac{1-t^j}{1-t}
\end{equation*}
for $k$ nonnegative integral (with the convention that $[0]!:=1$). Notice that our Hilbert space $\mathcal F_{n,m}$ \eqref{n-particle-sector}--\eqref{weight-function} is actually finite-dimensional:
\begin{equation}\label{dimension-Fnm}
\dim (\mathcal F_{n,m} ) = \frac{(n+m)!}{n! \, m!}
\end{equation}
($=$ the cardinality of $\Lambda_{n,m}$ \eqref{dominant}).

By collecting the $n$-particle Hilbert spaces $\mathcal{F}_{n,m}$ for $n=0,1,2,\ldots$ (with the convention that  $\Lambda_{0,m}:=\{ \emptyset \}$
and $\delta_{0,m}(\emptyset):=1$,  so
$\mathcal{F}_{0,m}\cong \mathbb{C}$) via a direct orthogonal sum, one ends up with the following
(infinite-dimensional) $q$-boson Fock space:
\begin{subequations}
\begin{equation}\label{q-boson-fock-space}
\mathcal F_m:=\bigoplus_{n\geq 0} \mathcal F_{n,m},
\end{equation}
consisting of all series $F=\sum_{n\geq 0}  f_n$ of $f_n\in \mathcal F_{n,m}$ such that
\begin{equation}\label{fock-ip}
(F,F)_m :=\sum_{n\geq 0}  (f_n,f_n)_{n,m}<\infty .
\end{equation}
\end{subequations}

\subsection{Unitary representation of the $q$-boson algebra}
For $\lambda\in \Lambda_{n,m}$ and $l\in \mathbb{N}_m$, let $\beta^*_l\lambda\in \Lambda_{n+1,m}$ be obtained from $\lambda$ by inserting an additional part of size $l$ and---assuming   $\text{m}_l(\lambda)>0$---let $\beta_l\lambda\in \Lambda_{n-1,m}$ be the result of the inverse operation removing a part of size  $l$ from $\lambda$.

We are now in the position to represent the $q$-boson algebra $\mathbb{A}_m$ on $\mathcal F_m$  \eqref{q-boson-fock-space}, \eqref{fock-ip} by specifying the action of the generators. 
For any $f\in\mathcal{F}_{n,m}$ and $l\in\mathbb{N}_m$, the functions $\beta_l f\in \mathcal{F}_{n-1,m}$,  $\beta_l^\ast f\in \mathcal{F}_{n+1,m}$
and $t^{\pm N_l} f\in \mathcal{F}_{n,m}$ are defined as:
\begin{subequations}
\begin{equation}\label{qboson-repa}
   (\beta_l f)(\lambda):=
 f(\beta_l^*\lambda) \quad \text{if}\  n>0
\end{equation}
($ \lambda\in \Lambda_{n-1,m}$) and $\beta_l f=0$ if $n=0$ (where we have used the convention that $\mathcal{F}_{-1,m}:=\{ 0\}$),
\begin{equation}  \label{qboson-repb} 
(\beta^*_l f)(\lambda):=
\begin{cases}
[\text{m}_l(\lambda)]f(\beta_l\lambda)&\text{if}\  \text{m}_l(\lambda)>0 \\
0&\text{otherwise}
\end{cases}
 \end{equation}
($ \lambda\in \Lambda_{n+1,m}$), and
\begin{equation}\label{qboson-repc}
(t^{\pm N_l} f)(\lambda):=t^{\pm \text{m}_l(\lambda)}f(\lambda)  
\end{equation}
($\lambda\in \Lambda_{n,m}$).
\end{subequations}
Indeed, it is readily verified that the operations in question satisfy the ultralocal $q$-boson algebra relations in Eqs. \eqref{e:qboson-a}, \eqref{e:qboson-b}. Moreover, the representation at issue is unitary (i.e. it preserves the $\ast$-structure):
\begin{subequations}
\begin{equation}
(\beta_l f,g)_{n,m}=(f,\beta_l^\ast g)_{n+1,m}\qquad (f\in \mathcal{F}_{n+1,m},\ g\in\mathcal{F}_{n,m})
\end{equation}
(where one exploits that $\delta_{n+1,m}(\beta_l^*\lambda )=\delta_{n,m}(\lambda )/[\text{m}_l(\lambda)+1]$) and
\begin{equation}
(t^{\pm N_l} f, g)_{n,m}=(f,t^{\pm N_l} g)_{n,m}\qquad (f,g\in \mathcal{F}_{n,m}) .
\end{equation}
\end{subequations}
It is clear from these definitions that the $q$-boson  annihilation-, creation-, and $q$-deformed number operators at site $l\in\mathbb{N
}_m$ are bounded on
$\mathcal{F}_m$:
\begin{align}
\langle \beta_l f,\beta_l f\rangle_{n-1,m}&\leq  (1-t)^{-1} \langle  f,f\rangle_{n,m}  ,\nonumber\\
\label{bounded} \langle \beta_l^* f,\beta_l^*f\rangle_{n+1,m}&\leq  (1-t)^{-1} \langle  f,f\rangle_{n,m},\\
\langle t^{N_l} f,t^{N_l} f\rangle_{n,m}&\leq  \langle  f,f\rangle_{n,m} \nonumber
\end{align}
for any $f\in\mathcal{F}_{n,m}$.

\begin{remark}\label{vacuum:rem}
If we denote the characteristic function in $ \mathcal F_{n,m}$ supported on $\lambda\in\Lambda_{n,m}$ by $|\lambda\rangle$,
then the functions  $|\lambda\rangle$, $\lambda\in \Lambda_{n,m}$ constitute a standard basis for  $ \mathcal F_{n,m}$
satisfying the orthogonality:
\begin{equation*}
\langle \mu |\lambda\rangle :=(|\lambda\rangle , |\mu\rangle )_{n,m}=\begin{cases} \delta_{n,m}(\lambda) &\text{if}\ \lambda =\mu \\ 0&\text{if}\ \lambda\neq \mu\end{cases}
\end{equation*}
($\lambda,\mu\in\Lambda_{n,m}$).
The state $|\lambda\rangle$ represents a configuration of $n$ $q$-bosons on $\mathbb{N}_m$ (with $\text{m}_l(\lambda)$ particles occupying the site $l\in\mathbb{N}_m$). The action of the $q$-boson generators on this standard basis reads:
\begin{equation*}
\beta_l |\lambda\rangle = \begin{cases}
|\beta_l\lambda\rangle &\text{if}\ \text{m}_l(\lambda)>0 \\
0&\text{otherwise}
\end{cases},\quad
\beta_l^* |\lambda\rangle =[\text{m}_l(\lambda)+1] | \beta_l^*\lambda\rangle,\quad t^{\pm N_l}|\lambda\rangle=t^{\pm \text{m}_l(\lambda)}|\lambda\rangle .
\end{equation*}
In this picture, the vacuum sector $\mathcal F_{0,m}$ amounts to $\mathbb{C} |\emptyset\rangle$ with $ |\emptyset\rangle$ playing the role of the vacuum state.
\end{remark}

\subsection{$n$-Particle Hamiltonian}
Let $\mathcal{D}_m\subset \mathcal{F}_m$ denote the dense subspace of \emph{terminating} series $F=\sum_{n\geq 0} f_n$ of
$f_n\in\mathcal{F}_{n,m}$ ($n=0,1,2,\ldots$). 
The boundary transfer operator $\mathcal{T}_m(u;a_+,a_-)$ \eqref{T-boundary} (with $u\in \mathbb{R}^*:=\mathbb{R}\setminus \{0\}$), the number operator $\mathcal{N}_m$ \eqref{Nm}, and the Hamiltonian $\mathcal{H}_m$ \eqref{Hm} 
are all mapped 
to essentially self-adjoint operators on $\mathcal{D}_m\subset \mathcal{F}_m$ by the representation of the $q$-boson algebra $\mathbb{A}_m$ in Eqs. \eqref{qboson-repa}--\eqref{qboson-repc}. The finite-dimensional $n$-particle sector  $\mathcal{F}_{n,m}$ \eqref{n-particle-sector}--\eqref{weight-function} constitutes a stable subspace for these operators. In particular, one has that
\begin{equation}
\mathcal{N}_m f = q^{m+1}t^n f\qquad \text{for\ all} \ f\in\mathcal{F}_{n,m} .
\end{equation}

\begin{proposition}[Self-adjointness]\label{self-adjoint:prp}
On the dense domain $\mathcal{D}_m\subset\mathcal{F}_m$, the action of the boundary transfer operator $\mathcal{T}_m(u;a_+,a_-)$ \eqref{T-boundary}, with $u\in \mathbb{R}^*$, is essentially self-adjoint in the $q$-boson Fock space $\mathcal{F}_m$  \eqref{q-boson-fock-space}, \eqref{fock-ip}.
Moreover, the 
finite-dimensional $n$-particle sector  $\mathcal{F}_{n,m}\subset\mathcal{D}_m$ constitutes an invariant subspace for this action.
\end{proposition}
\begin{proof}
Since the $n$-particle sector $\mathcal{F}_{n,m}$ constitutes an eigenspace with eigenvalue $q^{m+1}t^n$ for the action of $\mathcal{N}_m$ \eqref{Nm},
it follows from Propositions \ref{particle-conservation:prp} and \ref{hermitian:prp} that the action of
$\mathcal{T}_m(u;a_+,a_-)$ restricts to a self-adjoint operator in the invariant subspace $\mathcal{F}_{n,m}$. This means that the action of the boundary transfer operator gives rise to  a symmetric operator on $\mathcal{D}_m$. Moreover, as for $z\in\mathbb{C}\setminus\mathbb{R}$ the range $(\mathcal{T}_m(u;a_+,a_-)-z)\mathcal{D}_m$ is equal to $\mathcal{D}_m$
(because $(\mathcal{T}_m(u;a_+,a_-)-z)\mathcal{F}_{n,m}=\mathcal{F}_{n,m}$), the deficiency indices for the action of the boundary transfer operator on $\mathcal{D}_m$ vanish in $\mathcal{F}_m$, i.e. the symmetric operator under consideration is essentially self-adjoint.
\end{proof}

As the value of the spectral parameter $u\in \mathbb{R}^*$ was arbitrary, it follows from Proposition \ref{self-adjoint:prp} that the actions on $\mathcal{D}_m$ of all commuting quantum integrals $\tau_{m,k}(a_+,a_-)$ \eqref{T-boundary-exp}, \eqref{lead-coeff} generated by the boundary transfer operator  are essentially self-adjoint
in $\mathcal{F}_m$ and leave the $n$-particle sector $\mathcal{F}_{n,m}$ invariant.
This holds thus in particular for the Hamiltonian $\mathcal{H}_m$ \eqref{Hm} (which is actually seen to extend to a bounded self-adjoint operator on $\mathcal{F}_m$
in view of Eq. \eqref{bounded}).
The  following proposition makes the action of $\mathcal{H}_m$ in  $\mathcal{F}_{n,m}$ explicit.

\begin{proposition}[$n$-Particle Hamiltonian]\label{action-Hm:prp} For any
 $f\in \mathcal{F}_{n,m}$ \eqref{n-particle-sector}--\eqref{weight-function} and $\lambda\in\Lambda_{n,m}$ \eqref{dominant}, one has that
\begin{align}
(\mathcal{H}_m f)(\lambda)
=& \Bigl( a_-[\emph{m}_0(\lambda)] +
a_+[\emph{m}_m(\lambda)] \Bigr)
f(\lambda)\ +\\
&\sum_{\substack{1\leq j \leq n\\ \lambda+e_j\in\Lambda_{n,m}}} [\emph{m}_{\lambda_j}(\lambda)] f(\lambda+e_j)
+\sum_{\substack{1\leq j \leq n\\ \lambda-e_j\in\Lambda_{n,m}}} [\emph{m}_{\lambda_j}(\lambda)]  f(\lambda-e_j) ,\nonumber
\end{align}
where $e_1,\ldots ,e_n$ stand for the unit vectors of the standard basis in $\mathbb{Z}^n$.
\end{proposition}

\begin{proof}
It is clear from the definitions that
$([N_0]f)(\lambda)=[\text{m}_0(\lambda)]f(\lambda)$ and $([N_m]f)(\lambda)=[\text{m}_m(\lambda)]f(\lambda)$.
Moreover,  for any $0\leq l<m$:
$$
(\beta_{l+1} \beta_l^* f)(\lambda)=
\begin{cases}
[\text{m}_l(\lambda)] f(\beta_{l+1}^* \beta_l\lambda) &\text{if}\  \text{m}_l(\lambda)>0 ,\\
0&\text{otherwise},
\end{cases}$$
where $\beta_{l+1}^* \beta_l\lambda=\lambda+e_j$ with $j=\min\{k\mid \lambda_k=l\}$ (so $l=\lambda_j$),
and
$$(\beta_{l+1}^*\beta_l f)(\lambda)=
\begin{cases}
[\text{m}_{l+1}(\lambda)]f(\beta_{l+1}\beta_l^*\lambda)&\text{if}\  m_{l+1}(\lambda)>0 ,\\
0&\text{otherwise},
\end{cases}$$
where $\beta_{l+1}\beta_l^*\lambda=\lambda-e_j$ with $j=\max\{k\mid \lambda_k=l+1\}$ (so $l=\lambda_j-1$).
\end{proof}


\section{Action of the boundary transfer operator in Fock space}\label{sec7}
In this section, the action of the boundary transfer operator in the $q$-boson Fock space
is calculated explicitly. To this end, we first determine the corresponding actions
of the matrix elements of the $q$-boson monodromy matrices.

\subsection{Action of the periodic monodromy matrix elements}
The explicit actions in Fock space of the elements of the periodic $q$-boson monodromy matrix  were established 
in
\cite[Prp. 4.1]{kor:cylindric},
by viewing the matrix elements in question as partition functions for an associated statistical vertex model. 

To describe the result, some  further notation is needed.
For $\lambda\in\Lambda_{n,m}$ and $\mu\in \Lambda_{n,m}\cup \Lambda_{n-1,m}$ ($n\geq 1$), 
let us write $\mu\preceq\lambda$ if $\mu_j\leq\lambda_j$  for all $j$ and  the skew diagram $\lambda/\mu$ is a horizontal strip:
\begin{equation}
m\geq \lambda_1\geq\mu_1\geq\lambda_2\geq \cdots \geq\mu_{n-1}\geq\lambda_n\geq
\begin{cases}
\mu_n\geq 0 &\text{if}\ \mu\in\Lambda_{n,m} ,\\
0 &\text{if}\ \mu\in\Lambda_{n-1,m} .
\end{cases}
\end{equation}
In this situation, we define
\begin{subequations}
\begin{equation}
\varphi_{\lambda/\mu}(t):=\prod_{\substack{0 \leq l\leq m \\ \text{m}_l(\lambda)=\text{m}_l(\mu)+1}}
(1-t^{\text{m}_l(\lambda)}) 
\end{equation}
and
\begin{equation}
\psi_{\lambda/\mu}(t):=\prod_{\substack{0 \leq l\leq m \\ \text{m}_l(\lambda)=\text{m}_l(\mu)-1}}
(1-t^{\text{m}_l(\mu)}) .
\end{equation}
\end{subequations}
The latter two combinatorial quantities go back to \cite[Ch. III.5]{mac:symmetric}, with the caveat that here we consider (shifted) partitions with possibly
`parts of size zero'  whose multiplicity  is taken into account (cf. also Ref. \cite{whe-zin:refined} for a similar state of affairs). Finally, we also employ the standard abbreviation for the total number of boxes
$$|\lambda|:=\lambda_1+\lambda_2+\cdots +\lambda_n\qquad (\lambda\in\Lambda_{n,m}).$$

\begin{proposition}[Periodic Monodromy Matrix in Fock Space \mbox{\cite[Prp. 4.1]{kor:cylindric}}]\label{periodic-action:prp}
For $u\in\mathbb{C}^*$ and $n\geq 1$, the action of the periodic monodromy matrix elements on
$f \in \mathcal F_{n,m}$  is given by
\begin{subequations}
\begin{equation}\label{Amn-action}
\bigl(A_m(u) f\bigr)(\lambda)=
u^{-m-1}
\sum_{\substack{\mu\in \Lambda_{n,m} \\ \mu\preceq \lambda}}
u^{2(|\lambda|-|\mu|)}  \varphi_{\lambda/\mu }(t) f(\mu)
\qquad (\lambda\in \Lambda_{n,m}),
\end{equation}
\begin{equation}\label{Bmn-action}
\bigl(B_m(u)f\bigr)(\lambda) 
= u^{-m}  \sum_{\substack{\mu\in \Lambda_{n,m} \\ \mu\preceq \lambda}}
u^{2(|\lambda|-|\mu|)}  \varphi_{\lambda/\mu }(t) f(\mu)
\qquad (\lambda\in \Lambda_{n+1,m}) ,
\end{equation}
\begin{equation}\label{Cmn-action}
\bigl(C_m(u)f\bigr)(\lambda) 
= u^{m}  \sum_{\substack{\mu\in \Lambda_{n,m} \\ \lambda\preceq \mu}}
u^{2(|\lambda|-|\mu|)}  \psi_{\mu /\lambda  }(t) f(\mu)
\qquad (\lambda\in \Lambda_{n-1,m}) ,
\end{equation}
\begin{equation}\label{Dmn-action}
\bigl(D_m(u) f\bigr)(\lambda)=
u^{m+1}
\sum_{\substack{\mu\in \Lambda_{n,m} \\ \lambda\preceq \mu}}
u^{2(|\lambda|-|\mu|)}  \psi_{\mu  /\lambda }(t) f(\mu)
\qquad (\lambda\in \Lambda_{n,m}).
\end{equation}
\end{subequations}
\end{proposition}

For the reader's convenience, the end of this section (below) includes an alternative direct proof of these formulas  that is inductive in the number of sites and avoids the digression into vertex models altogether.

\subsection{Action of the boundary monodromy matrix elements}
By combining Lemma \ref{U-boundary-matrix-elements:lem} and Proposition \ref{periodic-action:prp}, one arrives at the action of the elements of the boundary monodromy matrix in the $q$-boson Fock space. 

To describe these actions, it is convenient to introduce the following three relations on our partitions.
Firstly, for $\lambda\in\Lambda_{n,m}$ and $\mu\in\Lambda_{n-1,m}$, we write
that $\mu\leq \lambda$ if there exists a $\nu\in\Lambda_{n,m}\cup \Lambda_{n-1,m}$ such that
$\mu\preceq\nu\preceq\lambda$.
Secondly, for $\lambda,\mu \in\Lambda_{n,m}$, we write
that $\mu\sim_- \lambda$ if there exists a $\nu\in\Lambda_{n,m}\cup \Lambda_{n-1,m}$ such that
$\nu\preceq\lambda$ and $\nu\preceq\mu$, and we write $\mu\sim_+ \lambda$ if there exists a $\nu\in\Lambda_{n,m}\cup \Lambda_{n+1,m}$ such that
$\lambda\preceq\nu$ and $\mu\preceq\nu$.

\begin{proposition}[Boundary Monodromy Matrix in Fock Space]\label{boundary-action:prp}
For $u\in\mathbb{C}^*$ and $n\geq 1$, the action of the boundary monodromy matrix elements on  $f\in\mathcal{F}_{n,m}$
is given by
\begin{subequations}
\begin{equation}
\bigl(\mathcal{A}_m(u;a)f\bigr)(\lambda) =  
q^{-m-1}t^{-n}u^{-1} \sum_{\substack{\mu\in \Lambda_{n,m} \\ \mu \sim_- \lambda}}
\emph{A}^{(n)}_{\lambda, \mu}(u^2;t,a) f(\mu)\quad (\lambda\in\Lambda_{n,m}),
\end{equation}
\begin{equation}\label{B-boundary-action}
\bigl({\mathcal{B}}_m(u;a)f\bigr)(\lambda) =  q^{-m}t^{-n-1} \sum_{\substack{\mu\in \Lambda_{n,m} \\ \mu \leq \lambda}}
\emph{B}^{(n)}_{\lambda,\mu}(u^2;t,a) f(\mu)\quad (\lambda\in\Lambda_{n+1,m}),
\end{equation}
\begin{equation}
\bigl({\mathcal{C}}_m(u;a)f\bigr)(\lambda) = q^{-m-1}t^{-n}
  \sum_{\substack{\mu\in \Lambda_{n,m} \\ \lambda \leq \mu}}
\emph{C}^{(n)}_{\lambda,\mu}(u^2;t,a) f(\mu)\quad (\lambda\in\Lambda_{n-1,m}),
\end{equation}
\begin{equation}
\bigl(\mathcal{D}_m(u;a)f\bigr)(\lambda) =   q^{-m}t^{-n-1} u\sum_{\substack{\mu\in \Lambda_{n,m} \\ \mu \sim_+ \lambda}}
\emph{D}^{(n)}_{\lambda, \mu}(u^2;t,a) f(\mu)\quad (\lambda\in\Lambda_{n,m}),
\end{equation}
\end{subequations}
where
 \begin{subequations}
  \begin{align}\label{A-boundary-coef}
\emph{A}^{(n)}_{\lambda, \mu}(z;t,a) :=&
z^{-m} (a-z^{-1})
\sum_{\substack{ \nu\in\Lambda_{n,m}\\\nu\preceq\lambda,\, \nu\preceq\mu  }}
\varphi_{\lambda/\nu} (t)\psi_{\mu/\nu} (t)   z^{|\lambda|+|\mu|-2|\nu|} \\
+&
z^{-m} (tz-a)
\sum_{\substack{ \nu\in\Lambda_{n-1,m}\\\nu\preceq\lambda,\, \nu\preceq\mu  }}
\varphi_{\lambda/\nu}(t)\psi_{\mu/\nu} (t)   z^{|\lambda|+|\mu|-2|\nu | } ,
 \nonumber
\end{align}
 \begin{align}\label{B-boundary-coef}
\emph{B}^{(n)}_{\lambda,\mu}(z;t,a) :=&
(z^{-1}-a)
\sum_{\substack{ \nu\in\Lambda_{n+1,m}\\\mu\preceq\nu\preceq\lambda  }}
\varphi_{\lambda/\nu} (t)\varphi_{\nu/\mu} (t)   z^{|\lambda|+|\mu|-2|\nu|} \\
+&
(a-tz)
\sum_{\substack{ \nu\in\Lambda_{n,m}\\\mu\preceq\nu\preceq\lambda  }}
\varphi_{\lambda/\nu}(t)\varphi_{\nu/\mu} (t)   z^{|\lambda|+|\mu|-2|\nu | } ,
 \nonumber
\end{align}
 \begin{align}\label{C-boundary-coef}
\emph{C}^{(n)}_{\lambda,\mu}(z;t,a) :=&
(a-z^{-1})
\sum_{\substack{ \nu\in\Lambda_{n,m}\\\lambda\preceq\nu\preceq\mu  }}
\psi_{\nu/\lambda} (t)\psi_{\mu/\nu} (t)   z^{|\lambda|+|\mu|-2|\nu|} \\
+&
(tz-a)
\sum_{\substack{ \nu\in\Lambda_{n-1,m}\\ \lambda \preceq\nu\preceq\mu  }}
\psi_{\nu/\lambda }(t)\psi_{\mu/\nu} (t)   z^{|\lambda|+|\mu|-2|\nu | } ,
 \nonumber
\end{align}
 \begin{align}\label{D-boundary-coef}
\emph{D}^{(n)}_{\lambda, \mu}(z;t,a) :=&
 z^m(z^{-1}-a)
\sum_{\substack{ \nu\in\Lambda_{n+1,m}\\\lambda\preceq\nu,\, \mu \preceq\nu  }}
\psi_{\nu/\lambda } (t)\varphi_{\nu/\mu} (t)   z^{|\lambda|+|\mu|-2|\nu|} \\
+&
 z^m (a-tz)
\sum_{\substack{ \nu\in\Lambda_{n,m}\\\lambda\preceq\nu,\, \mu \preceq\nu  }}
\psi_{\nu/\lambda } (t)\varphi_{\nu/\mu} (t)   z^{|\lambda|+|\mu|-2|\nu|} .
 \nonumber
\end{align}
\end{subequations}

\end{proposition}
\begin{proof}
We will compute the action of ${\mathcal B}_m(u;a) $. The actions of the other three matrix elements follow analogously by combining Lemma \ref{U-boundary-matrix-elements:lem} and Proposition \ref{periodic-action:prp}. Specifically,
upon composing the actions in Eqs. \eqref{Amn-action} and \eqref{Bmn-action}, it is seen that for any $f\in\mathcal{F}_{n,m}$ and $\lambda\in\Lambda_{n+1,m}$:
\begin{equation*}
\Bigl(A_m(u)   B_m(u^{-1})f\Bigr) (\lambda) 
=  u^{-1} \sum_{\substack{ \nu \in \Lambda_{n+1,m}  \\ \mu \in \Lambda_{n,m} \\  \mu\preceq  \nu \preceq \lambda}}                     
      u^{2(|\lambda|+|\mu| - 2|\nu|)}  \varphi_{\lambda/\nu} (t)  \varphi_{\nu/\mu }(t) f(\mu)  
\end{equation*}
and
\begin{equation*}
\Bigl(B_m(u) A_m(u^{-1}) f\Bigr)(\lambda) 
=    u \sum_{\substack{ \nu \in \Lambda_{n,m} \\  \mu \in \Lambda_{n,m} \\ \mu\preceq \nu \preceq \lambda}}        
         u^{2(|\lambda|+|\mu|-2|\nu|)}         \varphi_{\lambda/\nu}(t)     \varphi_{\nu/\mu}(t)      f(\mu) .
\end{equation*}
Armed with these two formulas, the asserted action follows upon
acting on $f$ with the expression for ${\mathcal B}_m(u;a) $ in Lemma \ref{U-boundary-matrix-elements:lem}.
\end{proof}

\subsection{Action of the boundary transfer operator}
From Proposition \ref{boundary-action:prp}, the following explicit formula for the action of the boundary transfer operator $\mathcal{T}_m(u;a_+,a_-)$ \eqref{T-boundary} in the $q$-boson Fock space is immediate.

\begin{theorem}[Boundary Transfer Operator in Fock Space]\label{boundary-T-action:thm}
For $u\in\mathbb{C}^*$ and $n\geq 1$, the action of the boundary transfer operator on  $f\in\mathcal{F}_{n,m}$
reads:
\begin{align*}
\bigl(\mathcal{T}_m(u;a_+,a_-)f\bigr)(\lambda) =  &
q^{-m}t^{-n-1} (a_+ -tu^{-2}) \sum_{\substack{\mu\in \Lambda_{n,m} \\ \mu \sim_- \lambda}}
\emph{A}^{(n)}_{\lambda, \mu}(u^2;t,a_-) f(\mu)  \\
&+
 q^{-m}t^{-n-1} (a_+ -u^2) \sum_{\substack{\mu\in \Lambda_{n,m} \\ \mu \sim_+ \lambda}}
\emph{D}^{(n)}_{\lambda, \mu}(u^2;t,a_-) f(\mu) ,
\end{align*}
($\lambda\in \Lambda_{n,m}$), where the coefficients $\emph{A}^{(n)}_{\lambda ,\mu}(z;t,a)$ and $\emph{D}^{(n)}_{\lambda, \mu}(z;t,a)$ are given
by Eqs. \eqref{A-boundary-coef} and \eqref{D-boundary-coef}, respectively.
\end{theorem}

\subsection{Proof of Proposition \ref{periodic-action:prp}}
We will first verify the formulas for the actions of $A_m(u)$ and $C_m(u)$ by induction in $m$ with the aid of the recursion in Eq. \eqref{e:Am+1}. Next, by computing the adjoints of these actions  in $\mathcal{F}_m$, one recovers
the actions of 
$B_m(u)$ and $D_m(u)$  upon recalling that
$B_m(u)=(1-t)C_m(u^{-1})^*$ and $D_m(u)=A_m(u^{-1})^*$ for $u\in \mathbb{R}^*$
(by Lemma \ref{U-adjoint-reverse:lem} and the unitarity of the representation of $\mathbb{A}_m$ in $\mathcal{F}_m$).

For $m=0$, one has that $\mathbb{N}_m=\mathbb{N}_0=\{ 0\}$, $\Lambda_{n,m}=\Lambda_{n,0}=\{ 0^n\}$, $A_m(u)=A_0(u)=u^{-1}$ and $C_m(u)=C_0(u)=\beta_0$ (cf. Eq. \eqref{L-op}). Indeed, in this trivial situation the asserted formulas 
correctly reproduce that $$\bigl(A_0(u)f\bigr)(0^n)= u^{-1}\varphi_{0^n/0^n}(t) u^{2(|0^n|-|0^n|)}f(0^n)=u^{-1}f(0^n)$$ and that
$$\bigl(C_0(u)f\bigr)(0^{n-1})= \psi_{0^n/0^{n-1}}(t) u^{2(|0^n|-|0^{n-1}|)}f(0^n)=f(0^n)=(\beta_0f)(0^{n-1}) .$$

For $m>0$, we read-off from Eq. \eqref{e:Am+1} that
$$
A_m(u)=u^{-1}A_{m-1}(u)+(1-t)\beta_m^* C_{m-1}(u), \ \
C_m(u)=u C_{m-1}(u)+\beta_m A_{m-1}(u),
$$
which enables to compute the action of $A_m(u)$ and $C_m(u)$ by means of the induction hypothesis (IH).
Indeed, given $f\in\mathcal{F}_{n,m}$ and $\lambda\in\Lambda_{n,m}$ with
$r=\text{m}_m(\lambda)$, we have that
\begin{align*}
\bigl( A_{m-1}(u)f\bigr)(\lambda) = &\bigl(A_{m-1}(u) \beta_m^r f \bigr)(\beta_m^r \lambda )  \\
\stackrel{(IH)}{=} &u^{-m} \sum_{\substack{\mu\in\Lambda_{n-r,m-1}\\ \mu\preceq \beta_m^r\lambda }}
\varphi_{\beta_m^r\lambda/\mu}(t) u^{2(|\beta_m^r\lambda |-|\mu|)} (\beta_m^r f) (\mu) \\
= & u^{-m} \sum_{\substack{\mu\in\Lambda_{n,m},\, \mu\preceq \lambda \\  \text{m}_m(\mu)=\text{m}_m(\lambda)}}
\varphi_{\lambda/\mu}(t) u^{2(|\lambda|-|\mu|)} f(\mu)
\end{align*}
and (if $r>0$)
\begin{align*}
\bigl( \beta_m^* C_{m-1}(u)f\bigr)(\lambda)= & [\text{m}_m(\lambda)] \bigl(  C_{m-1}(u)f\bigr)(\beta_m\lambda) 
= [\text{m}_m(\lambda)]  \bigl(  C_{m-1}(u)\beta_m^{r-1}f \bigr)(\beta_m^r \lambda ) \\
\stackrel{(IH)}{=}& u^{m-1}  [\text{m}_m(\lambda)] 
\sum_{\substack{\mu\in\Lambda_{n+1-r,m-1}\\ \beta_m^r\lambda \preceq   \mu  }}
\psi_{\mu/ \beta_m^r\lambda }(t) u^{2(|\beta_m^r\lambda |-|\mu|)}(\beta_m^{r-1}f )(\mu)  \\
= & u^{-m-1}  [\text{m}_m(\lambda)] 
\sum_{\substack{\mu\in\Lambda_{n,m},\,  \mu \preceq   \lambda \\ \text{m}_m(\mu)=\text{m}_m(\lambda)-1}}
\psi_{\beta_m^{r-1}\mu /\beta_m^r\lambda }(t) u^{2(|\lambda |-|\mu|)}f (\mu)  \\
=& u^{-m-1} (1-t)^{-1}
\sum_{\substack{\mu\in\Lambda_{n,m},\,  \mu \preceq   \lambda \\ \text{m}_m(\mu)=\text{m}_m(\lambda)-1}}
\varphi_{\lambda /\mu }(t) u^{2(|\lambda |-|\mu|)}f (\mu) 
\end{align*}
(while $\bigl( \beta_m^* C_{m-1}(u)f\bigr)(\lambda)=0$ if $r=0$),
whence $$\bigl( A_m(u)f\bigr) (\lambda)= u^{-m-1}
\sum_{\mu\in\Lambda_{n,m},\, \mu \preceq   \lambda }
\varphi_{\lambda /\mu }(t) u^{2(|\lambda |-|\mu|)}f (\mu)$$ as claimed. Similarly, we see that for $f\in\mathcal{F}_{n,m}$ and $\lambda\in\Lambda_{n-1,m}$ with $r=\text{m}_m(\lambda)$:
\begin{align*}
\bigl( C_{m-1}(u)f \bigr) (\lambda) =& \bigl( C_{m-1}(u) \beta_m^r f \bigr) (\beta_m^r \lambda) \\
\stackrel{(IH)}{=}& u^{m-1} 
\sum_{\substack{\mu\in\Lambda_{n-r,m-1}\\ \beta_m^r\lambda \preceq   \mu  }}
\psi_{\mu/ \beta_m^r\lambda }(t) u^{2(|\beta_m^r\lambda |-|\mu|)}(\beta_m^r f )(\mu)  \\
=& u^{m-1} 
\sum_{\substack{\mu\in\Lambda_{n,m},\,  \lambda \preceq   \mu \\ \text{m}_m(\mu)=\text{m}_m(\lambda)}}
\psi_{\mu /\lambda }(t) u^{2(|\lambda |-|\mu|)}f (\mu)
\end{align*}
and
\begin{align*}
\bigl( \beta_m A_{m-1}(u)f\bigr) (\lambda ) =& \bigl(  A_{m-1}(u)f\bigr) (\beta_m^* \lambda )  =
\bigl(  A_{m-1}(u) \beta_m^{r+1}f \bigr) (\beta_m^r \lambda ) \\
\stackrel{(IH)}{=}&
u^{-m} \sum_{\substack{\mu\in\Lambda_{n-r-1,m-1}\\ \mu\preceq \beta_m^r\lambda }}
\varphi_{\beta_m^r\lambda/\mu}(t) u^{2(|\beta_m^r\lambda |-|\mu|)} (\beta_m^{r+1}f) (\mu) \\
=& 
u^{m} \sum_{\substack{\mu\in\Lambda_{n,m},\, \lambda\preceq \mu \\ \text{m}_m(\mu)=\text{m}_m(\lambda)+1 }}
\psi_{\mu /\lambda }(t) u^{2(| \lambda |-|\mu|)}f (\mu) ,
\end{align*}
whence
$$
\bigl(C_m(u)f\bigr)(\lambda) 
= u^{m}  \sum_{\mu\in \Lambda_{n,m},\, \lambda\preceq \mu}
u^{2(|\lambda|-|\mu|)}  \psi_{\mu /\lambda  }(t) f(\mu)
$$
as claimed.

The computation of the adjoints hinges on the following two elementary identities:
\begin{subequations}
\begin{align}
\label{pearsona}  \varphi_{\lambda /\mu}(t) \delta_{n+1,m}(\lambda)=& (1-t)\psi_{\lambda /\mu}(t)\delta_{n,m}(\mu)\qquad (\lambda\in\Lambda_{n+1,m},\ \mu\in\Lambda_{n,m}) ,\\
\label{pearsonb}  \varphi_{\lambda /\mu}(t)\delta_{n,m}(\lambda) =& \psi_{\lambda /\mu}(t) \delta_{n,m}(\mu)\qquad \qquad \quad  (\lambda,\mu\in\Lambda_{n,m}) .
\end{align}
\end{subequations}
Indeed,  for any $f\in\mathcal{F}_{n,m}$, $g\in\mathcal{F}_{n+1,m}$ and $u\in\mathbb{R}^*$ one has that:
\begin{align*}
& \bigl(   B_m(u)f,g      \bigr)_{n+1,m}= \bigl(   f, (1-t) C_m(u^{-1})g      \bigr)_{n,m} \\
&=
(1-t) u^{-m} \sum_{\mu\in\Lambda_{n,m}}   \biggl( \delta_{n,m}(\mu) f(\mu) 
\sum_{\substack{\lambda\in\Lambda_{n+1,m}\\ \mu \preceq \lambda}}
\psi_{\lambda/\mu}(t) u^{2(|\lambda |-|\mu |)} \overline{g(\lambda )}
\biggr) \\
&\stackrel{\text{Eq.} \eqref{pearsona}}{=} 
 u^{-m} \sum_{\lambda \in\Lambda_{n+1,m}}   \biggl( \delta_{n+1,m}(\lambda )\overline{g(\lambda )}
\sum_{\substack{\mu\in\Lambda_{n,m}\\ \mu \preceq \lambda}}
\varphi_{\lambda /\mu}(t) u^{2(|\lambda|-|\mu |)}  f(\mu ) 
\biggr) ,
\end{align*}
whence
$$
\bigl( B_m(u)f \bigr) (\lambda)= u^{-m}\sum_{\substack{\mu\in\Lambda_{n,m}\\ \mu \preceq \lambda}}
\varphi_{\lambda /\mu}(t) u^{2(|\lambda|-|\mu |)}  f(\mu ) 
$$
as claimed. Similarly,
for any $f,g\in\mathcal{F}_{n,m}$ and $u\in\mathbb{R}^*$  one has that:
\begin{align*}
& \bigl(   D_m(u)f,g      \bigr)_{n,m}= \bigl(   f,  A_m(u^{-1})g      \bigr)_{n,m} \\
&=
 u^{m+1} \sum_{\mu\in\Lambda_{n,m}}   \biggl( \delta_{n,m}(\mu) f(\mu) 
\sum_{\substack{\lambda\in\Lambda_{n,m}\\ \lambda \preceq \mu}}
\varphi_{\mu/\lambda}(t) u^{2(|\lambda |-|\mu |)} \overline{g(\lambda )}
\biggr) \\
&\stackrel{\text{Eq.} \eqref{pearsonb}}{=} 
 u^{m+1} \sum_{\lambda \in\Lambda_{n,m}}   \biggl( \delta_{n,m}(\lambda )\overline{g(\lambda )}
\sum_{\substack{\mu\in\Lambda_{n,m}\\ \lambda \preceq \mu }}
\psi_{\mu /\lambda }(t) u^{2(|\lambda|-|\mu |)}  f(\mu ) 
\biggr) ,
\end{align*}
whence
$$
\bigl( D_m(u)f \bigr) (\lambda)= u^{m+1}\sum_{\substack{\mu\in\Lambda_{n,m}\\ \lambda \preceq \mu }}
\psi_{\mu /\lambda }(t) u^{2(|\lambda|-|\mu |)}  f(\mu )  
$$
as claimed.


\section{Algebraic Bethe Ansatz}\label{sec8}
In this section we construct eigenfunctions for
the boundary transfer operator in the $q$-boson Fock space by means of Sklyanin's extension  \cite{skl:boundary} of the algebraic Bethe Ansatz 
formalism \cite{tak:integrable,kor-bog-ize:quantum,jim-miw:algebraic,fad:how} allowing for open-end boundary interactions.

\subsection{Bethe Ansatz eigenfunction}
It is clear from Lemma \ref{U-N-com:lem} and confirmed explicitly by Propostion \ref{boundary-action:prp}
that the upper diagonal matrix element $\mathcal{B}_m(u;a) $ of the boundary monodromy matrix acts as a particle-creation operator in $\mathcal{F}_m$.  For our purposes it turns out convenient to modify the normalization of this Bethe Ansatz creation operator somewhat:
\begin{equation}\label{ABA-creation-operator}
\hat{\mathcal{B}}_m(u;a) :=  b(u)^{-1} \mathcal{B}_m(u;a) \mathcal{N}_m\quad\text{with}\quad  b(u):=qs(q) s(qu^2).
\end{equation}
The following proposition summarizes the main features of the normalized Bethe Ansatz creation operator.

\begin{proposition}[Bethe Ansatz Creation Operator]\label{ABA-creation-operator:prp}
The operator $\hat{\mathcal{B}}_m(u;a) $ \eqref{ABA-creation-operator} enjoys the following properties:
\begin{itemize}
\item[i)] Particle creation operator:  $\mathcal{N}_m \hat{\mathcal{B}}_m(u;a) = t\hat{\mathcal{B}}_m(u;a)\mathcal{N}_m $,
\item[ii)] Commuting family:    $[\hat{\mathcal{B}}_m(u;a),\hat{\mathcal{B}}_m(v;a) ]=0$ in
$\mathbb{C}[u^{\pm 1},v^{\pm 1}]\otimes \mathbb{A}$,
\item[iii)] Reflection symmetry:  $\hat{\mathcal{B}}_m(u^{-1};a) = \hat{\mathcal{B}}_m(u;a)  $,
\item[iv)] Even Laurent polynomial:  $\hat{\mathcal{B}}_m(u;a) \in \mathbb{C}[u^{\pm 2}]\otimes \mathbb{A}$.
\end{itemize}
\end{proposition}
\begin{proof}
Parts i) and ii) are immediate from the commutation relations in Lemmas \ref{U-N-com:lem} and \ref{boundary-monodromy-relations:lem}.
For part $iii)$, it suffices to verify that $b(u^{-1}) \mathcal{B}_m(u;a)\mathcal{N}_m=  b(u) \mathcal{B}_m(u^{-1};a)\mathcal{N}_m$.
With the aid of Lemma \ref{U-boundary-matrix-elements:lem}, the LHS is rewritten as
$$
b(u^{-1}) \Bigl( -q^{-1} e(u;a)  A_m(u) B_m(u^{-1}) +    f(u;a) B_m(u) A_m(u^{-1})\Bigr)
$$
and the RHS is rewritten as
$$
b(u) \Bigl( -q^{-1} e(u^{-1};a)  A_m(u^{-1}) B_m(u) +    f(u^{-1};a) B_m(u^{-1}) A_m(u)\Bigr) .
$$
After further  rewriting of the RHS by means of the relations (cf. Lemma \ref{periodic-monodromy-relations:lem})
 $$
 A_m(u^{-1}) B_m(u)
= \frac{s(q^{-1})}{s(q^{-1}u^{-2})}A_m(u) B_m(u^{-1})    + \frac{qs(u^{-2})}{s(q^{-1}u^{-2})} B_m(u) A_m(u^{-1})
 $$
and
$$
B_m(u^{-1}) A_m(u)
=  \frac{q^{-1} s(u^{-2})}{s(q^{-1}u^{-2})} A_m(u) B_m(u^{-1})  + \frac{s(q^{-1})}{s(q^{-1}u^{-2})}B_m(u) A_m(u^{-1})   ,
$$
one readily infers the desired equality upon comparing the coefficients of the terms involving $A_m(u) B_m(u^{-1}) $ and
$B_m(u) A_m(u^{-1})$ on both sides.
Finally, for Part iv) let us recall that it is clear from the definitions that  ${\mathcal{B}}_m(u;a) \in \mathbb{C}[u^{\pm 1}]\otimes \mathbb{A}$. To see that 
 $\hat{\mathcal{B}}_m(u;a) $ \eqref{ABA-creation-operator} is also a Laurent polynomial in $u$, one invokes the reflection-symmetry of Part $iii)$ to conclude that the apparent poles at $qu^2=\pm 1$ (arising from the normalization) are cancelled. The fact that the operator in question is actually even in $u$ is seen from Lemmas \ref{U-boundary-matrix-elements:lem} and \ref{T-exp:lem}.
\end{proof}

After these preparations, we are now in the position to define the following $n$-particle Bethe Ansatz wave function:
\begin{equation}\label{ABA-wave-function}
\Psi_{(v_1,\ldots ,v_n)}:= \hat{\mathcal{B}}_m(v_1;a_{-}) \cdots \hat{\mathcal{B}}_m(v_n;a_{-}) |\emptyset\rangle \in\mathcal{F}_{n,m},
\end{equation}
obtained by acting with Bethe Ansatz creation operators on the vacuum state specified in Remark \ref{vacuum:rem}.
Here (and below), it is assumed that
the spectral variables $v_1,\ldots ,v_n$ take values in $\mathbb{C}^*:=\mathbb{C}\setminus \{ 0\}$.
It is immediate from Proposition \ref{ABA-creation-operator:prp} that $\Psi_{(v_1,\ldots ,v_n)}$
\eqref{ABA-wave-function} constitutes an even Laurent polynomial in
$v_1,\ldots ,v_n$ that is invariant
with respect to the action of the hyperoctahedral group by permutations and reflections $v_j\to v_j^{-1}$
of these spectral variables. 
The main result of this section states that the $n$-particle Bethe Ansatz wave function \eqref{ABA-wave-function} satisfies the eigenvalue equation for the boundary transfer operator $\mathcal{T}_m(u;a_+,a_-)$ \eqref{T-boundary} in $\mathcal{F}_{n,m}$, provided the spectral variables satisfy a corresponding algebraic system of Bethe Ansatz equations.

\begin{theorem}[Bethe Ansatz Eigenfunction]\label{BAE:thm}
Let $u\in\mathbb{R}^*\setminus \{ 1,-1\}$ and $(v_1,\ldots ,v_n)\in (\mathbb{C}^*)^n$ be generic in the sense that
 $u^2v_j^{\pm 2}\neq 1$, $ v_j^{2} \neq a_\pm $,  and $v_j^2v_k^{\pm 2}\neq t$
($1\leq j\neq k\leq n$).
The $n$-particle Bethe Ansatz wave function 
$\Psi_{(v_1,\ldots ,v_n)}$
\eqref{ABA-wave-function}  solves the eigenvalue equation
\begin{subequations}
\begin{equation}
\mathcal{T}_m(u;a_+,a_-) \Psi_{(v_1,\ldots ,v_n)}= E_{n,m}(u;v_1,\ldots ,v_n)  \Psi_{(v_1,\ldots ,v_n)}
\end{equation}
with eigenvalue
\begin{align}\label{ABA-eigenvalue}
& E_{n,m}(u;v_1,\ldots ,v_n) := \\
q^{-m-1} \Biggl( &u^{-2m-2} \frac{s(q^{-1}u^{2})}{s(u^{2})}e(u;a_+)e(u;a_-)
     \prod_{1\leq j\leq n} \frac{s(quv_j) s(quv_j^{-1})}{s(uv_j) s(uv_j^{-1})}           
     \nonumber  \\
  & +  u^{2m+2}  \frac{s(q^{-1} u^{-2})}{s(u^{-2})}e(u^{-1};a_+)e(u^{-1};a_-)
     \prod_{1\leq j\leq n} \frac{s(qu^{-1}v_j) s(qu^{-1}v_j^{-1})}{s(u^{-1}v_j) s(u^{-1}v_j^{-1})} \Biggr) , \nonumber
\end{align}
provided the spectral variables $v_1,\ldots ,v_n$ satisfy the following algebraic system of  \emph{Bethe Ansatz equations}:
\begin{equation}\label{e:BAE}
v_j^{4m+4}
=\frac{e(v_j;a_+) e(v_j;a_-)}{e(v_j^{-1};a_+) e(v_j^{-1};a_-)}
\prod_{\substack{1\leq k \leq n\\ k\neq j}} 
    \frac{s(q v_j v_k) s(q v_j v_k^{-1})  }{s(q^{-1} v_jv_k) s(q^{-1} v_j v_k^{-1} )  }, 
\end{equation}
\end{subequations}
for $j=1,2,\dots, n$.
\end{theorem}

The eigenvalue $E_{n,m}(u;v_1,\ldots,v_n)$ \eqref{ABA-eigenvalue} is given by a rational expression in $u$ that becomes a Laurent polynomial in $u$ when the Bethe Ansatz equations \eqref{e:BAE} are satisfied. Its expansion around $u=0$ is of the form
\begin{subequations}
\begin{align}
E_{n,m}  & (u;v_1,\ldots ,v_n)=   q^{-m}t^{-n} u^{-2m-4}   \\ &+  q^{-m}t^{-n}\Bigl(  (1-t)E_n(v_1,\ldots,v_n) -a_+-a_-  \Bigr)    u^{-2m-2}    + O(u^{-2m}) ,\nonumber
\end{align}
where
\begin{equation}\label{Hm-ev}
E_n(v_1,\ldots,v_n):= \sum_{1\leq j \leq n} v_j^2+v_j^{-2} .
\end{equation}
\end{subequations}
By comparing with the expansion of the boundary transfer operator in Eqs. \eqref{T-boundary-exp}, \eqref{lead-coeff}, we read-off that the eigenvalue of the $q$-boson Hamiltonian $\mathcal{H}_m$ \eqref{Hm} on our $n$-particle Bethe Ansatz eigenfunction is given by $E_n(v_1,\ldots,v_n)$ \eqref{Hm-ev}.

\begin{corollary}[Eigenvalue of the Hamiltonian]\label{H-EV:cor}
For $(v_1,\ldots ,v_n)\in (\mathbb{C}^*)^n$ generic in the sense that
 $ v_j^{2} \neq a_\pm $  and $v_j^2v_k^{\pm 2}\neq t$
($1\leq j\neq k\leq n$), 
the $n$-particle Bethe Ansatz wave function 
$\Psi_{(v_1,\ldots ,v_n)}$
\eqref{ABA-wave-function}  solves the eigenvalue equation
\begin{equation}
\mathcal{H}_m \Psi_{(v_1,\ldots ,v_n)}= E_{n} (v_1,\ldots,v_n) \Psi_{(v_1,\ldots ,v_n)}
\end{equation}
for the $q$-boson Hamiltonian $\mathcal{H}_m$ \eqref{Hm}, provided
the spectral variables $v_1,\ldots ,v_n$ satisfy the algebraic system of Bethe Ansatz equations in Eq. \eqref{e:BAE}.
\end{corollary}

\subsection{Proof of Theorem \ref{BAE:thm}}
We first compute the action of the lower triangular matrix elements of the monodromy matrices on the vacuum state.
 
\begin{lemma}\label{U-periodic-action:lem}
For $u\in\mathbb{C}^*$, one has that
\begin{equation}
A_m(u)|\emptyset\rangle = u^{-m-1}|\emptyset\rangle ,\quad
C_m(u)|\emptyset\rangle = 0,\quad
D_m(u)|\emptyset\rangle = u^{m+1}|\emptyset\rangle .
\end{equation}
\end{lemma}
\begin{proof}
For $m=0$ these actions are immediate from the Lax matrix \eqref{L-op} and the fact that $|\emptyset\rangle$ is a vacuum state annihilated by the $q$-boson operators $\beta_l$, $l=0,\ldots ,m$ (cf. Remark \ref{vacuum:rem}). The case of general $m>0$ follows inductively via the recurrence
in Eq. \eqref{e:Am+1} (using also the ultralocality).
\end{proof}

In the case of the boundary monodromy matrix, it turns out to be convenient to employ the modified matrix element \cite{skl:boundary} (cf. also Appendix \ref{appB} below):
\begin{equation}\label{Dmhat}
\hat{\mathcal  D}_m(u;a) := \mathcal D_m(u;a) + \frac{s(q)}{s(u^2)} \mathcal A_m(u;a) .
\end{equation}

\begin{lemma}\label{U-boundary-action:lem}
For $u\in\mathbb{C}^*$, one has that
\begin{equation}
\mathcal{A}_m(u;a)|\emptyset\rangle = \alpha_m(u)|\emptyset\rangle ,\quad
\mathcal{C}_m(u;a)|\emptyset\rangle = 0,\quad
\hat{\mathcal{D}}_m(u;a )|\emptyset\rangle = \delta_m(u) |\emptyset\rangle ,
\end{equation}
with
\begin{equation*}
\alpha_m(u):= q^{-m-1}u^{-2m-2}e(u;a)\quad \text{and}\quad \delta_m(u):=q^{-m-1}u^{2m+2}\frac{s(qu^2)}{s(u^2)}   e(u^{-1};a) .
\end{equation*}
\end{lemma}
\begin{proof}
The stated actions are readily computed with the aid of Lemmas \ref{U-boundary-matrix-elements:lem} and \ref{U-periodic-action:lem}.
For $\mathcal{A}_m(u;a)$ and $\mathcal{C}_m(u;a)$ the computations in question are straightforward, whereas for 
$\hat{\mathcal{D}}_m(u;a )$ one relies on Lemma \ref{periodic-monodromy-relations:lem} in the form
of the identity
$$
C_m(u) B_m(u^{-1})
= t B_m(u^{-1}) C_m(u)
+ \frac{1-t}{s(u^2)} \Bigl( A_m(u^{-1}) D_m(u) - A_m(u) D_m(u^{-1})\Bigr)
$$
to deduce that
$$
C_m(u) B_m(u^{-1})|\emptyset\rangle
= (1-t) \frac{s(u^{2m})}{s(u^2)} |\emptyset\rangle .
$$
The latter formula renders the computation of the action of $\mathcal{D}_m(u;a)$ and thus that of  $\hat{\mathcal{D}}_m(u;a)$ straightforward, similar to the two previous cases.
\end{proof}

One learns from these lemmas in particular that $|\emptyset\rangle$ constitutes a vacuum state for the Bethe Ansatz annihilation operators $C_m(u)$ and $\mathcal{C}_m(u;a)$ (cf. Lemmas \ref{U-N-com:lem}, \ref{U-boundary-N-com:lem} and Propositions \ref{periodic-action:prp}, \ref{boundary-action:prp}). We now extend the above actions of $\mathcal{A}_m(u;a)$ and $\hat{\mathcal{D}}_m(u;a)$ to  arbitrary (unnormalized) $n$-particle Bethe Ansatz wave functions of the form
\begin{equation}\label{ABA-wave-function-unnormalized}
\psi_{(v_1,\ldots ,v_n)}:= {\mathcal{B}}_m(v_1;a) \cdots {\mathcal{B}}_m(v_n;a) |\emptyset\rangle .
\end{equation}

\begin{proposition}\label{U-boundary-action-general:prp}
For $u\in\mathbb{C}^*$ and $(v_1,\ldots,v_n)\in(\mathbb{C}^*)^{n}$ generic, in the sense that
$u^2\neq \pm 1$, $v_j^2\neq \pm 1$, $u^2v_j^{\pm 2}\neq 1$ and $v_j^2v_k^{\pm 2}\neq 1$ 
($1\leq j\neq k\leq n$), one has that
\begin{subequations}
\begin{align}\label{e:Auonphi}
\mathcal A_m(u;a) \psi_{(v_1,\dots,v_n)}
= &a_n(u;v_1,\dots, v_n)\psi_{(v_1,\dots, v_n)} \\
&+\sum_{j=1}^n a_{n,j}(u;v_1,\dots, v_n)\psi_{(v_1,\dots, v_{j-1}, u, v_{j+1}, \dots, v_n)}
\nonumber
\end{align}
and
\begin{align}\label{e:Duonphi}
\hat{\mathcal D}_m(u;a) \psi_{(v_1,\dots,v_n)}
=&d_n(u;v_1,\dots, v_n)\psi_{(v_1,\dots, v_n)} \\
&+\sum_{j=1}^n d_{n,j}(u;v_1,\dots, v_n)\psi_{(v_1,\dots, v_{j-1}, u, v_{j+1}, \dots, v_n)} ,
\nonumber
\end{align}
\end{subequations}
with
\begin{subequations}
\begin{align}
 a_n(u;v_1,\dots, v_n)  &= \alpha_m(u) \prod_{k=1}^n f_1(u,v_k), \\
d_n(u;v_1,\dots, v_n)  &= \delta_m (u) \prod_{k=1}^n g_1(u,v_k) ,
\end{align}
\end{subequations}
and
\begin{subequations}
\begin{align}
a_{n,j}& (u;v_1,\dots, v_n) =\\
&  \alpha_m(v_j) f_2(u,v_j) \prod_{\substack{1\le k \le n\\k\neq j}}f_1(v_j,v_k)
+ \delta_m(v_j) f_3(u,v_j) \prod_{\substack{1\le k \le n\\ k\neq  j}}g_1(v_j,v_k) ,\nonumber \\
d_{n,j} & (u;v_1,\dots, v_n) = \\
&   \delta_m(v_j) g_2(u,v_j) \prod_{\substack{1\le k\le n\\ k\neq j}}g_1(v_j,v_k) 
+ \alpha_m(v_j) g_3(u,v_j) \prod_{\substack{1\le k \le n\\ k\neq j}}f_1(v_j,v_k)
. \nonumber
\end{align}
\end{subequations}
Here $\alpha_m$ and $\delta_m$ denote the Laurent polynomials defined in Lemma \ref{U-periodic-action:lem}, while $f_r$ and $g_r$ ($r=1,2,3$) refer to the rational functions defined in
Lemma \ref{boundary-monodromy-relations-mod:lem}.
\end{proposition}
\begin{proof}

For $n=0$ the proposition reduces to Lemma \ref{U-boundary-action:lem}.
To deal with the case $n>0$, we proceed by induction. From the relations in Lemma \ref{boundary-monodromy-relations-mod:lem} it is seen that
\begin{align*}
\mathcal A_m(u;a) \psi_{(v_1,\dots, v_n)} \equiv & \mathcal A_m(u;a) \mathcal B(v_n;a)\psi_{(v_1,\dots, v_{n-1})}   \\
 =&
 f_1(u,v_n) \mathcal B_m(v_n;a) \mathcal A_m(u;a) \psi_{(v_1,\dots, v_{n-1})} 
\\
&+f_2(u,v_n) \mathcal B_m(u;a) \mathcal A_m(v_n;a)\psi_{(v_1,\dots, v_{n-1})} 
\\
&+f_3(u,v_n)  \mathcal B_m(u;a) \hat{\mathcal D}_m(v_n;a)
\psi_{(v_1,\dots, v_{n-1})} 
\end{align*}
and that
\begin{align*}
\hat{\mathcal D}_m(u;a) \psi_{(v_1,\dots, v_n)} \equiv &\hat{ \mathcal D}_m(u;a) \mathcal B(v_n;a)\psi_{(v_1,\dots, v_{n-1})}   \\
 =&
g_1(u,v_n) \mathcal B_m(v_n;a) \hat{\mathcal D}_m(u;a) \psi_{(v_1,\dots, v_{n-1})} 
\\
&+g_2(u,v_n) \mathcal B_m(u;a) \hat{\mathcal D}_m(v_n;a)\psi_{(v_1,\dots, v_{n-1})} 
\\
&+g_3(u,v_n)  \mathcal B_m(u;a) {\mathcal A}_m(v_n;a)
\psi_{(v_1,\dots, v_{n-1})} .
\end{align*}
When working out the actions of the operators on the RHS with the aid of the induction hypothesis, we end up with an expansion for the LHS of the form in Eqs. \eqref{e:Auonphi}, \eqref{e:Duonphi}, respectively, with
\begin{align*}
a_{n}(u;v_1,\ldots ,v_{n})&=
f_1(u,v_n) a_{n-1}(u;v_1,\ldots ,v_{n-1}),\\
d_{n}(u;v_1,\ldots ,v_{n})& =
g_1(u,v_n) d_{n-1}(u;v_1,\ldots ,v_{n-1}) ,
\end{align*}
and
\begin{multline*}
a_{n,j} (u;v_1,\dots, v_n) =\\
 f_1(u,v_n) 
\Bigl(
   \alpha_m(v_j) f_2(u,v_j)  \prod_{\substack{1\le k < n\\ k\neq j}} f_1(v_j,v_k)
   + \delta_m (v_j) f_3(u,v_j) \prod_{\substack{1\le k<  n\\ k\neq j}} g_1(v_j,v_k)
\Bigr)\\
+  f_2(u,v_n) 
\Bigl(
   \alpha_m(v_j) f_2(v_n,v_j)  \prod_{\substack{1\le k < n\\ k\neq j}} f_1(v_j,v_k)
   + \delta_m(v_j) f_3(v_n,v_j) \prod_{\substack{1\le k < n\\ k\neq j}} g_1(v_j,v_k)
\Bigr)\\
+  f_3(u,v_n) 
\Bigl(
   \alpha_m(v_j) g_3(v_n,v_j)  \prod_{\substack{1\le k < n\\ k\neq j}} f_1(v_j,v_k)
   + \delta_m(v_j) g_2(v_n,v_j) \prod_{\substack{1\le k < n\\ k\neq j}} g_1(v_j,v_k)
\Bigr) ,
\end{multline*}
\begin{multline*}
d_{n,j} (u;v_1,\dots, v_n) =\\
 g_1(u,v_n) 
\Bigl(
   \delta_m(v_j) g_2(u,v_j)  \prod_{\substack{1\le k < n\\ k\neq j}} g_1(v_j,v_k)
   + \alpha_m (v_j) g_3(u,v_j) \prod_{\substack{1\le k<  n\\ k\neq j}} f_1(v_j,v_k)
\Bigr)\\
+  g_2(u,v_n) 
\Bigl(
   \delta_m(v_j) g_2(v_n,v_j)  \prod_{\substack{1\le k < n\\ k\neq j}} g_1(v_j,v_k)
   + \alpha_m(v_j) g_3(v_n,v_j) \prod_{\substack{1\le k < n\\ k\neq j}} f_1(v_j,v_k)
\Bigr)\\
+  g_3(u,v_n) 
\Bigl(
   \delta_m(v_j) f_3(v_n,v_j)  \prod_{\substack{1\le k < n\\ k\neq j}} g_1(v_j,v_k)
   + \alpha_m(v_j) f_2(v_n,v_j) \prod_{\substack{1\le k < n\\ k\neq j}} f_1(v_j,v_k)
\Bigr) ,
\end{multline*}
for $j=1,\ldots ,n$. The proposition now follows by virtue of the following elementary rational identities:
$$
f_1(u,v) f_2(u,w) + f_2(u,v) f_2(v,w) + f_3(u,v) g_3(v,w)
=f_2(u,w) f_1(w,v),
$$
$$
f_1(u,v) f_3(u,w) + f_2(u,v) f_3(v,w) + f_3(u,v) g_2(v,w)
= f_3(u,w) g_1(w,v)
$$
and
$$
g_1(u,v) g_2(u,w) + g_2(u,v) g_2(v,w) + g_3(u,v) f_3(v,w)
=g_2(u,w) g_1(w,v),
$$
$$
g_1(u,v) g_3(u,w) + g_2(u,v) g_3(v,w) + g_3(u,v) f_2(v,w)
= g_3(u,w) f_1(w,v) .
$$ 
\end{proof}

Let us next rewrite the boundary transfer operator
$\mathcal T_m(u;a_+,a_-)$ \eqref{T-boundary} in terms of the modified matrix element
$\hat{\mathcal{D}}_m(u;a)$  \eqref{Dmhat}:
\begin{equation}\label{T-boundary-alternative}
\mathcal T_m(u;a_+,a_-)
=\frac{s(q^{-1}u^2)}{s(u^2)}e(u;a_+) \mathcal A_m(u;a_-) +  e(u^{-1};a_+) \hat{\mathcal D}_m(u;a_-) .
\end{equation}
We temporarily assume that the domain restrictions and the genericity assumptions of Theorem \ref{BAE:thm} and Proposition \ref{U-boundary-action-general:prp} are simultaneously satisfied.
From the proposition
it is then immediate
that $\psi_{(v_1,\ldots ,v_n)}$  (with $a=a_-$)---and thus also $\Psi_{(v_1,\ldots ,v_n)}$ \eqref{ABA-wave-function}---constitutes
an eigenfunction of  $\mathcal T_m(u;a_+,a_-)$ with eigenvalue
\begin{subequations}
\begin{align}
 E_{n,m} & (u;v_1,\ldots,v_n)
= \\
&\frac{s(q^{-1}u^2)}{s(u^2)}e(u;a_+) a_n(u;v_1,\ldots ,v_n) + e(u^{-1};a_+)  d_n(u;v_1,\ldots,v_n) 
\nonumber
\end{align}
provided
\begin{equation}
 \frac{s(q^{-1}u^2)}{s(u^2)}e(u;a_+)   a_{n,j}(u;v_1,\ldots,v_n) +e(u^{-1};a_+) d_{n,j} (u;v_1,\ldots,v_n)=0,
\end{equation}
\end{subequations}
for $ j=1,\ldots, n$. The expressions for the eigenvalue and the Bethe Ansatz equations formulated in the theorem are now readily recovered upon making $a_n$, $a_{n,j}$ and $d_n$, $d_{n,j}$ explicit (by means of the definitions detailed in Proposition \ref{U-boundary-action-general:prp}).
Finally, we relax our genericity assumptions to those of Theorem \ref{BAE:thm} by recalling that
 $\mathcal T_m(u;a_+,a_-)$ \eqref{T-boundary} and
 $\Psi_{(v_1,\ldots ,v_n)}$ \eqref{ABA-wave-function} are Laurent polynomials in $u$ and
 $v_1,\ldots ,v_n$, respectively.


\section{Branching rule}\label{sec9}
In this section, we employ the Bethe Ansatz creation operator $\hat{\mathcal{B}}_m(u;a)$ \eqref{ABA-creation-operator} to derive a branching rule that permits to compute the $n$-particle Bethe Ansatz wave function
$\Psi_{(v_1,\ldots ,v_n)}$ \eqref{ABA-wave-function} inductively in the number of particles/variables. This explicit construction confirms  that our
Bethe Ansatz wave function is not identically zero on $\Lambda_{n,m}$ \eqref{dominant}, and as such constitutes a genuine nontrivial eigenfunction of the $q$-boson boundary transfer operator in $\mathcal{F}_{n,m}$ (if the spectral variables satisfy the Bethe Ansatz equations).

\subsection{Bethe Ansatz wave function}
By evaluating the action of the Bethe-Ansatz creation operator on the vacuum state, one arrives at a closed expression for the one-particle Bethe Ansatz wave function $\Psi_v=\hat{\mathcal{B}}_m(v;a_-)|\emptyset\rangle$  in terms of Chebyshev polynomials.

\begin{proposition}[One-Particle Bethe Ansatz Wave Function]\label{BAF-n=1:prp}
For $v\in\mathbb{C}^*$, the value of the one-particle Bethe Ansatz wave function $\Psi_v=\hat{\mathcal{B}}_m(v;a_-)|\emptyset\rangle$ at $l\in\mathbb{N}_m$ \eqref{Nm} is given by
\begin{equation}
\Psi_{v}(l)= s_l(v^2)-a_-s_{l-1}(v^2)\quad\text{with}\quad  s_l(z):= \frac{z^{l+1}-z^{-l-1}}{z-z^{-1}}
\end{equation}
 (where $s_{-1}(z)=0$ by convention).
\end{proposition}

Before verifying this explicit formula for the one-particle Bethe Ansatz wave function below, let us first elucidate how one can iteratively augment the number of particles by means of a branching rule originating from the Bethe Ansatz creation operator.

\begin{theorem}[Branching Rule] \label{BAF-branching-rule:thm}
For $n>1$ and $(v_1,\ldots ,v_n)\in(\mathbb{C}^*)^n$ generic such that $tv_n^4\neq 1$, 
the value of the $n$-particle Bethe Ansatz wave function $\Psi_{(v_1,\ldots ,v_n)}$ \eqref{ABA-wave-function} at $\lambda\in\Lambda_{n,m}$
is determined by the following branching rule:
\begin{subequations}
\begin{equation}\label{BAF-branching-rule1}
\Psi_{(v_1,\ldots ,v_n)}(\lambda)=   \sum_{\substack{\mu\in \Lambda_{n-1,m} \\ \mu \leq \lambda}}
\hat{\emph{B}}^{(n)}_{\lambda/\mu}(v_n^2 ;t,a_-) \Psi_{(v_1,\ldots ,v_{n-1})}(\mu),
\end{equation}
with
 \begin{align}\label{BAF-branching-rule2}
\hat{\emph{B}}^{(n)}_{\lambda/\mu}(z;t,a)& :=
 \frac{(z^{-1}-a)}{(1-t)(z^{-1}-tz)}
\sum_{\substack{ \nu\in\Lambda_{n,m}\\\mu\preceq\nu\preceq\lambda  }}
\varphi_{\lambda/\nu} (t)\varphi_{\nu/\mu} (t)   z^{|\lambda|+|\mu|-2|\nu|} \\
+&
 \frac{(a-tz)}{(1-t)(z^{-1}-tz)}
\sum_{\substack{ \nu\in\Lambda_{n-1,m}\\\mu\preceq\nu\preceq\lambda  }}
\varphi_{\lambda/\nu}(t)\varphi_{\nu/\mu} (t)   z^{|\lambda|+|\mu|-2|\nu | } .
 \nonumber
\end{align}
\end{subequations}
\end{theorem}

\begin{proof}
The asserted branching rule follows by computing the explicit action of the Bethe Ansatz creation operator $\hat{\mathcal{B}}_m(u;a)$ \eqref{ABA-creation-operator} on the Bethe Ansatz wave function in $\mathcal{F}_{n-1,m}$
with the aid of Proposition \ref{boundary-action:prp}:
\begin{align*}
\Psi_{(v_1,\dots,v_n)}(\lambda)
=&\Bigl(\hat{\mathcal B}_m(v_n;a_-) \Psi_{(v_1,\dots,v_{n-1})}\Bigr)(\lambda)\\
=&\sum_{\substack{\mu\in \Lambda_{n-1,m} \\  \mu \le \lambda}}  \hat{\text{B}}^{(n)}_{\lambda/\mu}(v_n^2;t;a_-)
\Psi_{(v_1,\dots,v_{n-1})}(\mu)
\end{align*}
($\lambda\in\Lambda_{n,m}$, $n>1$), with $\hat{\text{B}}^{(n)}_{\lambda/\mu}(u^2;t,a)=\text{B}^{(n-1)}_{\lambda,\mu}(u^2;t,a) b(u)^{-1}q^{-1}$ (cf. Eqs. \eqref{B-boundary-action}, \eqref{B-boundary-coef} and \eqref{ABA-creation-operator}).
\end{proof}

By iterating the above branching rule, one ends up with a closed expression for the $n$-particle Bethe Ansatz wave function in terms
one-particle wave functions taken from Proposition \ref{BAF-n=1:prp}.
\begin{corollary}[$n$-Particle Bethe Ansatz Wave Function]\label{BAF-n:cor}
For $n\geq 1$ and $(v_1,\ldots ,v_n)\in(\mathbb{C}^*)^n$ generic such that $tv_j^4\neq 1$ ($1< j\leq n$), one has that at $\lambda\in\Lambda_{n,m}$
\begin{align}
&\Psi_{(v_1,\ldots ,v_n)}(\lambda)= \\
&  \sum_{\substack{\mu^{(j)}\in \Lambda_{j,m},\, j=1,\ldots ,n \\ \mu^{(1)} \leq \mu^{(2)}\leq \cdots \leq\mu^{(n)}=\lambda}}
\Psi_{v_1}(\mu^{(1)})
\prod_{1<j\leq n}\hat{\emph{B}}^{(j)}_{\mu^{(j)}/\mu^{(j-1)}}(v_j^2;t,a_-) 
. \nonumber
\end{align}
\end{corollary}

The value for the $n$-particle Bethe Ansatz wave function becomes particularly simple at the origin.
\begin{corollary}[Evaluation at the Origin]\label{BAF-at-origin:cor} For $n\geq 1$ and $(v_1,\ldots ,v_n)\in (\mathbb{C}^*)^n$, the value of the $n$-particle Bethe Ansatz wave function at $\lambda=0^n:=(\underbrace{0,\ldots ,0}_{n\ \text{times}})\in\Lambda_{n,m}$
is given by
\begin{equation}\label{BAF-at-origin}
\Psi_{(v_1,\ldots ,v_n)}(0^n)= [n]! .
\end{equation}
\end{corollary}
\begin{proof} 
For $n=1$, Eq. \eqref{BAF-at-origin} is immediate from Proposition \ref{BAF-n=1:prp}. For $n>1$, we deduce via the branching rule in Eqs. \eqref{BAF-branching-rule1}, \eqref{BAF-branching-rule2} that
$$
\Psi_{(v_1,\ldots ,v_n)}(0^n)= \hat{\text{B}}^{(n)}_{0^n/0^{n-1}}(v_n^2;t,a_-)\Psi_{(v_1,\ldots ,v_{n-1})}(0^{n-1})
$$
with
$$
\hat{\text{B}}^{(n)}_{0^n/0^{n-1}}(v_n^2 ;t,a)=\frac{(v_n^{-2}-a)(1-t^n)}{(1-t)(v_n^{-2}-tv_n^2)}+
 \frac{(a-tv_n^{2})(1-t^n)}{(1-t)(v_n^{-2}-tv_n^2)} = \frac{1-t^n}{1-t}=[n]
$$
(since $\varphi_{0^n/0^{n-1}} (t) =1-t^n$), whence $\Psi_{(v_1,\ldots ,v_n)}(0^n)=  [n][n-1]\cdots [2][1]=[n]!$.
\end{proof}
The latter corollary reveals
in particular that  $\Psi_{(v_1,\ldots ,v_n)}$ \eqref{ABA-wave-function} is a nontrivial $n$-particle wave function in $\mathcal{F}_{n,m}$, which confirms that the algebraic Bethe Ansatz formalism provides us 
in Theorem \ref{BAE:thm} with a genuine (nonzero) eigenfunction of the $q$-boson boundary transfer operator $\mathcal{T}_m(u;a_+,a_-)$ (upon solving the Bethe Ansatz equations).
In the remainder of this section, the formula for the one-particle wave function in Proposition \ref{BAF-n=1:prp} is verified.

\subsection{Proof of Proposition \ref{BAF-n=1:prp}}
First we compute the action of the Bethe Ansatz creation operator on the vacuum state in the periodic setting.
\begin{lemma}\label{B-periodic-action:lem}
For $u\in\mathbb{C}^*$ and $l\in\mathbb{N}_m$, one has that
\begin{equation}
(B_m(u)|\emptyset\rangle) (l)=(1-t)u^{2l-m}.
\end{equation}
\end{lemma}
\begin{proof}
The case $m=0$ is immediate from the Lax matrix \eqref{L-op} and the action of the $q$-boson creation operators $\beta^*_l$, $l=0,\ldots ,m$ on the vacuum state $|\emptyset\rangle$ (cf. 
Remark \ref{vacuum:rem}). The case of general $m>0$ then follows inductively via the recurrence
in Eq. \eqref{e:Am+1} and Lemma \ref{U-periodic-action:lem}.
\end{proof}
The computation of the action of the Bethe Ansatz creation operator $\hat{\mathcal{B}}_m(u;a)$  \eqref{ABA-creation-operator} on the vacuum state now hinges on Lemma \ref{U-boundary-matrix-elements:lem}.
Indeed, it is seen from Lemmas \ref{U-periodic-action:lem} and \ref{B-periodic-action:lem}  that for
$u\in\mathbb{C}^*$ and  $l\in\mathbb{N}_m$:
\begin{subequations}
\begin{equation}\label{BA-action-vacuuma}
(B_m(u)A_m(u^{-1})|\emptyset\rangle)(l) =(1-t)u^{2l+1}
\end{equation}
and
\begin{equation}\label{BA-action-vacuumb}
(A_m(u)B_m(u^{-1})|\emptyset\rangle)(l)=(1-t)u^{-2l-1}+(1-t)^2us_{l-1}(u^2),
\end{equation}
\end{subequations}
where in the latter case we relied on Eq \eqref{Rc} of Lemma \ref{periodic-monodromy-relations:lem} in the form of the identity
$$
A_m(u)B_m(u^{-1})= \frac{qs(q)}{s(u^{-2})} B_m(u)A_m(u^{-1})+\frac{qs(qu^2)}{s(u^2)} B_m(u^{-1})A_m(u)
$$
to reduce it to the former case. From Eqs. \eqref{BA-action-vacuuma}, \eqref{BA-action-vacuumb} and Lemma \ref{U-boundary-matrix-elements:lem}, it then follows that
$
(\mathcal{B}_m(u;a)|\emptyset\rangle )(l)= q^{-m-1}b(u) \left(  s_l(u^2)-as_{l-1}(u^2)\right)
$,
whence
\begin{equation}
(\hat{\mathcal{B}}_m(u;a)|\emptyset\rangle )(l)=   s_l(u^2)-a s_{l-1}(u^2) .
\end{equation}


\section{Orthogonality and completeness}\label{sec10}
In this section we first implement the method of Yang and Yang \cite{yan-yan:thermodynamics}
providing the solutions of the Bethe Ansatz equations via the minima of a family of strictly convex Morse functions.
Next, the self-adjointness of the boundary transfer operator  is exploited to deduce that the corresponding
Bethe Ansatz eigenfunctions at these spectral values constitute an orthogonal basis for the $q$-boson Fock space.

\subsection{Solutions of the Bethe Ansatz equations}
We seek for solutions of the Bethe Ansatz equations \eqref{e:BAE} corresponding to spectral variables $v_1,\ldots ,v_n$ placed on
the unit circle. This ensures in particular that all genericity conditions in Theorem \ref{BAE:thm} and Corollary \ref{H-EV:cor} are automatically satisfied for our parameter domain (cf. Eq. \eqref{parameter-domain}). Specifically, for
$v_j=e^{i\xi_j/2}$ ($j=1,\ldots ,n$)  with $\xi=(\xi_1,\ldots ,\xi_n)\in\mathbb{R}^n$ the Bethe Ansatz equations in Theorem \ref{BAE:thm} are rewritten as:
\begin{align}\label{BAE:eqs}
&e^{2i(m+1)\xi_j}=\\
&\left( \frac{1-a_+ e^{i\xi_j }}{e^{i\xi_j}-a_+}\right)
\left( \frac{1-a_-e^{i\xi_j }}{e^{i\xi_j}-a_-}\right)
\prod_{\substack{1\le k\le n   \\ k\neq j}}
\left(
\frac{1-te^{i(\xi_j + \xi_k) }}{e^{i(\xi_j+\xi_k)}-t}\right)\left( 
\frac{1-te^{i(\xi_j - \xi_k) }}{e^{i(\xi_j-\xi_k)}-t}
\right),\nonumber
\end{align}
$j=1,\ldots ,n$. Analogous Bethe Ansatz equations for the periodic $q$-boson system
were considered in Refs. \cite{bog-ize-kit:correlation,tsi:quantum,die:diagonalization,kor:cylindric},  and more generally in Ref. \cite{die-ems:discrete}  where generalizations of the $q$-boson system arising from double affine Hecke algebras at critical level were studied.

For any $\lambda\in\Lambda_{n,m}$, 
 the  (unique) global minimum  $\xi^{(n,m)}_\lambda\in\mathbb{R}^n$ of the strictly convex Morse function
 $V^{(n,m)}_\lambda:\mathbb{R}^n\to\mathbb{R}$ given by
 \begin{subequations}
\begin{multline}\label{q-boson-morse-a}
V^{(n,m)}_\lambda(\xi):= \sum_{1\leq j\leq n} \left(
(m+1)\xi_j^2-2\pi (\rho_j+\lambda_j)\xi_j
+  \int^{\xi_j}_0 \bigl(v_{a_+}(u)+v_{a_-}(u)\bigr)\text{d}u\right) \\
+\sum_{1\le j < k \le n }
\left(
    \int_0^{\xi_j+\xi_k} v_t(u)\text{d}u+   \int_0^{\xi_j-\xi_k} v_t(u)\text{d}u
\right) ,
\end{multline}
with $\rho_j=n+1-j$ ($j=1,\ldots ,n$) and 
\begin{equation}\label{q-boson-morse-b}
v_a(\theta ) := 
\int_0^\theta \frac{ (1-a^2)\ \text{d}u}{1-2a\cos(u)+a^2}
=
i \log
\biggl( \frac{1- ae^{i\theta}}{e^{i\theta} - a }  \biggr) \qquad (-1<a<1),
\end{equation}
\end{subequations}
provides a solution of the Bethe Ansatz equations in Eq. \eqref{BAE:eqs}.
Here the branches of the logarithmic function are assumed to be chosen such that $v_a(\theta)$ is quasi-periodic
($v_a(\theta+2\pi) = v_a(\theta)+2\pi$) and varies from $-\pi$ to $\pi$ as $\theta$ varies from $-\pi$ to $\pi$
(this corresponds to the principal branch). Moreover, the parameter constraint $a\in (-1,1)$ guarantees that the odd function $v_a(\theta)$ is smooth and strictly monotonously increasing on $\mathbb{R}$. 
Notice that the equation $\nabla_\xi V^{(n,m)}_\lambda (\xi)=0$ characterizing the critical point $\xi^{(n,m)}_\lambda$ is given by the system
\begin{equation}\label{q-boson-critical-eq}
2(m+1)\xi_j+v_{a_+}(\xi_j)+v_{a_-}(\xi_j)+\sum_{\substack{1\leq k\leq n\\k\neq j}} \Bigl( v_t(\xi_k+\xi_j)- v_t(\xi_k-\xi_j)\Bigr)=2\pi (\rho_j+\lambda_j) ,
\end{equation}
$j=1,\ldots,n$.

\begin{proposition}[Spectral Points]\label{q-boson-spectrum:prp}
\begin{subequations}
For $\lambda\in\Lambda_{n,m}$ \eqref{dominant}, the unique global minimum
$\xi^{(n,m)}_\lambda\in\mathbb{R}^n$ of
the strictly convex Morse function
 $V^{(n,m)}_\lambda (\xi)$ \eqref{q-boson-morse-a}, \eqref{q-boson-morse-b}
 enjoys the following properties.
\begin{itemize}
\item[i)] At the spectral point $\xi=\xi^{(n,m)}_\lambda$, the Bethe Ansatz equations  \eqref{BAE:eqs} are satisfied.
\item[ii)] The mapping $\lambda\mapsto \xi^{(n,m)}_\lambda$  is injective.
\item [iii)]The minimum $\xi^{(n,m)}_\lambda$ is assumed inside the alcove $2\pi \emph{A}$ (where $\emph{A}$ refers to the hyperoctahedral Weyl alcove in Eq. \eqref{A}).
\item[iv)] At the spectral point $\xi=\xi^{(n,m)}_\lambda$, the following estimates hold for $\xi_j$:
\begin{equation}\label{e:momentgaps1}
\frac{\pi(\rho_j+\lambda_j)}{m+1+\kappa_-} < \xi_j < \frac{\pi(\rho_j+\lambda_j)}{m+1+\kappa_+}
\end{equation}
(for $1\leq j\leq n$),  and for the moment gaps $\xi_j-\xi_k$:
\begin{equation}\label{e:momentgaps2}
\frac{\pi(\rho_j-\rho_k+\lambda_j-\lambda_k)}{m+1+\kappa_-} < \xi_j-\xi_k< \frac{\pi(\rho_j-\rho_k+\lambda_j-\lambda_k)}{m+1+\kappa_+} 
\end{equation}
(for $1 \le  j < k \le n$), where
\begin{align}
\kappa_{\pm}&=\kappa_{\pm}(t,a_+,a_-)\\
&:=\frac{1}{2}\left( \frac{1-a_+^2}{(1\pm |a_+|)^2}+ \frac{1-a_-^2}{(1\pm |a_-|)^2} \right)+  \frac{(n-1)(1-t^2)}{(1\pm t)^2}.
\nonumber
\end{align}
\item[v)] The position of $\xi^{(n,m)}_\lambda$ depends analytically on the parameters $t,a_+,a_-$ (taken from the domain in Eq. \eqref{parameter-domain}), and one 
has that
\begin{equation}\label{phase-model-limit}
\lim_{t,a_+,a_-\to 0} \xi^{(n,m)}_\lambda=   \frac{\pi (\rho+\lambda)}{m+n+1}
\end{equation}
(where $\rho=(\rho_1,\ldots ,\rho_n)$).
\end{itemize}
\end{subequations}
\end{proposition}
Analogs of this proposition for the periodic $q$-boson system and its generalization associated with the double affine Hecke algebra at critical level were previously formulated in \cite[Sec. 4]{die:diagonalization} and \cite[Prp. 7.3]{die-ems:discrete}, respectively. To keep the presentation self-contained, the main ingredients of the proof are briefly adapted to the present setup in Section \ref{q-boson-spectrum:proof} below.

\subsection{Diagonalization of the open-end $q$-boson system}
Let us abbreviate the notation for $n$-particle Bethe Ansatz wave function with spectral parameters on the unit circle:
\begin{equation}\label{ABA-wf-unitary}
\Psi^{(n,m)}(\xi,\mu) :=\Psi_{(\exp({\frac{i\xi_1}{2}}),\dots ,\exp({\frac{i\xi_n}{2}}))}(\mu)\qquad (\xi\in\mathbb{R}^n,\ 
\mu\in\Lambda_{n,m}).
\end{equation}
Our spectral analysis of the open-end $q$-boson system---by means of Sklyanin's version of the algebraic Bethe Ansatz formalism---now culminates in the following principal conclusion, which confirms that the Bethe Ansatz eigenfunctions $\Psi^{(n,m)}(\xi^{(n,m)}_\lambda):=\Psi^{(n,m)}(\xi^{(n,m)}_\lambda,\cdot )$, $\lambda\in\Lambda_{n,m}$ constitute an orthogonal basis
for the $n$-particle sector $\mathcal{F}_{n,m}$ \eqref{n-particle-sector}--\eqref{weight-function} of the $q$-boson Fock space $\mathcal{F}_m$ \eqref{q-boson-fock-space}, \eqref{fock-ip}.

\begin{theorem}[Diagonalization and Orthogonality]\label{q-boson-eigenfunctions:thm}
Let $u\in\mathbb{R}^*\setminus\{ 1,-1\}$.
 The $n$-particle Bethe Ansatz wave functions $\Psi^{(n,m)}(\xi^{(n,m)}_\lambda)$, $\lambda\in\Lambda_{n,m}$, which satisfy the eigenvalue
 equations
 \begin{subequations}
 \begin{equation}
 \mathcal{T}_m(u;a_+,a_-) \Psi^{(n,m)}(\xi^{(n,m)}_\lambda)= E^{(n,m)}(u;\xi^{(n,m)}_\lambda) \Psi^{(n,m)}(\xi^{(n,m)}_\lambda)
 \end{equation}
 with
 \begin{align}\label{Enm:eigenvalues}
&E^{(n,m)}(u;\xi) := q^{-m}t^{-n}\times\\
& \biggl( u^{-2(m+2)} \frac{(1- t^{-1}u^4)}{(1-u^4)}(1-a_+u^2)(1-a_-u^2)
     \prod_{1\leq j\leq n} \frac{(1-tu^2e^{i\xi_j})(1-tu^2e^{-i\xi_j})}{(1-u^2e^{i\xi_j})(1-u^2e^{-i\xi_j})}           +
     \nonumber  \\
  &  u^{2(m+2)}   \frac{(1-t^{-1}u^{-4})}{(1-u^{-4})}(1-a_+u^{-2})(1-a_-u^{-2})
     \prod_{1\leq j\leq n} \frac{(1-tu^{-2}e^{i\xi_j})(1-tu^{-2}e^{-i\xi_j})}{(1-u^{-2}e^{i\xi_j})(1-u^{-2}e^{-i\xi_j})}    \biggr)      , \nonumber
\end{align}
 \end{subequations}
and
\begin{subequations}
 \begin{equation}
 \mathcal{H}_m \Psi^{(n,m)}(\xi^{(n,m)}_\lambda)= E^{(n)}(\xi^{(n,m)}_\lambda) \Psi^{(n,m)}(\xi^{(n,m)}_\lambda)
 \end{equation}
with
\begin{equation}
E^{(n)}(\xi) :=\sum_{1\leq j\leq n} 2\cos(\xi_j),
\end{equation}
\end{subequations}
constitute an orthogonal basis for $\mathcal{F}_{n,m}$:
\begin{equation}
\Bigl( \Psi^{(n,m)}(\xi^{(n,m)}_\lambda), \Psi^{(n,m)}(\xi^{(n,m)}_\mu) \Bigr)_{n,m}   = 0 \quad \text{iff}\ \lambda\neq\mu 
\end{equation}
($\lambda,\mu\in\Lambda_{n,m}$).
\end{theorem}

It is clear from Theorem \ref{BAE:thm}, Corollary \ref{H-EV:cor}, Corollary \ref{BAF-at-origin:cor} and Proposition \ref{q-boson-spectrum:prp}, that
the $n$-particle Bethe-Ansatz wave 
function $\Psi^{(n,m)}(\xi ,\mu)$ \eqref{ABA-wf-unitary} provides a nonvanishing eigenfunction
of the $q$-boson boundary transfer operator and of the $q$-boson Hamilonian in the $n$-particle sector
$\mathcal{F}_{n,m}$, at any spectral point $\xi=\xi^{(n,m)}_\lambda$, $\lambda\in\Lambda_{n,m}$.
In view of Proposition \ref{self-adjoint:prp}, for confirming that the eigenfunctions in question are orthogonal
it is enough to verify that the corresponding eigenvalues of the boundary transfer operator are simple (as Laurent polynomials in $u$), i.e. that for any $\lambda,\mu\in\Lambda_{n,m}$:
\begin{equation}\label{simple}
E^{(n,m)}(u;\xi_\lambda^{(n,m)})\neq E^{(n,m)}(u;\xi^{(n,m)}_\mu)\quad\text{if}\quad \lambda\neq \mu . 
\end{equation}
The verification of this nondegeneracy of the $n$-particle Bethe Ansatz eigenvalues is relegated to Section \ref{q-boson-orthogonality:proof} below. 

From Theorem \ref{q-boson-eigenfunctions:thm}, we read-off the $n$-particle spectrum of the boundary transfer operator and  the  Hamiltonian for our open-end $q$-boson system.

\begin{corollary}[Open-End $q$-Boson Spectrum] Let $u\in\mathbb{R}^*\setminus\{ 1,-1\}$. The $n$-particle spectrum of the self-adjoint boundary transfer operator $\mathcal{T}_m(u;a_+,a_-)$ \eqref{T-boundary} and the Hamiltonian \eqref{Hm} in $\mathcal{F}_{n,m}$ is given, respectively, by the eigenvalues
$ E^{(n,m)}(u;\xi^{(n,m)}_\lambda) $ and $E^{(n)}(\xi^{(n,m)}_\lambda)$ (where $\lambda$ runs through $\Lambda_{n,m}$ \eqref{dominant}).
\end{corollary}

\begin{remark} 
Since $E^{(n,m)}(u^{-1};\xi)=E^{(n,m)}(u;\xi)$, it is clear that the commuting quantum integrals $\tau_{m,k}(a_+,a_-)$, $k=-m-2,-m-1,\ldots,m+2$ generated by the boundary transfer operator $\mathcal{T}_m(u;a_+,a_-)$ \eqref{T-boundary} in  Eqs. \eqref{T-boundary-exp}, \eqref{lead-coeff}, satisfy the symmetry
$\tau_{m,-k}(a_+,a_-)=\tau_{m,k}(a_+,a_-)$ (as operators in the $q$-boson Fock space $\mathcal{F}_m$ \eqref{q-boson-fock-space}, \eqref{fock-ip}). In other words, apart from the number operator $\mathcal{N}_m$ \eqref{Nm} and the Hamiltonian $\mathcal{H}_m$ \eqref{Hm} the boundary transfer operator  generates at most $m+1$ independent quantum integrals $\tau_{m,k}(a_+,a_-)$, $k=0,\ldots ,m$, in $\mathcal{F}_m$.
\end{remark}

\subsection{Proof of Proposition \ref{q-boson-spectrum:prp}}\label{q-boson-spectrum:proof}
It is immediate from the boundedness of the integrand characterizing $v_a (\theta)$ that the quadratic terms in the smooth function $V^{(n,m)}_\lambda (\xi)$ \eqref{q-boson-morse-a}, \eqref{q-boson-morse-b}  dominate as the norm of $\xi$ becomes large, whence the function under consideration is bounded from below on $ \mathbb{R}^n$  and thus possesses  a minimum. To confirm that this minimum is unique, we check that the function
$V^{(n,m)}_\lambda (\xi)$ is strictly convex.
Indeed, 
an explicit computation of  
\begin{align}
&H^{(n,m)}_{j,k}:=\partial_{\xi_j}\partial_{\xi_k} V^{(n,m)}_\lambda (\xi) \\
&=
\begin{cases}
2(m+1) + v_{a_+}'(\xi_j) + v_{a_-}'(\xi_j) + \sum_{l\neq j}\bigl (v_t'(\xi_j+\xi_l) +v_t'(\xi_j-\xi_l) \bigr) & \text{if $ k=j$}\\
v_t'(\xi_j+\xi_k) -v_t'(\xi_j-\xi_k) & \text{if $k\neq j$}\\
\end{cases} ,\nonumber
\end{align}
reveals that the Hessian under consideration is positive definite:
\begin{align*}
\sum_{1\leq j,k\leq n}  H^{(n,m)}_{j,k} x_j x_k   
= & \sum_{1\leq j\leq n} \Bigl( 2(m+1)+ v_{a_+}'(\xi_j) + v_{a_-}'(\xi_j)\Bigr) x_j^2 \\
&+  \sum_{1\leq j<k\leq n} \Bigl( v_t'(\xi_j +\xi_k)(x_j+x_k)^2+
   v_t'(\xi_j -\xi_k)(x_j-x_k)^2\Bigr) \\
 \ge & 2(m+1) \sum_{1\leq j\leq n} x_j^2> 0
\end{align*}
(if $(x_1,x_2,\ldots ,x_n)\neq (0,0,\ldots ,0)$), because
$$
v^\prime_a(\theta)= \frac{ 1-a^2}{1-2a\cos(\theta)+a^2} >0\quad \text{for}\ a\in (-1,1).
$$

The properties i)--v) for the unique global minimum  $\xi^{(n,m)}_\lambda$ 
are now verified by exploiting that Eq. \eqref{q-boson-critical-eq} is satisfied at $\xi=\xi^{(n,m)}_\lambda$.

i) That the Bethe-Ansatz equations \eqref{BAE:eqs} hold at $\xi=\xi^{(n,m)}_\lambda$  follows upon multiplication of
Eq.  \eqref{q-boson-critical-eq} by the imaginary unit and exponentiation of both sides (while using that
 $\displaystyle e^{-i v_a(\theta)}=(1-a e^{i\theta})/(e^{i\theta}- a)$).

ii) The injectivity of $\lambda\to \xi^{(n,m)}_\lambda$ is immediate from Eq. \eqref{q-boson-critical-eq} (if the LHS's are the same then RHS's must also be the same).

iii) Since the LHS of  Eq. \eqref{q-boson-critical-eq}  is odd and monotonously increasing in $\xi_j$, it follows that $\xi_j>0$ at $\xi=\xi^{(n,m)}_\lambda$ because $2\pi (\rho_j+\lambda_j)>0$. Furthermore, by subtracting the the $k$th equation from the $j$th equation:
\begin{align*}
2(m+1)(\xi_j-\xi_k)
+(v_{a_+}(\xi_j) - v_{a_+}(\xi_k))+(v_{a_-}(\xi_j) - v_{a_-}(\xi_k))+ 2v_t(\xi_j-\xi_k) \nonumber\\
+ \sum_{\substack{1\le l \le n \\ l\neq j,k}} \Bigl( \bigl(v_t(\xi_l+\xi_j) - v_t(\xi_l+\xi_k)\bigr)+ \bigl(v_t(\xi_l-\xi_k) - v_t(\xi_l-\xi_j)\bigr) \Bigr)\nonumber\\
=2\pi (\rho_j - \rho_k +\lambda_j-\lambda_k), 
\end{align*}
one sees in the same way that  for $j< k$: $\xi_j > \xi_k$ at $\xi=\xi^{(n,m)}_\lambda$. Finally, if $\xi_j\geq \pi$ then the LHS of
Eq. \eqref{q-boson-critical-eq} is larger than $2(m+1)\pi+2(n-1)\pi=2(n+m)\pi$ by the quasi-periodicity of $v_a(\theta)$. This implies that on the RHS $2\pi (\rho_j+\lambda_j)> 2\pi (m+n)$, i.e. $\lambda_j>m+n-\rho_j\geq m$ so $\lambda\ne\Lambda_{n,m}$. The upshot is that $\xi^{(n,m)}_\lambda\in 2\pi \text{A}$ for $\lambda\in\Lambda_{n,m}$.

iv) Since the integrand of $v_a(\theta)$ oscillates  between
$\frac{1-a^2}{(1+|a|)^2}$ and $ \frac{1-a^2}{(1-|a|)^2}$,  it follows from Eq. \eqref{q-boson-critical-eq} that
$$(m+1+\kappa_+)\xi_j < \pi (\rho_j+\lambda_j)< (m+1+\kappa_-)\xi_j$$ at $\xi=\xi^{(n,m)}_\lambda$. In the same way, one deduces from
the equation in the proof of part iii) that for $j<k$:  $$(m+1+\kappa_+)(\xi_j -\xi_k)< \pi (\rho_j-\rho_k+\lambda_j-\lambda_k)< (m+1+\kappa_-)(\xi_j-\xi_k)$$ at $\xi=\xi^{(n,m)}_\lambda$.

v)  Since $ V_\lambda^{(n,m)}(\xi)$ is strictly convex, the Jacobian with respect to $\xi$ of the equation $\nabla_\xi V_\lambda^{(n,m)}(\xi)=0$ in
Eq. \eqref{q-boson-critical-eq} is nonzero (viz. positive). Invoking of the implicit function theorem then entails that the critical point $\xi^{(n,m)}_\lambda$ solving Eq. \eqref{q-boson-critical-eq} inherits the analytic dependence on the parameters $t, a_+, a_-\in (-1,1)$ from this equation. The limit \eqref{phase-model-limit} is now immediate from the estimate in Eq. \eqref{e:momentgaps1}.

\subsection{Proof of Theorem \ref{q-boson-eigenfunctions:thm}}\label{q-boson-orthogonality:proof}
Let us assume that $\lambda,\mu$ are two \emph{not necessarily distinct} partitions belonging to $\Lambda_{n,m}$ \eqref{dominant}.
Following Dorlas' approach in \cite[Lem. 4.2]{dor:orthogonality}, we will check
the nondegeneracy of the Bethe Ansatz spectrum with the aid of the Casoratian (or $q$-difference Wronskian)
\begin{subequations}
\begin{equation}\label{Wa}
W(u):=F(qu)G(u)-F(u)G(qu)
\end{equation}
of the associated spectral polynomials
\begin{align}
F(u)&:=\prod_{1\leq j\leq n} (u^2-v_j^2)(u^2-v_j^{-2}),\label{Wb}\\
G(u)&:=\prod_{1\leq j\leq n} (u^2-w_j^2)(u^2-w_j^{-2}) .\label{Wc}
\end{align}
\end{subequations}
Here (and throughout the rest of the proof below) the vectors $(v_1,\ldots,v_n)$ and $(w_1,\ldots,w_n)$ in $(\mathbb{C}^*)^n$ are equal to $(e^{i\xi_1/2},\ldots,e^{i\xi_n/2})$ with $\xi\in 2\pi\text{A}$ taken at
$\xi=\xi^{(n,m)}_\lambda$ and $\xi=\xi^{(n,m)}_\mu$, respectively (cf. Part iii) of Proposition \ref{q-boson-spectrum:prp}).

\begin{lemma}\label{quasi-periodic:lem} If $E^{(n,m)}(u;\xi_\lambda^{(n,m)})= E^{(n,m)}(u;\xi^{(n,m)}_\mu)$  as an identity between Laurent polynomials  in $u$, then the associated Casorati polynomial $W(u)$ \eqref{Wa}--\eqref{Wc} is quasi-periodic:
\begin{subequations}
\begin{equation}\label{W-qp-a}
W(q^{-1}u) = t^{-2n} \frac{\hat{\alpha}_m(u)}{\hat{\alpha}_m(u^{-1})} W(u) 
\end{equation}
with
\begin{equation}\label{W-qp-b}
\hat{\alpha}_m(u):=u^{-2(m+2)}\left( \frac{1-t^{-1}u^4}{1-u^4}\right) (1-a_+u^2)(1-a_-u^2) .
\end{equation}
\end{subequations}
\end{lemma}
\begin{proof}
Because
\begin{equation*}
q^mt^n F(u) E^{(n,m)}(u;\xi^{(n,m)}_\lambda)  = \hat{\alpha}_m(u) F(qu) + t^{2n} \hat{\alpha}_m(u^{-1}) F(q^{-1}u)
\end{equation*}
and
\begin{equation*}
q^mt^n G(u) E^{(n,m)}(u;\xi^{(n,m)}_\mu)  = \hat{\alpha}_m(u) G(qu) + t^{2n} \hat{\alpha}_m(u^{-1}) G(q^{-1}u),
\end{equation*}
we have that
\begin{align*}
W(q^{-1}u)
 &= F(u) G(q^{-1}u) - F(q^{-1} u) G(u)\\
&= F(u)   \biggl(    \frac{q^mG(u) E^{(n,m)}(u;\xi^{(n,m)}_\mu)}{t^n\hat{\alpha}_m(u^{-1})} - \frac{\hat{\alpha}_m(u)}{t^{2n}\hat{\alpha}_m(u^{-1})} G(qu)   \biggr)\\
 &\qquad    - G(u)  \biggl(    \frac{q^mF(u) E^{(n,m)}(u;\xi^{(n,m)}_\lambda)}{t^n\hat{\alpha}_m(u^{-1})} - \frac{\hat{\alpha}_m(u)}{t^{2n}\hat{\alpha}_m(u^{-1})} F(qu)   \biggr)  \\
&= t^{-2n} \frac{\hat{\alpha}_m(u)}{\hat{\alpha}_m(u^{-1})} W(u),
\end{align*}
where the cancellations  in the last step hinge on our assumption that $E^{(n,m)}(u;\xi_\lambda^{(n,m)})$ be equal to $ E^{(n,m)}(u;\xi^{(n,m)}_\mu)$. 
\end{proof}
The quasi-periodicity of the Casorati polynomial $W(u)$ in Lemma \ref{quasi-periodic:lem}  implies that in fact  $W(u)\equiv 0$. This becomes manifest
after rewriting the quasi-periodicity in Eqs. \eqref{W-qp-a}, \eqref{W-qp-b} as a polynomial identity by clearing denominators and negative powers of $u$:
\begin{eqnarray*}
\lefteqn{u^{4(m+1)}t^{2n+1}(u^4-t^{-1})(u^2-a_+)(u^2-a_-)W(q^{-1}u)} && \\
&&= (u^4-t)(a_+u^2-1)(a_-u^2-1)W(u) .
\end{eqnarray*}
When comparing the degrees in $u$ of the polynomial expressions on both sides, it is seen that the degree of the LHS exceeds the degree of the RHS unless $W(u)\equiv 0$.  The vanishing of the Casoratian implies in particular that $W(v_j)=F(qv_j)G(v_j)=0$ for $j=1,\ldots,n$.
Because $qv_j$ is not on the unit circle, it is clear that $F(qv_j)\neq 0$ and thus $G(v_j)=0$, whence
\begin{equation}\label{roots}
v_j\in \{ w_k,w_k^{-1},-w_k,-w_k^{-1}\mid k=1,\ldots,n\}\quad\text{for}\ j=1,\ldots,n.
\end{equation}
Since 
our spectral points belong to $2\pi\text{A}$ (by Part iii) of Proposition \ref{q-boson-spectrum:prp}), the principal arguments of $v_1,\ldots ,v_n$ and
those of $w_1,\ldots ,w_n$ correspond to angles between $0$ and $\pi/2$ in strictly decreasing order.
Hence there is no possible ambiguity in the identification of the roots in Eq. \eqref{roots}: $v_j=w_j$ (for $j=1,\ldots n$), i.e. $\xi^{(n,m)}_\lambda=\xi^{(n,m)}_\mu$, and thus $\lambda=\mu$ (by the injectivity in Part ii) of Proposition \ref{q-boson-spectrum:prp}).  The upshot is that the equality of the Laurent polynomials $E^{(n,m)}(u;\xi_\lambda^{(n,m)})$ and $E^{(n,m)}(u;\xi^{(n,m)}_\mu)$ implies that $\lambda=\mu$, which confirmes the desired nondegeneracy of the Bethe Ansatz spectrum in Eq. \eqref{simple},  and therewith completes the proof of Theorem  \ref{q-boson-eigenfunctions:thm}.


\section{Hyperoctahedral Hall-Littlewood polynomials}\label{sec11}
In this section we combine the branching rule in Theorem \ref{BAF-branching-rule:thm} with recent results from Ref. \cite{whe-zin:refined}, to express our Bethe Ansatz wave functions
in terms of Macdonald's hyperoctahedral Hall-Littlewood polynomials associated with the root system $BC_n$.

\subsection{Macdonald's formula \cite[\S 10]{mac:orthogonal}}
For $\lambda\in\Lambda_{n,m}$ \eqref{dominant},
Macdonald's Hall-Littlewood polynomial associated with the root system $BC_n$ is a three-parameter hyperoctahedral-symmetric Laurent polynomial in the variables $z_1,\ldots ,z_n$ of the form
\begin{subequations}
\begin{equation}\label{BC-HLa}
P_\lambda (z_1,\ldots,z_n;t,a,\hat{a})=   \sum_{\substack{ \sigma\in S_n \\ \epsilon\in \{ 1,-1\}^n}}   C( z_{\sigma_1}^{\epsilon_1},\ldots ,  z_{\sigma_n}^{\epsilon_n};t,a,\hat{a})
z_{\sigma_1}^{\epsilon _1\lambda_1}\cdots z_{\sigma_n}^{\epsilon _n\lambda_n} ,
\end{equation}
with
\begin{eqnarray}\label{BC-HLb}
\lefteqn{C(z_1,\ldots, z_n;t,a,\hat{a}) :=}&& \\
 && \prod_{1\leq j\leq n}   \frac{(z_j-a)(z_j-\hat{a})}{z_j^2-1}     
 \prod_{1\leq j<k\leq n} \left( \frac{z_jz_k-t}{z_jz_k-1} \right) \left(\frac{ z_jz_k^{-1}-t}{z_jz_k^{-1}-1} \right) .\nonumber
\end{eqnarray}
\end{subequations}
In the present normalization, the coefficient of the leading monomial  $z_1^{\lambda_1}\cdots z_n^{\lambda_n}$ in $P_\lambda (z_1,\ldots,z_n;t,a,\hat{a})$ is given by 
\begin{equation}
l.c.=(1-a\hat{a})(1-a\hat{a}t)\cdots (1-a\hat{a}t^{\text{m}_0(\lambda)-1})\prod_{l\in\mathbb{N}_m}   [\text{m}_l(\lambda)]! .
\end{equation}
In particular, for $\lambda=0^n$ we have that
\begin{equation}
P_{0^n}(z_1,\ldots,z_n;t,a,\hat{a})=(1-a\hat{a})(1-a\hat{a}t)\cdots (1-a\hat{a}t^{n-1}) [n]! .
\end{equation}

\subsection{Branching rule \cite[Sec. 3.3]{whe-zin:refined}}
In \cite[Sec. 3.3]{whe-zin:refined} a realization of Macdonald's hyperoctahedral Hall-Littlewood polynomials was found in terms of lattice paths.
By expressing the polynomials in question as matrix elements of a $q$-boson system with integrable boundary interactions,
this approach gives rise to the following branching rule for $\hat{a}=0$:
\begin{equation}\label{BC-HL-branching-rule}
P_\lambda (z_1,\ldots,z_n;t,a,0)=   \sum_{\substack{\mu\in \Lambda_{n-1,m} \\ \mu \leq \lambda}}
\hat{\text{B}}^{(n)}_{\lambda/\mu}(z_n ;t,a ) P_\mu (z_1,\ldots,z_{n-1};t,a,0),
\end{equation}
where $\hat{\text{B}}^{(n)}_{\lambda/\mu}(z  ;t,a ) $ is given by Eq. \eqref{BAF-branching-rule2}.  When also $a=0$ the contributions from the boundary site are trivial, and in this case the branching rule was derived in Remark 4 of 
\cite[Sec. 3.3]{whe-zin:refined}.   When $a\neq 0$, there is a single nonzero contribution from the boundary site that can be absorbed in the (diagonal) $K$-matrix, after which one recovers the branching rule in Eq. \eqref{BC-HL-branching-rule} by repeating
the steps in Remark 4 of \cite[Sec. 3.3]{whe-zin:refined} with the modified $K$-matrix.

A branching rule of the form in Eq. \eqref{BC-HL-branching-rule}, but with much more intricate branching coefficients $\hat{\text{B}}^{(n)}_{\lambda/\mu}$, was recently found for the five-parameter Macdonald-Koornwinder $q$-deformation of the hyperoctahedral Hall-Littlewood polynomials \cite{die-ems:branching}.

\subsection{Relation with the $n$-particle Bethe Ansatz wave function}
For $\hat{a}=0$, Macdonald's hyperoctahedral Hall-Littlewood polynomials in Eqs. \eqref{BC-HLa}, \eqref{BC-HLb} reduce in the case of a single variable ($n=1$) to
\begin{equation}\label{BC-HL-n=1}
P_l (z;a,0) =  \Bigl( \frac{z-a
}{z^2-1}\Bigr) z^{l+1}+\Bigl(\frac{z^{-1}-a}{z^{-2}-1}\Bigr) z^{-l-1} = s_l (z)-as_{l -1}(z)
\end{equation}
($l\in\mathbb{N}_m$).
By comparing 
Eq. \eqref{BC-HL-n=1} with Proposition \ref{BAF-n=1:prp} and Eq. \eqref{BC-HL-branching-rule} with Theorem \ref{BAF-branching-rule:thm}, it is now immediate that our $n$-particle Bethe Ansatz wave function can be expressed explicitly in terms
of Macdonald's hyperoctahedral Hall-Littlewood polynomials as follows:
\begin{align}\label{ABA-wave-BC-HL}
\Psi_{(v_1,\ldots ,v_n)}(\lambda)&= P_\lambda (v_1^2,\ldots,v_n^2;t,a_-,0)  \\
&= 
\sum_{\substack{ \sigma\in S_n \\ \epsilon\in \{ 1,-1\}^n}}   C( v_{\sigma_1}^{2\epsilon_1},\ldots ,  v_{\sigma_n}^{2\epsilon_n};t,a_-,0)
v_{\sigma_1}^{2\epsilon _1\lambda_1}\cdots v_{\sigma_n}^{2\epsilon _n\lambda_n}  \nonumber
\end{align}
($\lambda\in\Lambda_{n,m}$).

For the $q$-boson system with periodic boundary conditions an analogous formula for the algebraic Bethe Ansatz wave functions in terms of Hall-Littlewood polynomials can be found in \cite[Prp 4]{tsi:quantum} and \cite[Rem. 4.2]{kor:cylindric}. For the $q$-boson systems on the infinite and semi-infinite integral lattices, corresponding formulas for the $q$-boson eigenfunctions in terms of the Hall-Littlewood polynomials (infinite lattice) and 
the hyperoctahedral Hall-Littlewood polynomials (semi-infinite lattice) can be found in \cite[Sec. 5]{die-ems:diagonalization} and \cite[Sec. 3]{die-ems:semi-infinite}, respectively.  Recently it has moreover been shown
in Ref.  \cite{bor-cor-pet-sas:spectral} that, on the infinite lattice, the explicit formula in terms of Hall-Littlewood polynomials
arises as a degeneration of a more general formula for the
Bethe Ansatz eigenfunction of a stochastic $q$-boson system
related to the asymmetric exclusion process  \cite{sas-wad:exact}.

\subsection{Affine Pieri formula and orthogonality}
With the aid of Eq. \eqref{ABA-wave-BC-HL}, we can reformulate the results of Theorem \ref{q-boson-eigenfunctions:thm} completely in terms of Macdonald's hyperoctahedral Hall-Littlewood polynomials. 

For $\lambda,\mu\in\Lambda_{n,m}$, let
$(z_1,\ldots ,z_n)$ and $(y_1,\ldots ,y_n)$  be equal to $(e^{i\xi_1},\ldots ,e^{i\xi_n})$ at $\xi=\xi^{(n,m)}_\lambda$ and $\xi=\xi^{(n,m)}_\mu$, respectively (cf. Proposition \ref{q-boson-spectrum:prp}). Then it follows from Proposition \ref{action-Hm:prp} and Theorem \ref{q-boson-eigenfunctions:thm} that at these spectral points
Macdonald's hyperoctahedral Hall-Littlewood polynomials satisfy the following affine Pieri formula
\begin{align}\label{affine-pieri}
P_\nu (z_1,\ldots ,z_n;t,a_-,0)& \sum_{1\leq j\leq n} \Bigl( z_j+z_j^{-1} \Bigr) = \\
& \Bigl( a_-[\text{m}_0(\nu] +
a_+[\text{m}_m(\nu)] \Bigr) 
P_\nu (z_1,\ldots ,z_n;t,a_-,0)\nonumber  \\
&
+\sum_{\substack{1\leq j \leq n\\ \nu\pm e_j\in\Lambda_{n,m}}} [\text{m}_{\nu_j}(\nu)] P_{\nu\pm e_j} (z_1,\ldots ,z_n;t,a_-,0) , \nonumber
\end{align}
($\nu\in\Lambda_{n,m}$), and the following finite-dimensional discrete orthogonality relations
\begin{equation}\label{BC-HL-orthogonality}
\sum_{\nu\in\Lambda_{n,m}}
P_\nu (z_1,\ldots ,z_n;t,a_-,0)
\overline{P_\nu (y_1,\ldots ,y_n;t,a_-,0)} \delta_{n,m}(\nu)=0
\quad \text{iff}\ \lambda\neq\mu .
\end{equation}

For the Hall-Littlewood polynomials diagonalizing the $q$-boson system with periodic boundary conditions, corresponding affine Pieri formulas and finite-dimensional discrete orthogonality relations were presented
in Theorems 5.1 and  5.2 of Ref. \cite{die:diagonalization}, respectively. A further generalization of those formulas, related to a generalized $q$-boson system associated with the double affine Hecke algebra at critical level, was recently studied in Ref. \cite{die-ems:discrete}.

\begin{remark}
Theorem \ref{boundary-T-action:thm} and Theorem \ref{q-boson-eigenfunctions:thm} imply in fact
a system of affine Pieri formulas for the hyperoctahedral Hall-Littlewood polynomials
generalizing the one in Eq. \eqref{affine-pieri}.
Specifically, for $\lambda, \nu\in\Lambda_{n,m}$ one has that
at $(z_1,\ldots ,z_n)=(e^{i\xi_1},\ldots ,e^{i\xi_n})$ with $\xi=\xi^{(n,m)}_\lambda$:
\begin{align}
&P_\nu (z_1,\ldots ,z_n;t,a_-,0) E^{(n,m)}(u;z_1,\ldots ,z_n)= \\
&q^{-m}t^{-n-1} (a_+ -tu^{-2}) \sum_{\substack{\mu\in \Lambda_{n,m} \\ \mu \sim_- \nu}}
\text{A}^{(n)}_{\nu, \mu}(u^2;t,a_-) P_\mu (z_1,\ldots ,z_n;t,a_-,0)  \nonumber \\
&+
 q^{-m}t^{-n-1} (a_+ -u^2) \sum_{\substack{\mu\in \Lambda_{n,m} \\ \mu \sim_+ \nu}}
\text{D}^{(n)}_{\nu, \mu}(u^2;t,a_-)P_\mu (z_1,\ldots ,z_n;t,a_-,0)  \nonumber
\end{align}
(as an identity between Laurent polynomials in $u$). Here $E^{(n,m)}(u;z_1,\ldots ,z_n)$ is given by Eq. \eqref{Enm:eigenvalues} with $e^{i\xi_j}=z_j$ ($j=1,\ldots ,n$) and the coefficients $\text{A}^{(n)}_{\nu ,\mu}(z;t,a)$ and $\text{D}^{(n)}_{\nu, \mu}(z;t,a)$ are of the form in
Eqs. \eqref{A-boundary-coef} and \eqref{D-boundary-coef}, respectively.
\end{remark}


\section{Continuum limit}\label{sec12}
In this section we obtain the orthogonality in Theorem \ref{orthogonality:thm} as a continuum limit
of the orthogonality in Eq. \eqref{BC-HL-orthogonality}, by adapting the analysis in \cite[Sec. 6]{die:diagonalization}
to the present setup. Similar techniques were employed previously by Ruijsenaars
when studying the continuum limit of the infinite isotropic Heisenberg chain
\cite{rui:continuum}.

Throughout the section we put
\begin{equation}
t=e^{-\frac{g}{2m}},\ a_+=e^{-\frac{g_+}{2m}},\ a_-=e^{-\frac{g_-}{2m}}\quad\text{with}\ 
g,g_+,g_->0.
\end{equation}

\subsection{Staircase functions}
For any $x=(x_1,\ldots,x_n)$ in the closure $\overline{\text{C}}$ of the chamber $C$
\eqref{chamber}, we consider its lattice approximation $\lfloor x\rfloor \in\Lambda$  \eqref{partitions} defined by
\begin{equation}
\lfloor x\rfloor :=\sum_{1\leq j \leq n}  
\lfloor x_j-x_{j+1} \rfloor (e_1+\cdots +e_j)  \qquad (\text{with}\ x_{n+1}:=0),
\end{equation}
where $\lfloor\cdot\rfloor$ refers to the floor function extracting the integral part
of a nonnegative real number by truncation.
Armed with this lattice approximation, one can embed $\mathcal{F}_{n,m}$ \eqref{n-particle-sector}--\eqref{weight-function} into $L^2(C,\text{d}x)$ 
by means of a linear injection
$J^{(n,m)}:\mathcal{F}_{n,m}\to L^2(\text{C},\text{d}x)$, which associates
to $f:\Lambda_{n,m}\to\mathbb{C}$ a staircase function
$J^{(n,m)}(f):\overline{\text{C}}\to\mathbb{C}$ of the form
\begin{equation}
(J^{(n,m)} f)(x):=
\begin{cases}
\sqrt{\delta_{n,m}({\lfloor 2mx \rfloor } )} f (\lfloor 2m x\rfloor )
&\text{for}\ \lfloor 2m x\rfloor \in\Lambda_{n,m},\\
0&\text{for}\ \lfloor 2m x\rfloor \not\in\Lambda_{n,m} .
\end{cases}
\end{equation}
Clearly the staircase function $J^{(n,m)}(f)$ is compactly supported inside 
$(1+\frac{n}{m})\overline{\text{A}}\subset \overline{\text{C}}$, where $\overline{\text{A}}$ denotes the closure of the alcove $\text{A}$ \eqref{A}, 
and moreover
\begin{equation}\label{ipr}
\int_{\text{C}} (J^{(n,m)}f)(x)\overline{
(J^{(n,m)}g) (x)}\text{d} x=
(2m)^{-n} (f,g)_{n,m}
\end{equation}
($f,g\in\mathcal{F}_{n,m}$).

\subsection{Convergence of the Bethe Ansatz wave function}
Let $\psi^{(n,m)}(\xi ,x)$ denote the embedding in $L^2(\text{C},\text{d}x)$ of the hyperoctahedral Hall-Littlewood polynomial
\begin{equation*}
2^n P_\lambda (e^{i\xi_1},\ldots ,e^{i\xi_n}; e^{-\frac{g}{2m}},e^{-\frac{g_-}{2m}},0)
\end{equation*}
(viewed as a function of $\lambda\in\Lambda_{n,m}$):
\begin{eqnarray}\label{staircase}
\psi^{(n,m)} (\xi,x) &:=& 2^n(J^{(n,m)}
P (e^{i\xi_1},\ldots ,e^{i\xi_n}; e^{-\frac{g}{2m}},e^{-\frac{g_-}{2m}},0))(x)  \\
&= &2^n\sqrt{\delta_{n,m}(\lfloor 2mx\rfloor )}
P_{ \lfloor 2mx\rfloor } (e^{i\xi_1},\ldots ,e^{i\xi_n}; e^{-\frac{g}{2m}},e^{-\frac{g_-}{2m}},0),  \nonumber 
\end{eqnarray}
for $\lfloor 2mx\rfloor\in\Lambda_{n,m}$ (and zero otherwise).

We now retrieve the  Bethe Ansatz wave function $\psi (\xi , x) $ \eqref{BA-WF} as a pointwise continuum limit
of the staircase $q$-boson wave function $\psi^{(n,m)} (\xi,x)$ \eqref{staircase}.
\begin{lemma}\label{wflim:lem}
For any $\xi\in\emph{C}$, one has that
\begin{equation}
\lim_{m\to\infty} \psi^{(n,m)} ({\textstyle \frac{1}{2m}}\xi ,x) =
\begin{cases}
\psi (\xi , x) &\text{if}\
x\in\emph{A} ,\\
0&\text{if}\ x\in \emph{C}\setminus \overline{\emph{A}},
\end{cases}
\end{equation}
with $\psi (\xi , x) $ given by Eq. \eqref{BA-WF}.
\end{lemma}
\begin{proof}
Since $\psi^{(n,m)} (\xi ,x)$ is supported inside $(1+\frac{n}{m})\overline{\text{A}}\subset \overline{\text{C}}$, it is clear that
for $m$ sufficiently large $\psi^{(n,m)} (\xi ,x)=0$
if $x\in \text{C}\setminus \overline{\text{A}}$, i.e. $\lim_{m\to\infty} \psi^{(n,m)} ({\textstyle \frac{1}{2m}}\xi ,x) =0$ in this situation.
If on the other hand $x\in \text{A}$, then  $\lfloor 2mx \rfloor\in\Lambda_{n,m}$ with---for
$m$ sufficiently large---all parts being distinct, whence
 $$
 \lim_{m\to\infty}  \delta_{n,m}(\lfloor 2mx \rfloor)=1\quad\text{for}\ x\in\text{A}.
 $$
In this situation, the pointwise limit of the staircase wave function is then immediate from the explicit expression for the hyperoctahedral Hall-Littlewood polynomials in Eqs. \eqref{BC-HLa}, \eqref{BC-HLb}, upon observing that
$$
 \lim_{m\to\infty} \frac{1}{2m} \lfloor 2mx \rfloor=x
 $$
 and
 $$
 \lim_{m\to\infty} 2^n C(e^{\frac{i\xi_1}{2m}},\ldots ,e^{\frac{i\xi_n}{2m}}; e^{-\frac{g}{2m}},e^{-\frac{g_-}{2m}},0) = C(\xi_1,\ldots,\xi_n),
$$
with $C(\xi_1,\ldots,\xi_n)$ taken from Eq. \eqref{BA-WF}.
\end{proof}

\subsection{Convergence of the spectral points}
Next, we retrieve the Bethe Ansatz spectral points $\xi_\lambda$, $\lambda\in\Lambda$ for the Laplacian with repulsive Robin boundary conditions on the hyperoctahedral Weyl alcove, as scaled limits of the $q$-boson spectral points when the number of sites grows to infinity.

\begin{lemma}\label{speclim:lem}
For any $\lambda\in\Lambda$, one has that
\begin{equation}
\lim_{m\to\infty}
2m\xi^{(n,m)}_\lambda=
\xi_\lambda  ,
\end{equation}
where $\xi_\lambda$ denotes the unique global minimum of $V_\lambda (\xi)$ \eqref{morse}.
\end{lemma}
\begin{proof}
Given $\lambda\in\Lambda$, let us pick $m$ sufficiently large such that
$\lambda\in\Lambda_{n,m}$.
The rescaled spectral point $2m\xi^{(n,m)}_\lambda$ then corresponds to the unique global minimum
of the rescaled Morse function $2m V^{(n,m)}_\lambda (\frac{1}{2m}\xi)$, which for $m\to\infty$ converges uniformly on compacts to the
Morse function $V_\lambda(\xi)$ \eqref{morse} (cf. Remark \ref{convergence-morse-function} below). Hence the
unique global minimum $2m\xi^{(n,m)}_\lambda$ of $2m V^{(n,m)}_\lambda (\frac{1}{2m}\xi)$ converges to the unique global minimum
$\xi_\lambda$ of  $V_\lambda(\xi)$ as $m\to \infty$.
\end{proof}

\begin{remark}\label{convergence-morse-function}
The proof of Lemma \ref{speclim:lem} asserts  that $2m V^{(n,m)}_\lambda (\frac{1}{2m}\xi)$ converges uniformly on compacts to
$V_\lambda(\xi)$ \eqref{morse} as $m$ grows to infinity.  Indeed,
for $-\pi<\theta<\pi$ the function $v_a(\theta)$ \eqref{q-boson-morse-b} can be rewritten as
\begin{equation}
v_a(\theta)= 2\arctan \left( \frac{1+a}{1-a} \tan \left({\frac{\theta}{2}}\right)\right) .
\end{equation}
Hence, for $a=e^{-\varepsilon g}$ (with $\varepsilon, g>0$) the rescaled function $v_a(\varepsilon \theta)$ converges uniformly on compacts 
to $2\arctan (\frac{\theta}{g})$ as $\varepsilon\to 0$  (because $\epsilon^{-1}\tan(\epsilon \theta)$ converges uniformly on compacts to $\theta$). It thus follows that $\varepsilon^{-1}\int_0^{\varepsilon \theta} v_a(u)\text{d}u=\int_0^\theta v_a(\varepsilon u)\text{d} u$ converges in this situation uniformly on compacts to $2\int_0^\theta \arctan (\frac{u}{g}) \text{d} u$.
\end{remark}

\subsection{Proof of the orthogonality}
After these preparations, we are now in the position
to retrieve the orthogonality in Theorem \ref{orthogonality:thm} as the continuum limit
of the orthogonality in Eq. \eqref{BC-HL-orthogonality}.

\begin{proposition}[Continuum Limit]\label{lim:prp}
For all $\lambda, \mu\in
\Lambda$, one has that
\begin{equation}
\lim_{m\to\infty} \int_{\emph{C}} \psi^{(n,m)}
(\xi^{(n,m)}_{\lambda},x) \overline{
\psi^{(n,m)}
(\xi^{(n,m)}_{\mu},x)}\text{d}x 
 = \int_{\emph{A}} \psi
(\xi_{\lambda},x)
\overline{\psi
(\xi_{\mu},x)}
\text{d} x
.
\end{equation}
\end{proposition}
\begin{proof}
Given $\lambda,\mu\in\Lambda$, let us pick $m$ suffciently large such that $\lambda,\mu\in\Lambda_{n,m}$.
To verify the asserted continuum limit, we make the expressions under consideration more explicit:
\begin{align*}
 \int_{\text{C}} \psi^{(n,m)}&
(\xi^{(n,m)}_\lambda,x)  \overline{ \psi^{(n,m)}
(\xi^{(n,m)}_\mu,x)}\text{d}x =\\
 \sum_{\substack{\sigma,\hat{\sigma}\in\mathcal{S}_n\\ \epsilon,\hat{\epsilon}\in \{ 1,-1\}^n }} &\biggl(
C^{(n,m)}_{\sigma,\epsilon}(\xi^{(n,m)}_\lambda) C^{(n,m)}_{\hat{\sigma},\hat{\epsilon}}(-\xi^{(n,m)}_\mu) \\
&\times\int_{(1+\frac{n}{m})\text{A}} \delta_{n,m} (\lfloor 2m x\rfloor ) \prod_{1\leq j\leq n}e^{i\lfloor 2m x\rfloor_j (\epsilon_j (\xi^{(n,m)}_\lambda)_{\sigma_j}-\hat{\epsilon}_j(\xi^{(n,m)}_\mu)_{\hat{\sigma}_j})}
 \text{d}x \biggl) ,
\end{align*}
where
$C^{(n,m)}_{\sigma,\epsilon}(\xi) := 
2^n C(e^{i\epsilon_1 \xi_{\sigma_1}},\ldots ,e^{i\epsilon_n \xi_{\sigma_n}}; e^{-\frac{g}{2m}},e^{-\frac{g_-}{2m}},0) $ and we have used the convention that $\delta_{n,m} (\lfloor 2m x\rfloor ) :=0$ if $\lfloor 2m x\rfloor \not\in\Lambda_{n,m}$.
It is clear from our previous analysis above that for $m\to\infty$
$$
C^{(n,m)}_{\sigma,\epsilon}(\xi^{(n,m)}_\lambda)  \longrightarrow C(\epsilon_1(\xi_\lambda)_{\sigma_1},\ldots,\epsilon_n(\xi_\lambda)_{\sigma_n})
$$
and
$$
e^{i\epsilon_j \lfloor 2m x\rfloor_j  (\xi^{(n,m)}_\lambda)_{\sigma_j}}\longrightarrow e^{i \epsilon_j  x_j (\xi_\lambda)_{\sigma_j}} ,
$$
$$
\delta_{n,m} (\lfloor 2m x\rfloor ) \longrightarrow \begin{cases} 1&\text{if}\ x\in \text{A}\\
0 &\text{if}\ x\in \text{C}\setminus\overline{\text{A}}
\end{cases} 
$$
(pointwise).
Moreover, since $| e^{i\epsilon_j \lfloor 2m x\rfloor_j  (\xi^{(n,m)}_\lambda)_{\sigma_j}} |=1$ and $0\leq \delta_{n,m} (\lfloor 2m x\rfloor ) \leq 1$,
the dominated convergence theorem of Lebesgue then guarantees that for $m\to\infty$ our integral converges to
\begin{align*}
 \sum_{\substack{\sigma,\hat{\sigma}\in\mathcal{S}_n\\ \epsilon,\hat{\epsilon}\in \{ 1,-1\}^n }} &\biggl(
C(\epsilon_1 (\xi_\lambda)_{\sigma_1},\ldots,\epsilon_n (\xi_\lambda)_{\sigma_n})
C(-\hat{\epsilon}_1 (\xi_\mu)_{\hat{\sigma}_1},\ldots,-\hat{\epsilon}_n (\xi_\mu)_{\hat{\sigma}_n}) \\
&\times\int_{\text{A}} \prod_{1\leq j\leq n}e^{i x_j (\epsilon_j (\xi_\lambda)_{\sigma_j}-\hat{\epsilon}_j(\xi_\mu)_{\hat{\sigma}_j})}
 \text{d}x \biggl) ,
\end{align*}
which is  precisely $ \int_{\emph{A}} \psi
(\xi_{\lambda},x)
\overline{\psi
(\xi_{\mu},x)}
\text{d} x$.
\end{proof}

From Proposition \ref{lim:prp}, we see that
\begin{align}\label{ip-limit}
&\int_{\emph{A}} \psi
(\xi_{\lambda},x)
\overline{\psi
(\xi_{\mu},x)}
\text{d} x
=\\
& \lim_{m\to\infty} 4^n\int_{\text{C}} (J^{(n,m)} P
(z_1,\ldots ,z_n;t,a_-,0))(x)\overline{
(J^{(n,m)} P (y_1,\ldots ,y_n;t,a_-,0))
(x)}\, \text{d} x , \nonumber
\end{align}
where $(z_1,\ldots ,z_n)$ and $(y_1,\ldots ,y_n)$  are given by $(e^{i\xi_1},\ldots ,e^{i\xi_n})$ at $\xi=\xi^{(n,m)}_\lambda$ and $\xi=\xi^{(n,m)}_\mu$, respectively.
Since it is immediate from Eqs. \eqref{BC-HL-orthogonality} and \eqref{ipr} that the RHS of Eq. \eqref{ip-limit} vanishes when $\lambda\neq \mu$,
the orthogonality in
Theorem \ref{orthogonality:thm} thus follows.


\section{Orthogonality of the Bethe Ansatz on classical Weyl alcoves}\label{sec13}
The eigenvalue problem for the Laplacian on the hyperoctahedral Weyl alcove
in Eqs. \eqref{LEP}--\eqref{BC2} admits a generalization in terms of the Weyl alcove associated with an arbitrary Weyl group \cite{gau:boundary,gau:bethe,gut-sut:completely,gut:integrable,hec-opd:yang,die:plancherel,ems-opd-sto:periodic,bus-die-maz:norm}. The completeness of the corresponding Bethe Ansatz eigenfunctions was shown in Ref. \cite{ems:completeness}.
For all classical Weyl alcoves, i.e. those of  type $A$, $B$, $C$ and $D$,  an orthogonal basis of these eigenfunctions now follows
from \cite{dor:orthogonality,die:diagonalization} (type $A$) and Theorem \ref{orthogonality:thm} (types $B$, $C$ and $D$). 

To formulate this orthogonality result more precisely, some additional definitions and notations are needed.  Let $R$ be a reduced crystallographic root system spanning a real Euclidean vector space $V$
with inner product $\langle \cdot , \cdot \rangle$,  and let $W$ be the Weyl group generated by the orthogonal reflections $r_\alpha:V\to V$ in the root hyperplanes of the form $r_\alpha (x):=x-\langle x,\alpha^\vee\rangle \alpha$ ($\alpha\in R$), where  $\alpha^\vee :=2\alpha / \langle \alpha ,\alpha \rangle $ (and $x\in V $).   For a fixed choice of positive roots $R_+\subset R$ and a root multiplicity parameter $g_\alpha >0$ on $R$ such that $g_\alpha=g_\beta$ if $\langle \alpha,\alpha\rangle= \langle \beta,\beta\rangle$,
let $\xi_\lambda\in V$ denote the unique
global minimum of the strictly convex function
\begin{equation}\label{morse-R}
V_{{\lambda}}({\xi}):= \frac{1}{2}\langle
{\xi},{\xi}\rangle-2\pi \langle
{\rho}+{\lambda},{\xi}\rangle
+\sum_{{\alpha}\in {R}_+}
\langle \alpha,\alpha \rangle \int_0^{\langle
{\xi},{\alpha}^\vee\rangle}\arctan\left(\frac{u}{g_{\alpha}}\right)\text{d}u
 ,
\end{equation}
where ${\rho}:=\frac{1}{2}\sum_{{\alpha}\in
{R}_+} {\alpha}$ and ${\lambda}$ is taken
from the cone of dominant weights
\begin{equation}\label{dominant-cone-R}
\Lambda := \{ {\lambda}\in {V}\mid \langle
{\lambda},{\alpha}^\vee\rangle \in
\mathbb{Z}_{\geq 0},\ \forall {\alpha}\in {R}_+ \}  .
\end{equation}
This minimum is determined by the critical equation $\nabla_\xi V_\lambda (\xi)=0$:
\begin{equation}
\xi + 2 \sum_{\alpha\in R_+} \arctan\biggl(\frac{ \langle
{\xi},{\alpha}^\vee\rangle }{g_{\alpha}}\biggr)   \alpha  =2\pi (
{\rho}+{\lambda}) ,
\end{equation}
and (thus) satisfies the following Bethe Ansatz equations 
\begin{equation}
e^{i\langle \xi,\alpha^\vee\rangle }= \prod_{\substack{ \beta\in R\\ \langle \beta ,\alpha^\vee\rangle >0   }}  
\biggl( \frac{ig_\beta +\langle \xi ,\beta^\vee\rangle }{ig_\beta -\langle \xi ,\beta^\vee\rangle} \biggr)^{\langle \beta ,\alpha^\vee\rangle }\qquad (\alpha\in R_+).
\end{equation}

It was shown in Ref. \cite{ems:completeness} that the linear span of the Bethe Ansatz wave functions $\psi (\xi_\lambda,x)$, $\lambda\in \Lambda$, where
\begin{equation}\label{BA-WF-R}
\psi (\xi ,x):= \sum_{w\in W} e^{i\langle
w{\xi} ,{x}\rangle} \prod_{\boldsymbol{\alpha}\in
{R}^+} \frac{\langle w
{\xi},{\alpha}^\vee \rangle
-ig_{{\alpha}}}{\langle  w {\xi},
{\alpha}^\vee \rangle}  ,
\end{equation}
is dense in the Hilbert space $L^2(\text{A},\text{d}{x})$  of quadratically integrable functions on the Weyl alcove
\begin{equation}\label{WA}
\text{A}=\{ {x}\in V \mid 0<\langle
{\alpha},{x}\rangle < 1,\ \forall
{\alpha}\in {R}^+\} .
\end{equation}
(Here the integration is meant with respect to the standard Lebesgue measure $\text{d}{x}$ inherited from the Euclidean space $V$.)

By specialization of the boundary parameters in Theorem \ref{orthogonality:thm}, we extend the orthogonality for the root system of type $A$ from Refs. \cite{dor:orthogonality,die:diagonalization} so as to arrive at an orthogonal basis of Bethe Ansatz wave functions on the Weyl alcove for any classical root system.

\begin{theorem}[Types A,B,C]\label{orthogonality-ABC:thm} When $R$ is of type $A$, $B$ or $C$, the repulsive Bethe Ansatz wave functions $\psi(\xi_\lambda ,x)$, $\lambda\in \Lambda$ constitute an orthogonal basis for the Hilbert space
$L^2(\text{A},\text{d}{x})$.
\end{theorem}
\begin{proof}
The proof hinges on a case-by-case analysis based on Bourbaki's tables for the irreducible root systems \cite{bou:groupes}.

For $R=A_n$ the stated orthogonality amounts to \cite[Thm. 5.1]{dor:orthogonality}, upon projection onto the center-of-mass plane as in  \cite[Thm. 6.3]{die:diagonalization} (where attention was restricted to the Laplacian in the center-of-mass plane).

For $R=B_n$ the stated orthogonality is retrieved from Theorem \ref{orthogonality:thm} 
through the parameter identification $g_\alpha=g$ when $\alpha$ is long
and $g_\alpha =2g_-$ when $\alpha$ is short (cf. e.g. \cite[Sec. 4.2]{bus-die-maz:norm}).
Indeed, for
$g_+\to 0$ and $g_+\to +\infty$ one recovers the Bethe Ansatz wave functions on the $B_n$-type Weyl alcove
\begin{equation}\label{A-Bn}
\{   {x}\in\mathbb{R}^n\mid  1-x_2> x_1> x_2> \cdots > x_{n-1}>  x_n>0 \} 
\end{equation}
that are respectively even and odd with respect to the extended affine reflection of length zero. (This reflection acts on the $B_n$-type Weyl alcove \eqref{A-Bn} as the reflection-symmetry $x_1\to1-x_1$ in the affine hyperplane $x_1=\frac{1}{2}$ for which  $A$ \eqref{A} determines a fundamental domain).

For $R=C_n$ the stated orthogonality is retrieved from that of Theorem \ref{orthogonality:thm}  through the parameter identification $g_\alpha=g$ when $\alpha$ is short 
and $g_\alpha=g_+=g_-$ when $\alpha$ is long (cf. e.g. \cite[Sec. 4.3]{bus-die-maz:norm}).
\end{proof}

For $R=D_n$, let $\lambda \to \lambda^\star$ be the orthogonal involution on the cone of dominant weights $\Lambda$ \eqref{dominant-cone-R} induced by the Dynkin diagram automorphism interchanging the two fundamental spin weights. We define
\begin{equation}\label{BA-WF-Dn}
\psi_{\pm }(\xi_\lambda, x) :=  \frac{1}{1+\delta_{\lambda ,\lambda^\star}} \left( \psi ( \xi_\lambda ,x ) \pm  \psi (\xi_{\lambda^\star} ,x)\right) ,
\end{equation}
where $\delta_{\lambda ,\mu}:=1$ if $\lambda =\mu$ and $\delta_{\lambda ,\mu}:=0$ otherwise
(so $\psi_{+ }(\xi_\lambda,x)= \psi (\xi_\lambda ,x)$ and $\psi_{- }(\xi_\lambda, x)=0$ if $\lambda =\lambda^\star$).

\begin{theorem}[Type D]\label{orthogonality-D:thm} When $R$ is of type $D$, 
the (nonvanishing) repulsive Bethe Ansatz wave functions $\psi_\pm (\xi_\lambda ,x )$, 
$\lambda\in \Lambda  \mod \star$ constitute an orthogonal basis for the Hilbert space
$L^2(\text{A},\text{d}{x})$.
\end{theorem}

\begin{proof}
The case $R=D_n$ is retrieved from the case $R=B_n$ using the standard
realizations of these root systems in accordance with the tables in Bourbaki \cite{bou:groupes}.
 Indeed,  the
 limits $g_\alpha\to 0$ and $g_\alpha\to +\infty$ for $\alpha$ short (in the second case upon renormalizing through an overall multiplication of the wave function by $(i/g_\alpha)^n\prod_{\substack{\alpha\in R_+\\ \alpha \ \text{short}}} \langle \xi,\alpha^\vee\rangle$) recover the
Bethe Ansatz wave functions $\psi_+ (\xi_\lambda ,x)$ and
$\psi_- (\xi_\lambda ,x)$, respectively (cf. e.g. \cite[Secs. 4.2, 4.4]{bus-die-maz:norm}). Since  the Bethe Ansatz wave functions $\psi (\xi_\lambda ,x) $ and $\psi (\xi_{\lambda^\star},x) $
on the  $D_n$-type Weyl alcove
\begin{equation}
\{   {x}\in\mathbb{R}^n\mid  1-x_2> x_1>x_2> \cdots >x_{n-1}>  |x_n| \} 
\end{equation}
are related
by the reflection-symmetry $x_n\to -x_n$  in the hyperplane $x_n=0$, it is clear that $\psi_+ (\xi_\lambda ,x)$ and
$\psi_- (\xi_\lambda ,x)$ are respectively even and odd with respect to this symmetry (for which
the $B_n$-type Weyl alcove \eqref{A-Bn} determines a fundamental domain).
\end{proof}

\begin{remark}
The analog of the eigenvalue problem \eqref{LEP}--\eqref{BC2} satisfied by the
Bethe Ansatz wave function $\psi (\xi_\lambda ,x)$ \eqref{BA-WF-R} reads  \cite{gau:boundary,gau:bethe,gut-sut:completely,gut:integrable,hec-opd:yang,die:plancherel,ems-opd-sto:periodic,bus-die-maz:norm}:
\begin{subequations}
\begin{equation}\label{LEP-R}
-\Delta \psi = \langle \xi_\lambda ,\xi_\lambda\rangle \psi ,
\end{equation}
where $\Delta$ denotes the Laplacian on the alcove $\text{A}$ \eqref{WA} with repulsive Robin boundary conditions at the walls of the form
\begin{equation}\label{BC1-R}
\langle \nabla \psi ,\alpha^\vee_j\rangle -g_{\alpha_j} \psi=0\quad \text{when}\quad \langle \alpha_j,x\rangle =0
\end{equation}
$(j=1,\ldots ,n )  $ and
\begin{equation}\label{BC2-R}
\langle \nabla \psi ,\alpha^\vee_0\rangle +g_{\alpha_0} \psi =0\quad \text{when}\quad \langle \alpha_0,x\rangle =1 .
\end{equation}
\end{subequations}
Here $\alpha_1,\ldots,\alpha_n$ and $\alpha_0$ refer to the basis of simple roots and to the maximal root
of $R_+$, respectively, and $\nabla$ denotes the gradient with respect to $x$. When $R$ is of type $D_n$, the Bethe Ansatz wave functions $\psi_+(\xi_\lambda ,x)$ and
$\psi_-(\xi_\lambda, x)$ \eqref{BA-WF-Dn} also satisfy the corresponding eigenvalue problem because $\langle \xi_\lambda,\xi_\lambda\rangle =
\langle \xi_{\lambda^\star},\xi_{\lambda^\star}\rangle$ (as the orthogonal reflection mapping $\lambda$ to $\lambda^\star$ also maps $\xi_\lambda$ to $\xi_{\lambda^\star}$). As a consequence of Theorems \ref{orthogonality-ABC:thm} and \ref{orthogonality-D:thm}, we may thus conclude that for any classical root system $R$ the repulsive  Laplacian $-\Delta$ in Eqs. \eqref{LEP-R}--\eqref{BC2-R} constitutes an unbounded essentially self-adjoint operator in
 $L^2(\text{A},\text{d}{x})$  that has purely discrete spectrum given by  the positive eigenvalues
 $\langle \xi_\lambda,\xi_\lambda\rangle $, $\lambda\in\Lambda$. Notice in this connection that it is immediate from (the gradient of) the Morse function $ V_\lambda (\xi)$ \eqref{morse-R} that the eigenvalues
 $\langle \xi_\lambda,\xi_\lambda\rangle$ do not remain bounded as $\lambda$ varies over the dominant cone $\Lambda$ \eqref{dominant-cone-R} (cf.  also \cite[Secs. 9, 10]{ems-opd-sto:periodic}).
\end{remark}

\begin{remark}
From the point of view of root systems, our eigenvalue problem on the hyperoctahedral Weyl alcove in Eqs. \eqref{LEP}--\eqref{BC2} corresponds to the case of a nonreduced root system of type $BC_n$ (cf. also Ref.  \cite{ems-opd-sto:trigonometric}, where the present eigenvalue problem is interpreted in terms of an affine root system associated with $R=C_n$).
\end{remark}

\appendix


\section{Sklyanin's boundary transfer operator}\label{appA}
This appendix summarizes a special instance of Sklyanin's extension of the quantum inverse scattering formalism enabling the construction of
transfer operators for open systems with boundary interactions \cite{skl:boundary}. Section \ref{sec3} relies on this formalism to provide the boundary transfer operator for the open finite $q$-boson system with diagonal boundary interactions at the lattice endpoints.

Let us assume that $R(u)\in \mathbb{C}[u^{\pm 1}]\otimes M_4(\mathbb{C})$ enjoys the following three symmetries:

\emph{i) $PT$-symmetry}
\begin{equation}\label{pt-symmetry}
R^t(u) =PR(u) P,
\end{equation}

\emph{ii) Unitarity}
\begin{equation}\label{unitarity}
R(u) R(u^{-1})= \rho(u)I,
\end{equation}

\emph{iii) Crossing unitarity}
\begin{equation}\label{crossing-unitarity}
(R(q u) P)^{t_1} (P R(qu^{-1})  )^{t_2} = \hat\rho(u) I,
\end{equation}
for certain  nonvanishing $\rho(u), \hat\rho(u) \in  \mathbb{C}[u^{\pm 1}] $ and a suitable shift $q\in\mathbb{C}\setminus \{ 0\}$. Here
$I\in M_4(\mathbb{C})$ stands for the identity matrix, $P\in M_4(\mathbb{C})$ denotes the permutation matrix
\begin{equation*}
P:= \begin{pmatrix}  1& 0&0&0 \\0&0&1&0 \\0&1&0&0\\ 0&0&0& 1      \end{pmatrix},
\end{equation*}
 $R^t$ represents the transposed of $R$, and $R^{t_1}$, $R^{t_2}$ refer to the partially transposed matrices characterized by the property that $(A\otimes B)^{t_1}=A^t\otimes B$ and $(A\otimes B)^{t_2}=A\otimes B^t$ for $A,B\in M_2(\mathbb{C})$ (and extended to  $\mathbb{C}[u^{\pm1}]\otimes M_4(\mathbb{C})$ by linearity).
For instance, our particular $R$-matrix in Eq.  \eqref{R-matrix} turns out to satisfy these  symmetries with $\rho(u)=\hat\rho(u)=s(qu)s(qu^{-1})$.

Let $\mathbb{A}_+$ and $\mathbb{A}_-$  be  unital associative algebras over $\mathbb{C}$, and let $K_\pm (u)\in \mathbb{C}[u^{\pm 1}]\otimes M_2(\mathbb{A}_\pm )$  be solutions of
the left reflection equation
\begin{subequations}
\begin{equation}\label{RE1}
R(u/v)   K_-(u)_1 R(quv) K_-(v)_1
=
 K_-(v)_1 R(quv)   K_-(u)_1 R(u/v)
\end{equation}
in $\mathbb{C}[u^{\pm 1},v^{\pm 1}]\otimes M_4(\mathbb{A}_- )$, and the right reflection equation
\begin{equation}\label{RE2}
R(u/v) K_+(u)_2  R(q u v) K_+(v)_2
= K_+(v)_2 R(quv)  K_+(u)_2  R(u/v) 
\end{equation}
\end{subequations}
in $\mathbb{C}[u^{\pm 1},v^{\pm 1}]\otimes M_4(\mathbb{A}_+)$, respectively.
Following standard conventions, both $\mathbb{A}_+$ and $\mathbb{A}_-$ are thought of as subalgebras of $\mathbb{A}_+\otimes \mathbb{A}_-$ via the embeddings
$a_+\to a_+\otimes 1$ ($a_+\in \mathbb{A}_+$) and $a_-\to 1\otimes a_-$ ($a_-\in \mathbb{A}_-$), which gives rise to the commutativity
 $$a_+a_-:=  (a_+\otimes 1)(1\otimes a_-) = a_+\otimes a_- =(1\otimes a_-)   (a_+\otimes 1)     =:a_-a_+ .$$

\begin{theorem}[\cite{skl:boundary}]\label{t:btm1}
The boundary transfer operator
$$
\mathcal{T}(u) := \text{tr} \left(  K_+^t(u^{-1}) K_-(u)\right)\in \mathbb{C}[u^{\pm 1}]\otimes \mathbb{A}_+\otimes\mathbb{A}_-
$$
satisfies the commutativity  $$[\mathcal{T}(u), \mathcal{T}(v)] =0$$ in $\mathbb{C}[u^{\pm 1},v^{\pm 1}]\otimes \mathbb{A}_+\otimes \mathbb{A}_-$.
\end{theorem}

For a  complex unital associative algebra $\mathbb{A}$, let $U(u)$ be an invertible element of $ \mathbb{C}[u^{\pm 1}]\otimes M_2(\mathbb{A})$ satisfying the
quantum Yang-Baxter equation
\begin{equation}\label{YB}
R(u/v) U(u)_1 U(v)_2 =  U(v)_1 U(u)_2 R(u/v)
\end{equation}
in $ \mathbb{C}[u^{\pm 1},v^{\pm 1}]\otimes M_4(\mathbb{A})$.

\begin{theorem}[\cite{skl:boundary}]\label{t:btm2}
The boundary monodromy matrix
$$\mathcal{U}(u):=U(u) K_-(u) U^{-1}(q^{-1}u^{-1})
\in  \mathbb{C}[u^{\pm 1}]\otimes M_2(\mathbb{A}_- \otimes \mathbb{A})
$$ solves the left reflection equation  \eqref{RE1} in $\mathbb{C}[u^{\pm 1},v^{\pm 1}]\otimes M_4(\mathbb{A}_-\otimes \mathbb{A} )$
(upon substituting $\mathcal{U}(u)$ for $K_-(u)$).
\end{theorem}

By combining these two theorems, it is immediate that the boundary transfer operator
\begin{equation}
\mathcal{T}(u) = \text{tr} \left(  K_+^t(u^{-1}) \,  \mathcal{U}(u)    \right)\in \mathbb{C}[u^{\pm 1}]\otimes \mathbb{A}_+\otimes  \mathbb{A}_-\otimes \mathbb{A}
\end{equation}
satisfies the desired commutativity $[\mathcal{T}(u),\mathcal{T}(v)]=0$ in $\mathbb{C}[u^{\pm 1},v^{\pm 1}]\otimes \mathbb{A}_+\otimes  \mathbb{A}_-\otimes \mathbb{A}$, which is the form in which  Sklyanin's formalism was applied in Section \ref{sec3} (with $\mathbb{A}_+=\mathbb{A}_-=\mathbb{C}$ and $\mathbb{A}=\mathbb{A}_m$).

Both theorems were proven (within a broader setup) by Sklyanin \cite{skl:boundary}, who assumed additionally that $R^t(u)=R(u)$. While minor variations on  Sklyanin's arguments allow to suppress the latter assumption that $R(u)$ be symmetric, in the mean time revised versions of the theorems at issue  have been formulated requiring milder conditions
on $R(u)$ \cite{mez-nep:integrable,fan-shi-hou-yan:integrable,vla:boundary}. 
In particular, the formulation in Ref. \cite{vla:boundary} only imposes that a partial transposition of the  $R$-matrix be invertible. We close this appendix by a) indicating how Theorem \ref{t:btm1}  can be retrieved upon specialization from a more general statement taken from \cite[Thm. 2.4]{vla:boundary}  and b) recalling how Theorem \ref{t:btm2} is verified via an elementary computation going back to \cite[Prp. 2]{skl:boundary}.

\subsection{Proof of Theorem \ref{t:btm1}}
We apply \cite[Thm. 2.4]{vla:boundary} with $\mathcal{U}^+(u) = K_+^t(q^{1/2}u^{-1})$,
$\mathcal{U}^-(u) = K_-(q^{-1/2}u)$ and $R_V(u)=R(u)P$, where $R_V$ refers to the $R$-matrix employed in Ref. \cite{vla:boundary}.
To conclude the commutativity for the boundary transfer operator
\begin{equation*}
 \mathcal{T}(u)=  \text{tr} \left(  K_+^t(u^{-1}) K_-(u)\right) = \text{tr} \left( \mathcal U^+(q^{1/2}u)\mathcal U^-(q^{1/2}u)\right) 
\end{equation*}
from \cite[Thm. 2.4]{vla:boundary},
it is sufficient to check that these $\mathcal{U}^+$, $\mathcal{U}^-$ and $R_V$ comply with the assumptions stated in the theorem.
Indeed,  the crossing unitarity \eqref{crossing-unitarity}  of $R(u)$ implies that $R_V^{t_1}(u)$ is invertible in $\mathbb{C}[u^{\pm 1}]\otimes M_4(\mathbb{C})$, and the left reflection equation \eqref{RE1} implies that
$$
PR_V(u/v) P \mathcal{U}^-(u)_1 R_V(uv)  \mathcal{U}^-(v)_2 =
\mathcal{U}^-(v)_2 P R_V(uv) P \mathcal{U}^-(u)_1  R_V(u/v) ,
$$
which confirms that
$\mathcal{U}^-(u)$ solves the reflection equation \cite[Eq. (2.5)]{vla:boundary}. 
Finally, the right reflection equation \eqref{RE2} implies that
\begin{eqnarray*}
\lefteqn{R_V(v/u) P (\mathcal{U}^+(u)_2)^{t_2}  R_V(q^2 u^{-1} v^{-1}) P ( \mathcal{U}^+(v)_2)^{t_2}}&& \\
&&= (\mathcal{U}^+(v)_2)^{t_2} R_V(q^2 u^{-1} v^{-1})  P (\mathcal{U}^+(u)_2)^{t_2} R_V(v/u) P .
\end{eqnarray*}
Upon exploiting the PT-symmetry \eqref{pt-symmetry}, the unitarity \eqref{unitarity}, and the crossing unitarity \eqref{crossing-unitarity} of $R(u)$, this identity can be recasted in the form
\begin{eqnarray*}
\lefteqn{P R_V^{-t} (u/v) P (\mathcal{U}^+(u)_1)^{t_1}  \tilde{R}_V^t(uv)  ( \mathcal{U}^+(v)_2)^{t_2}}&& \\
&&= (\mathcal{U}^+(v)_2)^{t_2} P\tilde{R}_V^t(uv)  P (\mathcal{U}^+(u)_1)^{t_1} R_V^{-t}(u/v) ,
\end{eqnarray*}
where $\tilde{R}_V(u):= ((( R_V(u) )^{t_1})^{-1})^{t_1}$. This confirms that
$\mathcal{U}^+(u)$ solves the dual reflection equation \cite[Eq. (2.7)]{vla:boundary}. 

\subsection{Proof of Theorem \ref{t:btm2}}
Elementary manipulations involving the left reflection equation \eqref{RE1}  and the quantum Yang-Baxter equation \eqref{YB}
readily entail the desired equality \cite[Prp. 2]{skl:boundary}:
\begin{align*}
R(u/v) &  \mathcal{U} (u)_1 R(quv) \mathcal{U}(v)_1\\
=& 
R(u/v)  U(u)_1 K_-(u)_1 U^{-1}(q^{-1}u^{-1})_1 R(quv) U(v)_1 K_-(v)_1 U^{-1}(q^{-1}v^{-1})_1\\
\stackrel{YB}{=}& 
R(u/v)  U(u)_1 K_-(u)_1 U(v)_2 R(quv) U^{-1}(q^{-1}u^{-1})_2 K_-(v)_1 U^{-1}(q^{-1}v^{-1})_1 \\
=& 
R(u/v)  U(u)_1  U(v)_2 K_-(u)_1 R(quv)  K_-(v)_1 U^{-1}(q^{-1}u^{-1})_2 U^{-1}(q^{-1}v^{-1})_1 \\
\stackrel{YB}{=}& 
 U(v)_1  U(u)_2 R(u/v)  K_-(u)_1 R(quv)  K_-(v)_1 U^{-1}(q^{-1}u^{-1})_2 U^{-1}(q^{-1}v^{-1})_1 \\
\stackrel{RE}{=}& 
 U(v)_1  U(u)_2   K_-(v)_1 R(quv)  K_-(u)_1 R(u/v) U^{-1}(q^{-1}u^{-1})_2 U^{-1}(q^{-1}v^{-1})_1 \\
 \stackrel{YB}{=}& 
 U(v)_1  U(u)_2   K_-(v)_1 R(quv)  K_-(u)_1  U^{-1}(q^{-1}v^{-1})_2 U^{-1}(q^{-1}u^{-1})_1 R(u/v)\\
=& 
 U(v)_1    K_-(v)_1 U(u)_2 R(quv)   U^{-1}(q^{-1}v^{-1})_2  K_-(u)_1  U^{-1}(q^{-1}u^{-1})_1 R(u/v)\\
 \stackrel{YB}{=}& 
 U(v)_1 K_-(v)_1  U^{-1}(q^{-1}v^{-1})_1    R(quv) U(u)_1    K_-(u)_1  U^{-1}(q^{-1}u^{-1})_1 R(u/v)\\
=&
\mathcal{U}(v)_1 R(quv)   \mathcal{U}(u)_1 R(u/v).
\end{align*}


\section{Relations between the entries of the monodromy matrix}\label{appB}
The quantum Yang-Baxter equation \eqref{YB} and the left reflection equation \eqref{RE1} each encode up to sixteen relations between the entries of the monodromy matrices
\begin{equation}\label{monodromy-elements}
U(u)=\begin{pmatrix}
A(u) & B(u) \\
C(u) & D(u)
\end{pmatrix} 
\quad\text{and}\quad
\mathcal{U}(u)=\begin{pmatrix}
\mathcal{A}(u) & \mathcal{B}(u) \\
\mathcal{C}(u) & \mathcal{D}(u)
\end{pmatrix} ,
\end{equation}
respectively. In this appendix we exhibit the (for our purposes) most relevant ones explicitly, assuming
that $R(u)$ is of the concrete form in Eq. \eqref{R-matrix}.

\subsection{Relations for the periodic monodromy matrix  \cite{tak:integrable,kor-bog-ize:quantum,jim-miw:algebraic,fad:how}}

\begin{lemma}\label{periodic-monodromy-relations:lem} 
The entries $A(u)$, $B(u)$, $C(u)$ and $D(u)$ in $\mathbb{C}[u^{\pm 1}]\otimes\mathbb{A}$ of a (periodic monodromy) matrix $U(u)$ \eqref{monodromy-elements} solving the quantum Yang-Baxter equation \eqref{YB} associated with $R(u)$ \eqref{R-matrix}, satisfy the following relations in $\mathbb{C}[u^{\pm 1},v^{\pm 1}]\otimes\mathbb{A}$:
\begin{subequations}
\begin{align}
\label{Ra} [X(u) ,X(v)]&=0\ \text{for}\ X=A,B,C,D , \\
\label{Rb} s(q^{-1}u/v)A(u) B(v) &= s(q^{-1}) A(v) B(u) + q s(u/v) B(v) A(u) ,\\
\label{Rc} s(q^{-1}u/v)B(u) A(v) &= q^{-1} s(u/v) A(v) B(u) +  s(q^{-1}) B(v) A(u) ,\\
\label{Rd}  C(u) B(v) - tB(v) C(u) &= \frac{1-t}{s(u/v)}\Bigl( A(v) D(u) - A(u) D(v)\Bigr)  ,\\
\label{Re}  t B(u)C(v) - C(v) B(u) & = \frac{1-t}{s(u/v)}\Bigl( D(v) A(u) - D(u) A(v)\Bigr)  , \\
\label{Rf} A(v) D(u) + D(v) A(u) &=  A(u) D(v) + D(u) A(v)  .
\end{align}
\end{subequations}
\end{lemma}
\begin{proof}
The stated relations follow by comparing matrix entries on both sides of the quantum Yang-Baxter equation \eqref{YB}. 
For \eqref{Ra} one considers the matrix entries at the four corners $(1,1)$, $(1,4)$, $(4,1)$ and $(4,4)$.
Eqs. \eqref{Rb}, \eqref{Rc}, \eqref{Rd} and \eqref{Re} correspond to the matrix entries at the positions
$(1,2)$, 
$(1,3)$, 
$(2,2)$ and $(3,3)$, respectively. Finally, Eq. \eqref{Rf} follows from Eq. \eqref{Ra} and the commutativity $[T(u),T(v)]=0$ of the transfer operator
$T(u)=A(u)+D(u)$. 
\end{proof}

\subsection{Relations for the boundary monodromy matrix \cite{skl:boundary}}

\begin{lemma}\label{boundary-monodromy-relations:lem} 
The entries $\mathcal{A}(u)$, $\mathcal{B}(u)$, $\mathcal{C}(u)$ and $\mathcal{D}(u)$ in $\mathbb{C}[u^{\pm 1}]\otimes\mathbb{A}_-\otimes\mathbb{A}$ of a (boundary monodromy) matrix $\mathcal{U}(u)$ \eqref{monodromy-elements} solving the left reflection equation \eqref{RE1} (upon substituting $\mathcal{U}(u)$ for $K_-(u)$) associated with $R(u)$ \eqref{R-matrix}, satisfy the following relations in $\mathbb{C}[u^{\pm 1},v^{\pm 1}]\otimes\mathbb{A}_-\otimes\mathbb{A}$:
\begin{subequations}
\begin{equation}\label{BC-commutativity}
[\mathcal B(u), \mathcal B(v)]=
[\mathcal C(u), \mathcal C(v)]=0 ,
\end{equation}
\begin{align}\label{e:A(u)B(v)oud}
\mathcal A(u) &\mathcal B(v) 
= \\
&\frac{s(qu/v) s(quv)}{s(u/v) s(uv) } \mathcal B(v) \mathcal A(u)
+ \frac{s(q^{-1}) s(quv)}{s(u/v) s(uv)} \mathcal B(u) \mathcal A(v)
+ \frac{s(q)}{s(uv)} \mathcal B(u) \mathcal D(v)  , \nonumber
\end{align}
and
\begin{align} \label{e:D(u)B(v)}
\mathcal D(u) \mathcal B(v) 
&= 
 \frac{s(q^{-1} u/v) s(q^{-1}uv) }{ s(u/v) s(uv)}  \mathcal B(v) \mathcal D(u)
+\frac{s(q) s(q^{-1}uv)}{s(u/v) s(uv) } \mathcal B(u) \mathcal D(v)
\\
&\quad
+\frac{s(q)s(q^{-1}) c(q) }{ s(u/v) s(uv)} \mathcal B(v) \mathcal A(u) 
+ 
\frac{s(q^{-1}) s(q^{-2}u/v)}{s(u/v) s(uv)} \mathcal B(u) \mathcal A(v).\nonumber
\end{align}
\end{subequations}
\end{lemma}
\begin{proof}
The stated relations follow by comparing matrix entries on both sides of the left reflection equation \eqref{RE1}. 
For the relations in Eq. \eqref{BC-commutativity} one
considers the two corner matrix entries (1,4) and (4,1), respectively.
The relation in Eq. \eqref{e:A(u)B(v)oud} is read-off in turn from the matrix entries at the position $(1,2)$. For Eq.  \eqref{e:D(u)B(v)} one considers the matrix entries at position (2,4) to deduce
that
\begin{align*}
\mathcal D(u) \mathcal B(v) 
&= 
 \frac{s(q^{-1} u/v) s(uv) }{ s(u/v) s(q uv)}  \mathcal B(v) \mathcal D(u)
+\frac{s(q) s(uv)}{s(u/v) s(q uv) } \mathcal B(u) \mathcal D(v)
\\
&\quad 
+\frac{s(q)s(q^{-1})  }{ s(u/v) s(quv)}\mathcal A(u)  \mathcal B(v)
+ 
\frac{s(q^{-1}) s(q^{-1}u/v)}{s(u/v) s(quv)}  \mathcal A(v) \mathcal B(u) .\nonumber
\end{align*}
Invoking of Eq. \eqref{e:A(u)B(v)oud} allows to rewrite 
$\mathcal A(u) \mathcal B(v)$ in terms of $ \mathcal B(v) \mathcal A(u)$, $ \mathcal B(u) \mathcal A(v)$ and $ \mathcal B(u) \mathcal D(v)$, and
to rewrite $\mathcal A(v) \mathcal B(u)$ in terms of 
 $ \mathcal B(u) \mathcal A(v)$, $ \mathcal B(v) \mathcal A(u)$ and $ \mathcal B(v) \mathcal D(u)$, respectively. Upon simplification of the coefficients, Eq. \eqref{e:D(u)B(v)} now follows.
\end{proof}

Following Sklyanin \cite{skl:boundary} we put
\begin{equation}\label{Dhat}
\hat{\mathcal  D}(u) := \mathcal D(u) + \frac{s(q)}{s(u^2)} \mathcal A(u) ,
\end{equation}
which permits to rewrite the last two relations of Lemma \ref{boundary-monodromy-relations:lem}
in a somewhat more symmetric form.
\begin{lemma}\label{boundary-monodromy-relations-mod:lem} One has that
\begin{subequations}
\begin{equation} \label{e:A(u)B(v)} 
\mathcal A(u) \mathcal B(v) 
= 
f_1(u,v) \mathcal B(v) \mathcal A(u)
+f_2(u,v) \mathcal B(u) \mathcal A(v)
+f_3(u,v)  \mathcal B(u) \hat{\mathcal D}(v)
\end{equation}
and that
\begin{equation} \label{e:Dhat(u)B(v)} 
\hat{\mathcal D}(u) \mathcal B(v) 
= 
g_1(u,v) \mathcal B(v) \hat {\mathcal D}(u)
+g_2(u,v) \mathcal B(u) \hat{\mathcal D}(v)
+g_3(u,v)  \mathcal B(u) \mathcal A(v) ,
\end{equation}
\end{subequations}
where
\begin{alignat*}{2}
f_1(u,v) &=  \frac{s(quv) s(qu/v)}{ s(u/v) s(uv)},    & \qquad\qquad   g_1(u,v) &=
\frac{   s(q^{-1}u/v)   s(q^{-1}uv)}{   s(u/v) s(uv)},\\
f_2(u,v) &= \frac{s(q^{-1}) s(qv^2) }{s(u/v) s(v^2)},  & \qquad\qquad g_2(u,v) &= 
 \frac{s(q) s(q^{-1}u^2)}{   s(u/v) s(u^2)},\\
f_3(u,v) &=  \frac{ s(q)}{s(uv)} ,  & \qquad\qquad g_3(u,v) &= \frac{ s(q^{-1}) s(q^{-1} u^2) s(qv^2)  }{  s(uv) s(u^2) s(v^2)} .
\end{alignat*}
\end{lemma}
\begin{proof}
Substitution of $\hat{\mathcal  D}(u)$ \eqref{Dhat} into Eq. \eqref{e:A(u)B(v)} readily reproduces
Eq. \eqref{e:A(u)B(v)oud} upon simplification of the coefficients. Similarly,
substitution of $\hat{\mathcal  D}(u)$ \eqref{Dhat} into Eq.
\eqref{e:Dhat(u)B(v)} results in an expression for $ \mathcal{D}(u)\mathcal{B}(v) $
in terms of  
$ \mathcal B(v) {\mathcal D}(u)$, $ \mathcal B(u) {\mathcal D}(v)$,
 $\mathcal B(v) \mathcal A(u) $, $\mathcal B(u) \mathcal A(v) $ and $\mathcal A(u) \mathcal B(v) $.
After rewriting this last term with the aid of Eq. \eqref{e:A(u)B(v)oud}, one reproduces
Eq. \eqref{e:D(u)B(v)} upon simplification of the coefficients.
\end{proof}

\vspace{3ex}
\section*{Acknowledgments.}
We thank S.N.M. Ruijsenaars for helpful comments concerning the quasi-periodicity of the Casoratian in the proof of  Theorem \ref{q-boson-eigenfunctions:thm}.
We are also most grateful to M. Wheeler for explaining us how the proof  of the branching rule  in Eq. \eqref{BC-HL-branching-rule}, satisfied by the hyperoctahedral Hall-Littlewood polynomials with $\hat{a}=0$,
extends in Remark 4 of \cite[Sec. 3.3]{whe-zin:refined} from $a=0$ to $a\neq 0$.

\bibliographystyle{amsplain}

\begin{thebibliography}{000000}



\bibitem[BB]{bog-bul:q-deformed} N.M. Bogoliubov and R.K. Bullough, A q-deformed completely integrable Bose gas model, J. Phys. A {\bf 25} (1992), 4057--4071.

\bibitem[BIK]{bog-ize-kit:correlation} N.M. Bogoliubov, A.G. Izergin, and A.N. Kitanine,
Correlation functions for a strongly correlated boson system,
Nuclear Phys. B {\bf 516} (1998), 501--528.

\bibitem[BCPS]{bor-cor-pet-sas:spectral} A. Borodin, I. Corwin, L. Petrov, and T. Sasamoto,
Spectral theory for the $q$-Boson particle system, Compos. Math. {\bf 151} (2015), 1--67. 

\bibitem[B]{bou:groupes} N. Bourbaki, {\em Groupes et Alg\`ebres de Lie,
Chapitres 4--6}, Hermann, Paris, 1968.

\bibitem[BDM]{bus-die-maz:norm} M.D. Bustamante, J.F. van Diejen, and A.C. de la Maza,
Norm formulae for the Bethe Ansatz on root systems of small rank,
J. Phys. A: Math. Theor. {\bf 41} (2008), 025202.

\bibitem[C]{che:factorizing} I.V. Cherednik, Factorizing particles on a half line and root systems, Theoret. and Math. Phys. {\bf 61} (1984),
977--983.

\bibitem[DF]{das-fok:basic}  G. Dassios and A.S. Fokas, The basic elliptic equations in an equilateral triangle, Proc. R. Soc. Lond. Ser. A Math. Phys. Eng. Sci. {\bf 461} (2005), 2721–--2748.

\bibitem[D1]{die:plancherel} J.F. van Diejen, On the Plancherel formula for the (discrete) Laplacian in a Weyl chamber with repulsive boundary conditions at the walls, Ann. Henri Poincar\'e {\bf 5} (2004), 135--168.

\bibitem[D2]{die:diagonalization} \bysame, Diagonalization of an integrable discretization of the repulsive delta Bose gas on the circle,
Comm. Math. Phys. {\bf 267} (2006), 451--476.

\bibitem[DE1]{die-ems:discrete} J.F. van Diejen and E. Emsiz, Discrete harmonic analysis on a Weyl alcove,
J. Funct. Anal. {\bf 265} (2013), 1981--2038.

\bibitem[DE2]{die-ems:diagonalization} \bysame, Diagonalization of the infinite $q$-boson system, J. Funct. Anal. {\bf 266} (2014), 5801--5817. 

\bibitem[DE3]{die-ems:semi-infinite} \bysame, The semi-infinite $q$-boson system with boundary interaction, Lett. Math. Phys. {\bf 104} (2014), 103--113. 

\bibitem[DE4]{die-ems:branching} \bysame, Branching formula for Macdonald-Koornwinder polynomials, J. Algebra {\bf 444} (2015), 606--614.

\bibitem[Do]{dor:orthogonality} T.C. Dorlas,
Orthogonality and completeness of the Bethe ansatz eigenstates of the
nonlinear Schroedinger model, Commun. Math. Phys. {\bf 154} (1993), 347--376.

\bibitem[E]{ems:completeness} E. Emsiz,
Completeness of the Bethe ansatz on Weyl alcoves,
Lett. Math. Phys. {\bf 91} (2010), 61--70.

\bibitem[EOS1]{ems-opd-sto:periodic} E. Emsiz, E.M.  Opdam, and J.V. Stokman,
Periodic integrable systems with delta-potentials,
Comm. Math. Phys. {\bf 264} (2006), 191--225.

\bibitem[EOS2]{ems-opd-sto:trigonometric} \bysame,
Trigonometric Cherednik algebra at critical level and quantum many-body problems,
Sel. Math. New Ser. {\bf 14} (2009), 571–--605.

\bibitem[F]{fad:how} L.D. Faddeev, How the algebraic Bethe Ansatz works for integrable models. in: {\em Sym\'etries Quantiques} (Les Houches, 1995),  A. Connes, K. Gawedzki, and J. Zinn-Justin (eds.), 
North-Holland Publishing Co., Amsterdam, 1998, 149--219.

\bibitem[FSHY]{fan-shi-hou-yan:integrable} H. Fan, K.-J. Shi, ,B.-Y. Hou, and Z.-X. Yang, Integrable boundary conditions associated with the $Z_n\times Z_n$ Belavin model and solutions of reflection equation,
Internat. J. Modern Phys. A {\bf 12} (1997), 2809--2823. 

\bibitem[G1]{gau:boundary} M. Gaudin,
Boundary energy of a Bose gas in one dimension. Phys. Rev. A {\bf 4} (1971), 386--394.

\bibitem[G2]{gau:bethe} \bysame,
{\em The Bethe Wavefunction}, Cambridge University Press, Cambridge, 2014.

\bibitem[G]{gut:integrable} E. Gutkin, Integrable systems with delta-potential.
Duke Math. J. {\bf 49} (1982), 1--21.

\bibitem[GS]{gut-sut:completely} E. Gutkin and B. Sutherland,
Completely integrable systems and groups generated by reflections. Proc. Natl. Acad. Sci. USA {\bf 76}
(1979), 6057--6059.

\bibitem[HO]{hec-opd:yang} G.J. Heckman and E.M. Opdam,
Yang's system of particles and Hecke algebras. Ann. of Math. (2) {\bf 145} (1997), 139--173.

\bibitem[JM]{jim-miw:algebraic}  M. Jimbo and T. Miwa, {\em Algebraic Analysis of Solvable Lattice Models}, CBMS Regional Conference Series in Mathematics {\bf 85}, American Mathematical Society, Providence, RI, 1995.

\bibitem[KS]{kli-sch:quantum} A. Klimyk and K. Schm\"udgen, {\em Quantum Groups and their Representations}, Springer-Verlag, Berlin, 1997.

\bibitem[KBI]{kor-bog-ize:quantum} V.E. Korepin, N.M. Bogoliubov, and A.G. Izergin,
{\em Quantum Inverse Scattering Method and Correlation Functions},
Cambridge University Press, Cambridge, 1993.

\bibitem[Ko]{kor:cylindric} C. Korff, Cylindric versions of specialised Macdonald functions and a deformed Verlinde algebra, Comm. Math. Phys, {\bf 318} (2013), 173--246.

\bibitem[Ku]{kul:quantum} P.P. Kulish,
Quantum difference nonlinear Schr\"odinger equation,
Lett. Math. Phys. {\bf 5} (1981), 191--197. 

\bibitem[MN]{mez-nep:integrable}  L. Mezincescu and R.I.  Nepomechie, Integrable open spin chains with nonsymmetric $R$-matrices, J. Phys. A {\bf 24} (1991), L17--L23.

\bibitem[LL]{lie-lin:exact} E.H. Lieb and W. Liniger, Exact analysis of an
interacting Bose gas, I. The general solution and the ground
state, Phys. Rev. (2) {\bf 130} (1963), 1605--1616.

\bibitem[LW]{li-wan:exact} B. Li and Y.-S. Wang,
Exact solving $q$ deformed boson model under open boundary condition,
Modern Phys. Lett. B {\bf 26} (2012), 1150008.

\bibitem[M1]{mac:symmetric}  I.G. Macdonald, {\em Symmetric Functions and
Hall Polynomials}, Second Edition, Clarendon Press, Oxford, 1995.

\bibitem[M2]{mac:orthogonal}  \bysame, Orthogonal polynomials associated
with root systems, S\'em. Lothar. Combin. {\bf 45} (2000/01), Art.
B45a.

\bibitem[Ma]{maj:foundations} S. Majid,
{\em Foundations of Quantum Group Theory}, Cambridge University Press, Cambridge, 1995.

\bibitem[ML]{mat:many-body} D.C. Mattis, {\em
The Many-Body Problem: An Encyclopedia of Exactly Solved Models in
One Dimension}, World Scientific, Singapore, 1994.

\bibitem[Mc]{mcc:laplacian} B.J. McCartin,
{\em Laplacian Eigenstructure of the Equilateral Triangle}, Hikari Ltd., Ruse, 2011.

\bibitem[NR]{nel-ram:kostka} K. Nelsen and A. Ram, Kostka-Foulkes polynomials
and Macdonald spherical functions. in: {\em Surveys in
Combinatorics}, C. D. Wensley (ed.), London Math. Soc. Lecture Note
Ser. {\bf 307}, Cambridge University Press, Cambridge, 2003, 325--370.

\bibitem[R]{rui:continuum} S.N.M. Ruijsenaars, The continuum limit of the
infinite isotropic Heisenberg chain in its ground state
representation, J. Funct. Anal. {\bf 39} (1980), 75--84.

\bibitem[SW]{sas-wad:exact} T. Sasamoto and M.  Wadati, Exact results for one-dimensional totally asymmetric diffusion models, J. Phys. A {\bf 31}(1998), 6057--6071.

\bibitem[S]{skl:boundary} E.K. Sklyanin,
Boundary conditions for integrable quantum systems,
J. Phys. A {\bf 21} (1988), 2375–-2389.

\bibitem[T]{tak:integrable}  L.A. Takhtajan, Integrable models in classical and quantum field theory, {\em Proceedings of the International Congress of Mathematicians, Vol. 1, 2} (Warsaw, 1983),
Z. Ciesielski and C. Olech (eds.), North-Holland Publishing Co., Amsterdam, 1984, 1331--1346.

\bibitem[Ts]{tsi:quantum} N.V. Tsilevich, The quantum inverse scattering method for the $q$-boson model and symmetric functions, Funct. Anal. Appl. {\bf 40} (2006), 207--217.

\bibitem[V]{vla:boundary} B. Vlaar, Boundary transfer operators and boundary quantum
KZ equations, J. Math. Phys. {\bf 56} (2015), 071705.

\bibitem[WZ]{whe-zin:refined} M.
Wheeler and P. Zinn-Justin, Refined Cauchy/Littlewood identities and six-vertex model partition functions: III. Deformed bosons, arXiv:1508.02236.

\bibitem[YY]{yan-yan:thermodynamics} C.N. Yang and C.P. Yang, 
Thermodynamics of a one-dimensional system of bosons with repulsive delta-function interaction,
J. Mathematical Phys. {\bf 10} (1969), 1115--1122. 


\end{thebibliography}

\end{document}